\documentclass[letterpaper, 11pt,  reqno]{amsart}
\usepackage{graphicx} 
\usepackage{amsmath,amssymb,amscd,amsthm,amsxtra,esint}
\usepackage{color}
\usepackage[margin=1.1in,marginparwidth=1.5cm, marginparsep=0.5cm]{geometry} 
\usepackage{tikz-cd}
\tikzset{smalltext/.style={"\textup{\small #1}" description}}

\usepackage[markup=default]{changes}
\definechangesauthor[name={Reviewer 1}, color=orange]{R1} 
\definechangesauthor[name={Reviewer 2}, color=blue]{R2} 
\definechangesauthor[name={Reviewer 3}, color=green]{R3}

\usepackage[implicit=true]{hyperref}
\usepackage{comment}

\allowdisplaybreaks[4]

\definecolor{gr}{rgb}   {0.,   0.69,   0.23 }
\definecolor{bl}{rgb}   {0.,   0.5,   1. }
\definecolor{mg}{rgb}   {0.85,  0.,    0.85}
\definecolor{yl}{rgb}   {0.8,  0.7,   0.}
\definecolor{or}{rgb}  {0.7,0.2,0.2}

\newcommand{\noi}{\noindent}
\newcommand{\R}{\mathbb{R}}
\newcommand{\T}{\mathbb{T}}
\newcommand{\Z}{\mathbb{Z}}
\newcommand{\N}{\mathbb{N}}

\newcommand{\F}{\mathcal{F}}

\newcommand{\Id}{\textup{Id}}

\newcommand{\al}{\alpha}
\newcommand{\be}{\beta}
\newcommand{\ta}{\theta}
\newcommand{\s}{\sigma}
\newcommand{\eps}{\varepsilon}

\newcommand{\Dl}{\Delta}
\newcommand{\dl}{\delta}
\newcommand{\ld}{\lambda}

\newcommand{\dt}{\partial_t}
\newcommand{\dx}{\partial_x}
\renewcommand{\dh}{\partial_h}
\newcommand{\ind}{\mathbf 1}
\newcommand{\ft}{\widehat}
\newcommand{\wt}{\widetilde}
\newcommand{\les}{\lesssim}
\newcommand{\ges}{\gtrsim}
\newcommand{\cj}{\overline}

\newcommand{\jb}[1]{\langle #1 \rangle}

\newcommand{\deff}{\stackrel{\textup{def}}{=}}

\newtheorem{theorem}{Theorem}[section]
\newtheorem{lemma}[theorem]{Lemma}
\newtheorem{proposition}[theorem]{Proposition}

\newtheorem{remark}[theorem]{Remark}

\newtheorem*{ackno}{Acknowledgements}

\numberwithin{equation}{section}
\numberwithin{theorem}{section}

\makeatletter
\@namedef{subjclassname@2020}{%
  \textup{2020} Mathematics Subject Classification}
\makeatother

\title[KdV limit of Toda lattice]{The Korteweg-de Vries limit for the global dynamics of the Toda lattice}

\author[H.~Koch and R.~Liu]{Herbert Koch and Ruoyuan Liu}

\address{
Herbert Koch, Mathematical Institute\\
University of Bonn\\
Endenicher Allee 60\\
53115\\
Bonn\\
Germany}

\email{koch@math.uni-bonn.de}

\address{
Ruoyuan Liu, Mathematical Institute\\
University of Bonn\\
Endenicher Allee 60\\
53115\\
Bonn\\
Germany}

\email{ruoyuanl@math.uni-bonn.de}


\begin{document}

\baselineskip = 14pt


\begin{abstract} 
It has been observed that the dynamics of the Toda lattice can be well described by solutions of the Korteweg-de Vries (KdV) equation in the continuum limit. We show that, under the KdV scaling and a suitable translation, the solution of the Toda lattice with $H^1$ initial data converges to that of the KdV equation globally in time. Our proof relies on tools from harmonic analysis and also on the construction and the conservation of mass and energy of the Toda lattice, the latter of which are derived from the completely integrable structure of the Toda lattice. As a consequence, we obtain long-wave KdV limits for the Toda lattice.
\end{abstract}


\maketitle

\tableofcontents

\section{Introduction}

\subsection{Toda lattice}

The Toda lattice is a Hamiltonian lattice model describing the motion of a chain of particles with nearest neighbor interaction. Discovered by Toda in \cite{Toda, Toda2}, the Toda lattice is one of the earliest examples of a nonlinear completely integrable system. It is in fact a particularly interesting case of the more general Fermi-Pasta-Ulam (FPU) system introduced by Fermi, Pasta, and Ulam, with contributions by Tsingou, in \cite{FPUT}. More specifically, the FPU Hamiltonian is given by
\begin{align}
    H (q, p) = \sum_{n \in \Z} \Big( \frac{p (n)^2}{2} + V (q (n + 1) - q (n)) \Big) ,
\label{defH}
\end{align}

\noi
where $q : \Z \to \R$ denotes the position of the $n$th string, $p : \Z \to \R$ denotes the momentum of the $n$th particle, and $V : \R \to \R$ is the potential energy from nearest neighbor interactions. The Hamiltonian $H (q, p)$ in \eqref{defH} generates the following FPU system:
\begin{align*}
\begin{cases}
    \dt q (t, n) = \frac{\partial H (q, p)}{\partial p} (t, n) = p (t, n) \vspace{5pt} \\
    \dt p (t, n) = - \frac{\partial H (q, p)}{\partial q} (t, n) = V' (q (t, n + 1) - q (t, n)) - V' (q (t, n) - q (t, n - 1)).
\end{cases}
\end{align*}

\noi
The Toda lattice is the case when the potential function $V: \R \to \R$ is given by (after normalization)
\begin{align*}
V (r) = e^{-r} + r - 1 ,
\end{align*}

\noi
which gives rise to the system
\begin{align}
\begin{cases}
    \dt q (t, n) = p (t, n) \vspace{5pt} \\
    \dt p (t, n) = e^{- (q (t, n) - q (t, n - 1))} - e^{- (q (t, n + 1) - q (t, n))}.
\end{cases}
\label{toda_qp}
\end{align}

\noi
Over the last few decades, there has been extensive studies in various aspects of the FPU system. See \cite{BI, CGG, Gal} and the references therein.

It is now well known that the dynamics of the FPU system can be well described in the continuum limit by solutions of the Korteweg-de Vries (KdV) equation:
\begin{align}
    \dt u + \partial_x^3 u + 6 u \partial_x u = 0.
\label{KdV}
\end{align}

\noi
Such connection of the two dynamics was first discovered by Zabusky and Kruskal in \cite{ZK} via formal computations. Since then, there has been some research work on the convergence of solitary waves for the FPU system to the soliton solutions for the KdV equation; see \cite{Eil, FW, FP1, FP2, FP3, FP4, Mizu}. On the other hand, for more general solutions, Schneider and Wayne \cite{SW} proved that the FPU dynamics can be approximated by counter-propagating KdV flows using a dynamical system approach; see \cite{BP, PB} for research on such phenomenon in the periodic setting, \cite{Mie, CCPS, GMWZ} for generalized discrete models, and also \cite{GM04, GM06, Sch10} for the connection of FPU to the cubic nonlinear Schr\"odinger equation. The result in \cite{SW} was then improved by \cite{HKY}, where Hong, Kwak, and Yang simplified the assumptions on initial data and reduced the regularity requirement. In particular, they viewed the FPU system as a nonlinear dispersive equation and established smoothing properties to justify the KdV limit of the FPU system. We point out that the convergence result in \cite{HKY} is only valid for a short time interval.

Compared to the general FPU system, the Toda lattice has the advantage that it is a completely integrable system. See \cite{Gie1, Tes, Gie2, BKP1, BKP2, BKP3} on research that exploited the integrable structure of the Toda lattice. In particular, the Toda lattice admits infinitely many conservation laws, which are useful not only for showing the existence of a global-in-time solution but also for proving its continuum limit to KdV globally in time. The main goal of this paper is to show the latter convergence result. Specifically, we show that under suitable initial conditions, if we perform the KdV scaling and a suitable translation on the Toda lattice, then its solution converges to that of the KdV equation globally in time.
However, we choose not to directly work with the form \eqref{toda_qp}, which is not convenient for us to write out the relevant conservation laws of the Toda lattice. See the next subsection for more details on the setup for the problem.

\subsection{Setup and main results}

By introducing Flaschka's variables \cite{Fla}
\begin{align*}
    \al (t, n) = \exp \Big( - \frac{q (t, n + 1) - q (t, n)}{2} \Big) - 1 \quad \text{and} \quad \be (t, n) = \frac 12 p (t, n) ,
\end{align*}

\noi
we obtain the following equivalent form of the Toda lattice:
\begin{align}
\begin{cases}
    \dt \al (t, n) = - ( 1 + \al (t, n) ) ( \be (t, n + 1) - \be (t, n) ) \vspace{5pt} \\
    \dt \be (t, n) = - \big(1 + \frac 12 \al (t, n) + \frac 12 \al (t, n - 1) \big) (\al (t, n) - \al (t, n - 1)) .
\end{cases}
\label{toda_ab}
\end{align}

\noi
We refer to \eqref{toda_ab} as the {\it Flaschka's form} of the Toda lattice.

Let us now consider the continuum limit of the Toda lattice in Flaschka's form. Given a small parameter $h > 0$, we define the scaled variables $\al^h, \be^h : \R \times h \Z \to \R$ by performing the following scaling compatible with the KdV equation (with an extra factor ``3'' in front of the time variable which will make our analysis more convenient):
\begin{align}
\begin{split}
    \al^h (t, \ld) &:= h^{-2} \al (3 h^{-3} t, h^{-1} \ld) , \\
    \be^h (t, \ld) &:= h^{-2} \be (3 h^{-3} t, h^{-1} \ld)
\end{split}
\label{def_abh}
\end{align}

\noi
for any $\ld \in h \Z$. Then, the discrete system \eqref{toda_ab} is equivalent to the following scaled version of \eqref{toda_ab}:
\begin{align}
\begin{cases}
    \dt \al^h (t, \ld) = - 3 (h^{-3} + h^{-1} \al^h (t, \ld)) (\be^h (t, \ld + h) - \be^h (t, \ld)) \vspace{5pt} \\
    \dt \be^h (t, \ld) = - 3 \big( h^{-3} + \frac 12 h^{-1} \al^h (t, \ld) + \frac 12 h^{-1} \al^h (t, \ld - h) \big) (\al^h (t, \ld) - \al^h (t, \ld - h)) .
\end{cases}
\label{toda_abh}
\end{align}

\noi
The linear part of the system \eqref{toda_abh} suggests that we may define the variable $\gamma^h : \R \times h \Z \to \R$ as
\begin{align}
\gamma^h (t, \ld) := 
\begin{cases}
    \al^h ( t, \frac{\ld}{2} ) & \text{if $\frac{\ld}{h}$ is even} \vspace{5pt} \\
    \be^h ( t, \frac{\ld + h}{2} ) & \text{if $\frac{\ld}{h}$ is odd}
\end{cases}
\label{defgh}
\end{align}

\noi
given $t \in \R$ and $\ld \in h \Z$.
Then, the scaled Toda lattice \eqref{toda_abh} is equivalent to 
\begin{align}
    \dt \gamma^h (t, \ld) = - 3 h^{-3} ( \gamma^h (t, \ld + h) - \gamma^h (t, \ld - h) ) - \mathcal{N}^h [\gamma^h] (t, \ld) ,
\label{todah}
\end{align}

\noi
where the scaled nonlinearity $\mathcal{N}^h [\gamma^h]$ is given by
\begin{align}
&\mathcal{N}^h [\gamma^h] (t, \ld) \notag \\
&=
\begin{cases}
    3 h^{-1} \gamma^h (t, \ld) ( \gamma^h (t, \ld + h) - \gamma^h (t, \ld - h) ) & \text{if $\dfrac{\ld}{h}$ is even} \vspace{5pt} \\
    \dfrac 32 h^{-1} (\gamma^h (t, \ld + h) + \gamma^h (t, \ld - h)) (\gamma^h (t, \ld + h) - \gamma^h (t, \ld - h)) & \text{if $\dfrac{\ld}{h}$ is odd} .
\end{cases}  
\label{defNh}
\end{align}



For any function $f^h$ defined on $h \Z$, we extend it to the continuum $\R$ via
\begin{align}
    \mathcal{E} f (x) := \frac{1}{2 \pi} \int_{- \frac{\pi}{h}}^{\frac{\pi}{h}} (\F^h f) (\xi) e^{i \xi x} d \xi ,
\label{extend}
\end{align}

\noi
where $\F^h$ denotes the discrete Fourier transform defined in \eqref{defFh} below. This extension is inspired by the work \cite{KOVW1, KOVW} by Killip, Ouyang, Vi\c{s}an, and Wu on the continuum limits of the Ablowitz-Ladik system. Using \eqref{extend}, we can extend $\gamma^h$ defined in \eqref{defgh} to be a function $\mathcal{E} \gamma^h$ defined on $\R$. Similar to \cite{KOVW} on the modified KdV limit of the Ablowitz-Ladik system, in order to obtain the KdV equation in the continuum limit, we remove the leading order translation associated with the scaled Toda lattice \eqref{todah} by defining
\begin{align}
    u^h (t, x) := e^{6 h^{-2} t \dx} \mathcal{E} \gamma^h (t, x) = \mathcal{E} \gamma^h (t, x + 6 h^{-2} t)
\label{defuh}
\end{align}

\noi
given $t \in \R$ and $x \in \R$. This is required due to the fact that the symbol for the linear flow of \eqref{todah} is $e^{- 6 i h^{-3} t \sin (h \xi)}$ and we have the expansion $h^{-3} \sin (h \xi) = h^{-2} \xi - \frac{\xi^3}{6} + O(h^2 \xi^5)$, so that the translation removes the diverging $h^{-2} \xi$ term. The equation for $u^h$ then exhibits cubic dispersion corresponding to the Airy flow, which, combined with the quadratic nonlinearity in \eqref{defNh}, provides the heuristic of the continuum limit of the Toda lattice to the KdV equation. 

Before moving on, we mention that well-posedness problem for the KdV equation \eqref{KdV} has been studied extensively; see \cite{BS, KPV, Bour93k, KPV96, CCT, CKSTT03, Guo, Kish, KV} and the references therein. In particular, the unconditional uniqueness of the KdV equation in $L^2 (\R)$ (and so in $H^1 (\R)$) was established in \cite{Zhou}, which we will use in stating our main theorem below.

We are now ready to state the first main result in this paper. For the definitions of the function spaces on $h \Z$, see Subsection~\ref{SUB:space} below.

\begin{theorem}[Continuum limit of Toda to KdV: Flaschka's form]
\label{THM:conv}   
Given $0 < h \leq 1$, let $u_0 \in H^1 (\R)$ and $\gamma_0^h \in H^1 (h \Z)$. Suppose that there exist constants $C_1, C_2 > 0$ independent of $h$ such that 
\begin{align}
    \| \gamma_0^h \|_{H^1 (h \Z)} + \| u_0 \|_{H^1 (\R)} \leq C_1 \quad \text{and} \quad \| \mathcal{E} \gamma_0^h - u_0 \|_{L^2 (\R)} \leq C_2 h,
\label{u0cond}
\end{align}

\noi
where $\mathcal{E}$ is the extension operator defined in \eqref{extend}. We also suppose that
\begin{align}
    h^{\frac 32} \| \gamma_0^h \|_{L^2 (h \Z)} \leq \frac 14 .
\label{h_cond}
\end{align}

\noi
Then, the following statements hold.

\smallskip \noi
\textup{(i)} There exists a unique solution $\gamma^h \in C(\R; H^1 (h \Z))$ to the scaled Toda lattice \eqref{todah} with initial data $\gamma^h |_{t = 0} = \gamma_0^h$.

\smallskip \noi
\textup{(ii)} Let $u \in C (\R; H^1 (\R))$ be the unique global-in-time solution to the KdV equation \eqref{KdV} with initial data $u|_{t = 0} = u_0$.
Let $u^h$ be defined in \eqref{defuh} with $\gamma^h$ being given by part \textup{(i)}. Then, for any $T > 0$, there exist some constants $C, C' > 0$ depending on $C_1$ and $C_2$ such that
\begin{align*}
    \| u^h - u \|_{C ([-T, T] ; L^2 (\R))} \leq C  e^{C' T} h^{\frac 25}.
\end{align*}
\end{theorem}

Our assumption \eqref{u0cond} on the initial data is rather general. In particular, given $0 < h \leq 1$ and $\gamma_0^h \in H^1 (h \Z)$, one may take its extension on $\R$ and set $u_0 = \mathcal{E} \gamma_0^h$ as the initial data for the KdV equation. As another example, one may proceed as in \cite{KOVW} by choosing
\begin{align}
    \gamma_0^h (\ld) = P_{\frac{\pi}{h}} u_0 (\ld)
\label{g0h_ex}
\end{align}

\noi
where $P_{\frac{\pi}{h}}$ is the frequency cutoff onto frequencies $\{ \xi \in \R: |\xi| \leq \frac{\pi}{h} \}$. In both settings, one can easily check that $\gamma_0^h$ and $u_0$ fulfill the assumption \eqref{u0cond} (thanks to Lemma~\ref{LEM:Hsbdd} below). Moreover, due to the $L^2$ difference bound in \eqref{u0cond} (and Lemma~\ref{LEM:Hsbdd}), the assumption \eqref{h_cond} can always be satisfied with a sufficiently small $h > 0$.

Our Theorem~\ref{THM:conv} shows that, for any $T > 0$, $u^h$ converges to $u$ in $C ([-T, T]; L^2 (\R))$ with rate $h^{\frac 25}$ as $h \to 0$. In particular, we obtain global-in-time convergence of the dynamics as compared to the local-in-time convergence in \cite{HKY} (for the FPU system). The convergence rate $h^{\frac 25}$ matches the rate obtained in \cite{HKY}. In fact, this rate is optimal in view of the convergence of the linear flows; see Lemma~\ref{LEM:conv0} below.

Note that the solution of the form \eqref{defgh} for the Toda lattice is considered for the first time in this work. It turns out that this form is convenient not only for showing its continuum limit to the KdV equation, but also it allows us to build the discrete analogues of mass and energy; see Subsection~\ref{SUB:cons} and Remark~\ref{RMK:cons} (see also Remark~\ref{RMK:toda} below for a comment on higher order conserved quantities). The discrete mass and energy play a central role in obtaining uniform-in-$h$ $H^1$ a priori bounds for the solution to the scaled Toda lattice \eqref{todah}. It is here that we need the technical condition \eqref{h_cond} as stated in Theorem~\ref{THM:conv}, which is guaranteed by \eqref{u0cond} (and Lemma~\ref{LEM:Hsbdd}) with $h = h (C_2, \| u_0 \|_{L^2 (\R)}) > 0$ sufficiently small and is used to deal with the logarithmic terms in the discrete mass and energy; see Subsection~\ref{SUB:H1}. The a priori bounds then allow us to show global-in-time convergence of the dynamics in Subsection~\ref{SUB:bdd} and Subsection~\ref{SUB:dyn}. 

Our method for establishing the convergence is based on linear estimates from Fourier analytic tools and resembles that of \cite{KOVW} on the modified KdV limit of the Ablowitz-Ladik system. Nevertheless, we are able to show an explicit convergence rate by using the idea of local smoothing estimate for the discrete system established in \cite{HKY}. 
However, instead of the extension \eqref{extend} via the inverse Fourier transform, the authors in \cite{HKY} used a linear interpolation procedure; see also \cite{Lad, SSB, KLS, HY1, HY} on the work on other discrete models based on the linear interpolation procedure. We do not claim that the procedure based on the Fourier extension \eqref{extend} is better than the linear interpolation procedure, but we do find the former convenient for directly invoking many Fourier analytic tools on the Euclidean space. Moreover, we find an interesting and convenient identity on the Fourier extension of a product of two discrete functions; see Lemma~\ref{LEM:Eprod} below. See also Remark~\ref{RMK:HKY} below for a further comparison of this paper with \cite{HKY}.

\medskip
As a consequence of the continuum limit in Theorem~\ref{THM:conv}, we obtain a long-wave limit result for the Toda lattice \eqref{todah}. 

\begin{theorem}[Long-wave limit of Toda to KdV: Flaschka's form]
\label{THM:lw}
Let $0 < h \leq 1$, $u_0 \in H^1 (\R)$, and $\gamma_0^h \in H^1 (h \Z)$ satisfying the assumptions \eqref{u0cond} and \eqref{h_cond} with constants $C_1 , C_2 > 0$. 
Let $\gamma^h \in C(\R; H^1 (h \Z))$ be the unique global-in-time solution to the scaled Toda lattice \eqref{todah} with initial data $\gamma^h |_{t = 0} = \gamma_0^h$, which is guaranteed by Theorem~\ref{THM:conv}~\textup{(i)}. Let $u \in C (\R; H^1 (\R))$ be the unique global-in-time solution to the KdV equation \eqref{KdV} with initial data $u|_{t = 0} = u_0$. Then, for any $T > 0$, there exist some constants $C, C' > 0$ depending on $C_1$, $C_2$, and $\| u_0 \|_{H^1 (\R)}$ such that
\begin{align*}
    \big\| \gamma^h (t, \ld) - u (t, \ld - 6 h^{-2} t) \big\|_{C ([-T, T] ; L^2 ( h \Z ))} \leq C e^{C' T} h^{\frac 25} .
\end{align*}
\end{theorem}

We present the proof of Theorem~\ref{THM:lw} in Subsection~\ref{SUB:lw1}. By defining
\begin{align*}
    \gamma (t, n) &:= h^2 \gamma^h (\tfrac{h^3}{3} t, h n), 
\end{align*}

\noi
we scale back to the usual Toda lattice on $\R \times \Z$:
\begin{align}
    \dt \gamma (t, n) = - ( \gamma (t, n + 1) - \gamma (t, n - 1) ) - \mathcal{N} [\gamma] (t, n) ,
\label{todan}
\end{align}

\noi
where the nonlinearity $\mathcal{N} [\gamma]$ is given by
\begin{align}
    \mathcal{N} [\gamma] (t, n) = 
    \begin{cases}
        \gamma (t, n) (\gamma (t, n + 1) - \gamma (t, n - 1)) & \text{if } n \text{ is even} \vspace{5pt} \\
        \frac 12 (\gamma (t, n + 1) + \gamma (t, n - 1)) (\gamma (t, n + 1) - \gamma (t, n - 1))  & \text{if } n \text{ is odd}.
    \end{cases}
\label{defNN}
\end{align}

\noi
Then, from Theorem~\ref{THM:lw}, we obtain the following estimate for $\gamma$:
\begin{align*}
    \big\| \gamma (t, n) - h^2 u (\tfrac{h^3}{3} t, hn - 2 h t) \big\|_{C ([-T, T] ; \ell^2 ( \Z ))} \leq C e^{\frac 13C' h^3 T} h^{\frac 32 + \frac 25} .
\end{align*}

\begin{remark} \rm
\label{RMK:gm-}
We also consider the long-wave limit on the other direction of propagation. Let us define $(\gamma^h)^- : \R \times h \Z \to \R$ by
\begin{align}
(\gamma^h)^- (t, \ld) := 
\begin{cases}
    \al^h ( t, \frac{\ld}{2} ) & \text{if $\frac{\ld}{h}$ is even} \vspace{5pt} \\
    - \be^h ( t, \frac{\ld + h}{2} ) & \text{if $\frac{\ld}{h}$ is odd} 
\end{cases}
\label{defgh-}
\end{align}

\noi
given $t \in \R$ and $\ld \in h \Z$. Then, we see that $(\gamma^h)^-$ satisfies
\begin{align}
    \dt (\gamma^h)^- (t, \ld) = 3 h^{-3} \big( (\gamma^h)^- (t, \ld + h) - (\gamma^h)^- (t, \ld - h) \big) + \mathcal{N}^h [(\gamma^h)^-] (t, \ld),
\label{todah-}
\end{align}

\noi
where the nonlinearity $\mathcal{N}^h$ is defined in \eqref{defNh}. Under the same assumptions of initial data in Theorem~\ref{THM:conv}, we know from minor modifications of Theorem~\ref{THM:conv} that $(u^h)^- (t, x) := \mathcal{E} (\gamma^h)^- (t, x - 6 h^{-2}t)$ converges to $u^-$ satisfying
\begin{align}
    \dt u^- - \dx^3 u^- - 6 u^- \dx u^- = 0 . 
\label{KdV-}
\end{align}

\noi
Then, similar to Theorem~\ref{THM:lw}, for any $T > 0$, there exist some constants $C, C' > 0$ such that
\begin{align*}
    \big\| (\gamma^h)^- (t, \ld) - u^- (t, \ld + 6 h^{-2} t) \big\|_{C ([-T, T] ; L^2 ( h \Z ))} \leq C e^{C' T} h^{\frac 25} .
\end{align*}

\noi
Moreover, by scaling back to $\R \times \Z$, we see that
\begin{align*}
    \gamma^- (t, n) := h^2 (\gamma^h)^- ( \tfrac{h^3}{3} t, h n ) 
\end{align*}

\noi
satisfies
\begin{align*}
    \dt \gamma^- (t, n) = \gamma^- (t, n + 1) - \gamma^- (t, n - 1) + \mathcal{N} [\gamma^-] (t, n)
\end{align*}

\noi
with the nonlinearity $\mathcal{N}[\gamma^-]$ defined in \eqref{defNN}, and that we have the following estimate for $\gamma^-$:
\begin{align*}
    \big\| \gamma^- (t, n) - h^2 u^- (\tfrac{h^3}{3} t, hn + 2 h t) \big\|_{C ([-T, T] ; \ell^2 ( \Z ))} \leq C e^{\frac 13C' h^3 T} h^{\frac 32 + \frac 25} .
\end{align*}

We also point out that the Toda lattice \eqref{todan} is invariant under the symmetry $\gamma (t, n) \mapsto (-1)^{n + 1} \gamma (t, - n)$. Namely, if $\gamma$ is a solution to \eqref{todan}, then $S \gamma (t, n) := (-1)^{n+1} \gamma (t, - n)$ is also a solution to \eqref{todan}. It may be possible to use this invariance to investigate waves going to both directions.
\end{remark}

\medskip
We now turn our attention back to the original form of the Toda lattice \eqref{toda_qp} and consider its continuum limits. For this purpose, we define
\begin{align*}
    r (t, n) = - \frac 12 (q (t, n + 1) - q (t, n)),
\end{align*}

\noi
so that \eqref{toda_qp} is equivalent to
\begin{align}
    \dt^2 r (t, n) = \frac 12 \Dl_1 (e^{2r}) (t, n) ,
\label{eqr0}
\end{align}

\noi
where $\Dl_1 f \deff f(\cdot + 1) + f (\cdot - 1) - 2 f$. We then subtract the linear term on the right-hand-side of \eqref{eqr0} to obtain the following discrete nonlinear wave equation: 
\begin{align}
    \dt^2 r (t, n) - \Dl_1 r (t, n) = \frac 12 \Dl_1 (e^{2r} - 2r) (t, n)
\label{eqr1}
\end{align}

\noi
for any $n \in \Z$. Similar to \eqref{defgh}, we define the scaled version of $r$ by performing the KdV scaling:
\begin{align*}
    r^h (t, \ld) := h^{-2} r (3 h^{-3} t, h^{-1} \ld),
\end{align*}

\noi
which then satisfies the following scaled version of \eqref{eqr1}:
\begin{align}
    \dt^2 r^h (t, \ld) - 9 h^{-4} \Dl_h r^h (t, \ld) = \frac 92 h^{-6} \Dl_h (e^{2 h^2 r^h} - 2 h^2 r^h) (t, \ld)
\label{eqr2}
\end{align}

\noi
for any $\ld \in h \Z$, where $\Dl_h$ is a discrete Laplacian on $h \Z$ defined by
\begin{align*}
    \Dl_h f^h \deff h^{-2} \big( f^h (\cdot + h) + f^h (\cdot - h) - 2 f^h (\cdot) \big) .
\end{align*}



We now state the following theorem regarding the long-wave limit of $r^h$. Below, we denote $\dh^+$ by $\dh^+ f^h (\ld) = h^{-1} (f^h (\ld + h) - f^h (\ld))$.

\begin{theorem}[Long-wave limit of Toda to KdV: original form]
\label{THM:long-wave}
Let $u_0 \in H^1 (\R)$ and $u_0^h \in H^1 (h \Z)$ with $0 < h \leq 1$. Suppose that for any $0 < h \leq 1$, there exist constants $C_1, C_2 > 0$ such that 
\begin{align}
    \| u_0^h \|_{H^1 (h \Z)} \leq C_1 \quad \text{and} \quad \| \mathcal{E} u_0^h - u_0 \|_{L^2 (\R)} \leq C_2 h,
\label{u0cond2}
\end{align}

\noi
where $\mathcal{E}$ is the extension operator defined in \eqref{extend}. We define functions $r_{\textup{even}}^h$ and $r_{\textup{odd}}^h$ on $h \Z$ by 
\begin{align}
    r_{\textup{even}}^h (\ld) := u_0^h (2 \ld) \quad \text{and} \quad r_{\textup{odd}}^h (\ld) := u_0^h (2 \ld - h)
\label{r0cond}
\end{align}

\noi
for any $\ld \in h \Z$. Then, there exists $h_0 = h_0 ( C_2, \| u_0 \|_{L^2 (\R)} ) > 0$ sufficiently small such that for any $0 < h \leq h_0$, the following statements hold. 

\smallskip \noi
\textup{(i)} There exists a unique solution $r^h \in C (\R; H^1 (h\Z))$ to the scaled Toda lattice \eqref{eqr2} with initial data $(r^h, \dt r^h) |_{t = 0} = (r_{\textup{even}}^h, - 3 h^{- 2} \dh^+ r_{\textup{odd}}^h)$. Moreover, there exists a constant $\wt C > 0$ depending on $C_1$, $C_2$, and $\| u_0 \|_{L^2 (\R)}$ but independent of $h$ such that
\begin{align*}
    \| r^h \|_{L^\infty (\R ; H^1 ( h \Z ))} \leq \wt C.
\end{align*}

\smallskip \noi
\textup{(ii)} Let $u \in C (\R; H^1 (\R))$ be the unique global-in-time solution to the KdV equation \eqref{KdV} with initial data $u_0$. Then, for any $T \geq 1$, there exist some constants $C, C' > 0$ depending on $C_1$, $C_2$, and $\| u_0 \|_{H^1 (\R)}$ but independent of $h$ such that
\begin{align*}
    \big\| r^h (t, \ld) - u (t, 2 \ld - 6 h^{-2} t) \big\|_{C ([-T, T] ; L^2 ( h \Z ))} \leq C e^{C' T} h^{\frac 25} .
\end{align*}
\end{theorem}

We present the proof of Theorem~\ref{THM:long-wave} in Subsection~\ref{SUB:lw2}. In fact, we will reduce the problem to the setting of Theorem~\ref{THM:lw} via the following relation:
\begin{align}
\begin{split}
    \al^h (t, \ld) &= h^{-2} \exp \big( h^2 r^h (t, \ld) \big) - h^{-2}, \\
    \be^h (t, \ld) &= - \frac 13 h^2 (\dh^+)^{-1} \dt r^h (t, \ld).
\end{split}
\label{abr}
\end{align}

\noi
Once again, by scaling back to the model \eqref{eqr1} on $\R \times \Z$ via
\begin{align*}
    r (t, n) := h^2 r^h (\tfrac{h^3}{3} t, hn),
\end{align*}

\noi
we get from Theorem~\ref{THM:long-wave} that
\begin{align*}
    \big\| r (t, n) - h^2 u (\tfrac{h^3}{3} t, 2 hn - 2 h t) \big\|_{C ([-T, T] ; \ell^2 ( \Z ))} \leq C e^{\frac 13 C' h^3 T} h^{\frac 32 + \frac 25} .
\end{align*}

\begin{remark} \rm
Similar to Remark~\ref{RMK:gm-}, we also obtain the long-wave limit for the other direction of propagation. Let $u_0$, $r_{\textup{even}}^h$, and $r_{\textup{odd}}^h$ be as given in Theorem~\ref{THM:long-wave}. Then, similar to Theorem~\ref{THM:long-wave}, there exists a unique solution $(r^h)^- \in C (\R; H^1 (h\Z))$ to the scaled Toda lattice \eqref{eqr2} with initial data $((r^h)^-, \dt (r^h)^-) |_{t = 0} = (r_{\textup{even}}^h, 3 h^{- 2} \dh^+ r_{\textup{odd}}^h)$. If we let $u^- \in C (\R; H^1 (\R))$ be the unique global-in-time solution to the KdV equation \eqref{KdV-}, then for any $T \geq 1$, there exist some constants $C, C' > 0$ such that
\begin{align*}
    \big\| (r^h)^- (t, \ld) - u^- (t, 2 \ld + 6 h^{-2} t) \big\|_{C ([-T, T] ; L^2 ( h \Z ))} \leq C e^{C' T} h^{\frac 25} .
\end{align*}

\noi
Moreover, by scaling back to $\R \times \Z$, we see that 
\begin{align*}
    r^- (t, n) : = h^2 (r^h)^- (\tfrac{h^3}{3} t, h n) 
\end{align*}

\noi
satisfies \eqref{eqr1} with scaled initial data, and that we have the following estimate:
\begin{align*}
    \big\| r^- (t, n) - h^2 u^- (\tfrac{h^3}{3} t, 2 hn + 2 h t) \big\|_{C ([-T, T] ; \ell^2 ( \Z ))} \leq C e^{\frac 13 C' h^3 T} h^{\frac 32 + \frac 25} .
\end{align*}
\end{remark}

\medskip
We close the introduction by stating several remarks.

\begin{remark} \rm
\label{RMK:HKY}
As mentioned earlier, our approximation procedure from Toda to KdV is different from that of FPU to KdV by Hong-Kwak-Yang in \cite{HKY}. Moreover, in this paper we choose to focus on Flaschka's form of the Toda lattice \eqref{todah} in order to get access to the $H^1$ conservation law, whereas the authors in \cite{HKY} mainly performed the analysis on the form \eqref{eqr2}. The main difference is manifested in the linear flows of the discrete systems. In our setting, we have mentioned above that the symbol for the linear flow of \eqref{todah} is $e^{- 6 i h^{-3} t \sin (h \xi)}$. In \cite{HKY}, the authors split the discrete linear wave flow (i.e. the linear part of \eqref{eqr2}) into two linear flows with symbols $e^{\pm 12 i h^{-3} t \sin (\frac{h \xi}{2})}$ (after rescaling). 
While our procedure appears to be more transparent, it creates additional difficulty in the analysis around the inflection points $\xi = \pm \frac{\pi}{h}$. This is relevant in Subsection~\ref{SUB:lin}, where we address the issue via some technical treatment.


Nevertheless, a slight variant of Theorem~\ref{THM:long-wave}~(i) allows us to extend the local-in-time long-wave limit in the setting of \cite[Theorem~1.2]{HKY} (with $s = 1$) to a global-in-time long-wave limit in the case of the Toda lattice. Indeed, all we need for the global-in-time limit is a uniform-in-$h$ $H^1$ a priori bound for the solution $r^h$ to the equation \eqref{eqr2}. This is essentially achieved in the proof of Theorem~\ref{THM:long-wave}~(i) in Subsection~\ref{SUB:lw2}, where we use Flaschka's variables to transform the model \eqref{eqr2} into \eqref{toda_abh} so that we can use the $H^1$ conservation law for the Toda lattice. In view of the initial condition in \cite[Theorem~1.2]{HKY}, the use of the conservation laws is slightly different, and we need the following additional assumption on the initial data $(r_{\textup{even}}^h, r_{\textup{odd}}^h)$ for \eqref{eqr2}:
\begin{align*}
    h \sum_{\ld \in h \Z} \big( | r_{\textup{even}}^h (\ld) - r_{\textup{odd}}^h (\ld) |^2 + | r_{\textup{odd}}^h (\ld + h) - r_{\textup{even}}^h (\ld) |^2 \big) \leq C
\end{align*}

\noi
for some $C > 0$ independent of $h$.
\end{remark}

\begin{remark} \rm
In this paper, we work on the $H^1$ initial data, which matches the conservation of energy for the Toda lattice in Flaschka's form \eqref{todah}. If we only want local-in-time convergence of Toda to KdV, the regularity assumption of the initial data in all the above theorems can be reduced to $s > \frac 34$; see \cite{HKY} for the more general FPU case. The constraint $s > \frac 34$ comes from the maximal function estimate in Lemma~\ref{LEM:lin_d}~(iii).

It would be of interest to see if the convergence of Toda (or FPU) to KdV works for $L^2$ initial data (i.e. $s = 0$), since the $L^2$-norm is controlled by the mass (or the Hamiltonian). More precisely, we consider initial data such that the mass $M^h$ in \eqref{defH1} is well-defined and finite. In this setting, the Toda lattice (or FPU) is expected to be globally well-posed thanks to the conservation of mass. Also, the KdV equation is known to be globally well-posed in $L^2$ due to Bourgain \cite{Bour93k}. The main difficulty lies in estimating the difference of the two dynamics. See also \cite{HKY} for a discussion on this issue.
\end{remark}

\begin{remark} \rm
There has been extensive research also on the periodic Toda lattice (or the periodic FPU system) and its connection to the KdV equation on the circle. See \cite{Gie1, Gie2, BP, BKP2, BKP3} and the references therein. In a recent preprint \cite{KY}, Kwak and Yang studied the periodic FPU system where the number of particles tends to infinity in a fixed domain. Specifically, they showed a local-in-time result on the continuum limit of the periodic FPU system to the KdV equation on the unit circle. In particular, they reduced the regularity requirement of the initial data to $s > 0$.
This result is remarkable due to the lack of local smoothing and the maximal function estimates in the periodic setting. To overcome this issue, the authors in \cite{KY} implemented the normal form reduction to regularize the system. It would be of interest to see whether, in the case $s = 1$, the global-in-time continuum limit holds true for the periodic Toda lattice with the help of the conservation laws. Nevertheless, the interesting case $s = 0$ for the $L^2$ initial data is left open also in the periodic setting even for the local-in-time continuum limit.
\end{remark}

\begin{remark} \rm
\label{RMK:toda}
As mentioned above, the Toda lattice \eqref{toda_ab} is a completely integrable system. Thus, it enjoys infinitely many conservation laws, which have been built in \cite{HKV}. As we will see in Subsection~\ref{SUB:cons} and Remark~\ref{RMK:cons} below, after a rearrangement of the quantities, we build two conserved quantities for the Toda lattice that are clearly connected to the corresponding conservation laws for the KdV equation. One may then ask the same question for all the other higher order conservation laws for the Toda lattice, i.e.~whether there exists a one-to-one correspondence between the conservation laws for the Toda lattice built in \cite{HKV} and those for the KdV equations.
\end{remark}

\subsection{Organization of the paper}

This paper is organized as follows. In Section~\ref{SEC:lem}, we discuss notations, function spaces, and properties of relevant linear propagators. In Section~\ref{SEC:cons}, we establish some useful conserved quantities of the scaled Toda lattice and prove a uniform $H^1$ a priori bound for the solution of the scaled Toda lattice. This leads to global well-posedness of the scaled Toda lattice as stated in Theorem~\ref{THM:conv}~(i). In Section~\ref{SEC:conv}, we prove all the other claimed results. Specifically, in Subsection~\ref{SUB:bdd}, we prove some useful a priori bounds for solutions of both the scaled Toda lattice and the KdV equation. In Subsection~\ref{SUB:dyn}, we prove Theorem~\ref{THM:conv}~(ii), dynamical convergence of the scaled Toda lattice to the KdV equation. 
In Subsection~\ref{SUB:lw1}, we prove Theorem~\ref{THM:lw}, the long-wave limit of the scaled Toda lattice in Flaschka's form. 
In Subsection~\ref{SUB:lw2}, we prove Theorem~\ref{THM:long-wave}, global well-posedness and also the long-wave limit of the scaled Toda lattice in the original form.

\section{Function spaces and preliminary lemmas}
\label{SEC:lem}

\subsection{Notations}

For any quantities $A, B > 0$, we write $A \les B$ if $A \leq C B$ for some constant $C > 0$ independent of the ranges where $A$ and $B$ are allowed to vary. We write $A \sim B$ to denote that both $A \les B$ and $A \ges B$ hold. We write $A \ll B$ if $A \leq c B$ for some small $c > 0$.

For a Schwartz function $f$, we define the Fourier transform of $f$ as
\begin{align*}
    \ft f (\xi) = \int_\R f (x) e^{- i \xi x} dx ,
\end{align*}

\noi
and also the inverse Fourier transform
\begin{align*}
    f^\vee (x) = \frac{1}{2 \pi} \int_\R f (\xi) e^{i \xi x} d\xi .
\end{align*}

\noi
Note that the Fourier transform of the product of two functions is equal to the convolution of their Fourier transforms:
\begin{align}
\ft{f g} (\xi) = \frac{1}{2 \pi} (\ft f * \ft g) (\xi) = \frac{1}{2 \pi} \int_\R \ft f (\xi - \xi') \ft g (\xi') d \xi' .
\label{prod0}
\end{align}

\noi
We denote by $\dx$ the Fourier multiplier with symbol $i \xi$. Given $s \in \R$, we define 
the inhomogeneous $L^2$-based Sobolev space $H^s (\R)$ via the norms
\begin{align*}
    \| f \|_{H^s (\R)} := \| \jb{\dx}^s f \|_{L^2 (\R)} = \frac{1}{(2 \pi)^{\frac 12}} \| \jb{\cdot}^s \ft f \|_{L^2 (\R)},
\end{align*}

\noi
where the second equality follows from Plancherel's identity. Given $R > 0$, we denote by $P_R$ the frequency projector onto $[-R, R]$, namely $\ft{P_R f} = \ind_{[-R, R]} \ft f$. 
We also denote $P_R^\perp = \textup{Id} - P_R$.

Given $T > 0$ and a Banach space $X$, we sometimes use the short-hand notation $C_T X$ and $L_T^\infty X$ to denote the function spaces $C ([-T, T]; X)$ and $L^\infty ([-T, T]; X)$, respectively.


\subsection{Function spaces on scaled lattice}
\label{SUB:space}
Given a function $f^h : h \Z \to \mathbb{C}$, we define the Lebesgue spaces $L^p (h \Z)$ via the norm
\begin{align*}
    \| f^h \|_{L^p (h \Z)} := \Big( h \sum_{\ld \in h \Z} | f^h (\ld) |^p \Big)^{\frac 1p}
\end{align*}

\noi
for $1 \leq p < \infty$ and
\begin{align*}
    \| f^h \|_{L^\infty (h \Z)} := \sup_{\ld \in h \Z} |f^h (\ld)| .
\end{align*}

\noi
For $f^h \in L^1 (h \Z)$, we define the discrete Fourier transform of $f^h$ by
\begin{align}
    (\F^h f^h) (\xi) := h \sum_{\ld \in h \Z} f^h (\ld) e^{- i \xi \ld}
\label{defFh}
\end{align}

\noi
for any $\xi \in \R / (\frac{2 \pi}{h} \Z) \cong [- \frac{\pi}{h}, \frac{\pi}{h})$, which is a $\frac{2 \pi}{h}$-periodic function. Note that we have the following Plancherel's (or Parseval's) identity:
\begin{align}
    h \sum_{\ld \in h \Z} f^h (\ld) \cj{g^h (\ld)} = \frac{1}{2 \pi} \int_{- \frac{\pi}{h}}^{\frac{\pi}{h}} (\F^h f^h) (\xi) \cj{(\F^h g^h) (\xi)} d \xi .
\label{plan}
\end{align}

\noi
We also have the following identity regarding the discrete Fourier transform of the product of functions:
\begin{align}
    \F^h (f^h g^h) (\xi) 
    = \frac{1}{2 \pi} \int_{- \frac{\pi}{h}}^{\frac{\pi}{h}} (\F^h f^h) (\xi - \xi') (\F^h g^h) (\xi') d \xi' .
\label{Fprod}
\end{align}


We denote by $\dh^+$ the discrete Fourier multiplier with symbol $\frac{e^{i h \xi} - 1}{h} = \frac{2i}{h} e^{\frac{i h \xi}{2}} \sin (\frac{h \xi}{2})$ and by $\dh^-$ the discrete Fourier multiplier with symbol $\frac{1 - e^{- i h \xi}}{h} = \frac{2i}{h} e^{- \frac{i h \xi}{2}} \sin (\frac{h \xi}{2})$. 
These two symbols correspond to the discrete differential operators
\begin{align*}
    \dh^+ f^h (\ld) = \frac{1}{h} \big( f^h (\ld + h) - f^h (\ld) \big) 
    \quad \text{and} \quad
    \dh^- f^h (\ld) = \frac{1}{h} \big( f^h (\ld) - f^h (\ld - h) \big)
\end{align*}

\noi
respectively.
Note that the two symbols have the same modulus: $|\dh^+| = |\dh^-|$. We also write $\dh = \frac 12 ( \dh^+ + \dh^- )$, which is the discrete Fourier multiplier with symbol $\frac{i}{h} \sin (h \xi)$. Given $s \in \R$, we define the homogeneous $L^2$-based Sobolev space $\dot{H}^{s} (h \Z)$ and the inhomogeneous $L^2$-based Sobolev space $H^{s} (h \Z)$ via the norms
\begin{align*}
    \| f^h \|_{\dot{H}^{s} (h \Z)} &:= \| |\dh^+|^s f^h \|_{L^2 (h \Z)}, \\
    \| f^h \|_{H^{s} (h \Z)} &:= \| \jb{\dh^+}^s f^h \|_{L^2 (h \Z)},
\end{align*}

\noi
respectively. 
The reason why we use $|\dh^+|$ instead of $|\dh|$ to define the Sobolev spaces is that for $\xi \in [- \frac{\pi}{h}, \frac{\pi}{h})$, we have $|\frac{2}{h} \sin (\frac{h \xi}{2})| \sim |\xi|$ for $|\dh^+|$, but we only have $|\frac{1}{h} \sin (h \xi)| \les |\xi|$ for $|\dh|$. Another reason is that $\dh$ separates the analysis of functions on even and odd indices, which is not suitable for defining the norm of a function on the lattice. Note that when $s = 1$, we have the following identity in view of Plancherel's identity \eqref{plan}: 
\begin{align}
    \| f^h \|_{\dot{H}^1 (h \Z)}^2 = \sum_{\ld \in h \Z} h^{-1} |f^h (\ld + h) - f^h (\ld)|^2 .
\label{H1norm}
\end{align}

\noi
In view of the extension of a function on $h \Z$ to a function on $\R$ in \eqref{extend}, it makes sense to view $\dh^\pm$ and $\dh$ also as Fourier multipliers on functions defined on $\R$ with the same symbol. In particular, for any function $f$ defined on $\R$, we have
\begin{align}
    \dh f (\cdot) = \frac{1}{2h} \big( f (\cdot + h) - f (\cdot - h) \big) .
\label{dhf}
\end{align}

\medskip
Below, we present some useful lemmas on function spaces on scaled lattice.

\begin{lemma}
\label{LEM:Lpbdd}
Let $0 < h \leq 1$ and $1 \leq p \leq q \leq \infty$. Then, for any $f^h \in L^p (h \Z)$, we have
\begin{align*}
    \| f^h \|_{L^q (h \Z)} \leq h^{\frac{1}{q} - \frac{1}{p}} \| f^h \|_{L^p (h \Z)} .
\end{align*}
\end{lemma}

\begin{proof}
The estimate follows directly from the fact that $\| f^h \|_{L^q (h \Z)} = h^{\frac{1}{q}} \| f^h \|_{\ell^q (h \Z)}$ and the embedding $\ell^p \subset \ell^q$ whenever $1 \leq p \leq q \leq \infty$.
\end{proof}

We also have the following Gagliardo-Nirenburg inequality. For a proof, see \cite[Proposition~2.4]{HY}.
\begin{lemma}
\label{LEM:GN}
Let $0 < h \leq 1$. Let $f^h$ be a function on $h \Z$.
Let $s \geq 0$, $2 \leq q \leq \infty$, and $0 < \ta < 1$ be such that $\ta s = \frac{1}{2} - \frac{1}{q}$. Then, we have
\begin{align*}
    \| f^h \|_{L^q (h \Z)} \les \| f^h \|_{L^2 (h \Z)}^{1 - \ta} \| f^h \|_{\dot{H}^{s} (h \Z)}^{\ta} ,
\end{align*}

\noi
where the underlying constant is independent of $h$.
\end{lemma}

\subsection{On the extension operator}
We recall that $\mathcal{E}$ denotes the extension of a function on $h \Z$ to be a function on $\R$ defined in \eqref{extend}. The following lemma shows that, under the extension operator $\mathcal{E}$, the $H^s$-norm on the lattice and the $H^s$-norm on the continuum are equivalent.

\begin{lemma}
\label{LEM:Hsbdd}
Let $0 < h \leq 1$ and $s \in \R$. Let $f^h$ be a function on $h \Z$. Then, we have
\begin{align*}
    \| f^h \|_{H^s (h \Z)} \les \| \mathcal{E} f^h \|_{H^s (\R)} \les \| f^h \|_{H^s (h \Z)},
\end{align*}

\noi
where the underlying constant is independent of $h$,
\end{lemma}

\begin{proof}
From Plancherel's identities and the fact that $|\frac{2}{h} \sin (\frac{h \xi}{2})| \sim |\xi|$ for $\xi \in [- \frac{\pi}{h}, \frac{\pi}{h}]$, we have
\begin{align*}
    \| f^h \|_{H^s (h \Z)}^2 
    &= \frac{1}{2\pi} \int_{- \frac{\pi}{h}}^{\frac{\pi}{h}} \langle\tfrac{2}{h} \sin ( \tfrac{h \xi}{2} ) \rangle^{2 s} | \F^h (f^h) (\xi) |^2 d \xi \\
    &\les \frac{1}{2 \pi} \int_{- \frac{\pi}{h}}^{\frac{\pi}{h}} \jb{\xi}^{2 s} \big| \ft{f^h} (\xi) \big|^2 d \xi \\
    &= \| \mathcal{E} f^h \|_{H^s (\R)}^2 .
\end{align*}

\noi
Also, from Plancherel's identities and the fact that $|\frac{2}{h} \sin (\frac{h \xi}{2})| \sim |\xi|$ for $\xi \in [- \frac{\pi}{h}, \frac{\pi}{h}]$, we have
\begin{align*}
    \| \mathcal{E} f^h \|_{H^s (\R)}^2 &= \frac{1}{2 \pi} \int_{- \frac{\pi}{h}}^{\frac{\pi}{h}} \jb{\xi}^{2 s} \big| \ft{f^h} (\xi) \big|^2 d \xi \\
    &\les \frac{1}{2 \pi} \int_{- \frac{\pi}{h}}^{\frac{\pi}{h}} \langle\tfrac{2}{h} \sin ( \tfrac{h \xi}{2} ) \rangle^{2 s} | \F^h (f^h) (\xi) |^2 d \xi \\
    &= \| f^h \|_{H^s (h \Z)}^2 .
\end{align*}

\noi
Thus, we obtain the norm equivalence.
\end{proof}

We now establish a useful identity regarding the extension of a product of two functions defined on the lattice.
\begin{lemma}
\label{LEM:Eprod}
Let $0 < h \leq 1$. Let $f^h$ and $g^h$ be two functions on $h \Z$. Then, we have
\begin{align*}
    \mathcal{E} (f^h g^h) = P_{\frac{\pi}{h}} \big( (1 + 2 \cos (\tfrac{2 \pi}{h} \cdot)) \mathcal{E} f^h \mathcal{E} g^h \big) .
\end{align*}
\end{lemma}

\begin{proof}
From \eqref{Fprod} and the fact that $\F^h f^h$ and $\F^h g^h$ are $\frac{2 \pi}{h}$-periodic functions, we have that for any $\xi \in [-\frac{\pi}{h}, \frac{\pi}{h})$,
\begin{align}
\begin{split}
    \F^h (f^h g^h) (\xi) &= \frac{1}{2\pi} \int_{- \frac{\pi}{h}}^{\frac{\pi}{h}} (\F^h f^h) (\xi - \xi') (\F^h g^h) (\xi') d \xi' \\
    &= \frac{1}{2\pi} \int_\R (\F^h f^h) (\xi - \xi') \ind_{[- \frac{\pi}{h}, \frac{\pi}{h})} (\xi - \xi') (\F^h g^h) (\xi') \ind_{[- \frac{\pi}{h}, \frac{\pi}{h})} (\xi') d \xi' \\
    &\,\,\,\, + \frac{1}{2\pi} \int_\R (\F^h f^h) (\xi - \tfrac{2 \pi}{h} - \xi') \ind_{[- \frac{\pi}{h}, \frac{\pi}{h})} (\xi - \tfrac{2 \pi}{h} - \xi') (\F^h g^h) (\xi') \ind_{[- \frac{\pi}{h}, \frac{\pi}{h})} (\xi') d \xi' \\
    &\,\,\,\, + \frac{1}{2\pi} \int_\R (\F^h f^h) (\xi + \tfrac{2 \pi}{h} - \xi') \ind_{[- \frac{\pi}{h}, \frac{\pi}{h})} (\xi + \tfrac{2 \pi}{h} - \xi') (\F^h g^h) (\xi') \ind_{[- \frac{\pi}{h}, \frac{\pi}{h})} (\xi') d \xi' .
\end{split}
\label{Eprod1}
\end{align}

\noi
Then, from \eqref{Eprod1} and \eqref{prod0}, we write
\begin{align}
\begin{split}
    \F^h (f^h g^h) (\xi) &= \frac{1}{2 \pi} \big( (\ft{\mathcal{E} f^h} * \ft{\mathcal{E} g^h} ) (\xi) + (\ft{\mathcal{E} f^h} * \ft{\mathcal{E} g^h} ) (\xi - \tfrac{2 \pi}{h}) + (\ft{\mathcal{E} f^h} * \ft{\mathcal{E} g^h} ) (\xi + \tfrac{2 \pi}{h}) \big) \\
    &= \ft{\mathcal{E} f^h \mathcal{E} g^h} (\xi) + \ft{\mathcal{E} f^h \mathcal{E} g^h} (\xi - \tfrac{2 \pi}{h}) + \ft{\mathcal{E} f^h \mathcal{E} g^h} (\xi + \tfrac{2 \pi}{h}) \\
    &= \big( (1 + e^{i \frac{2 \pi}{h} \cdot} + e^{- i \frac{2 \pi}{h} \cdot}) \mathcal{E} f^h \mathcal{E} g^h \big)^\wedge (\xi) \\
    &= \big( (1 + 2 \cos (\tfrac{2 \pi}{h} \cdot)) \mathcal{E} f^h \mathcal{E} g^h \big)^\wedge (\xi) .
\end{split}
\label{Eprod2}
\end{align}
Thus, by putting $\ind_{[- \frac{\pi}{h}, \frac{\pi}{h})} (\xi)$ on both sides of \eqref{Eprod2} and taking the inverse Fourier transform, we obtain the desired identity.
\end{proof}

The cosine function in Lemma~\ref{LEM:Eprod} turns out to be useful via the following estimate.
\begin{lemma}
\label{LEM:Ecos}
Let $0 < h \leq 1$. Let $f$ be a function on $\R$. Then, we have
\begin{align*}
    \big\| P_{\frac{\pi}{h}} ( e^{\pm i \frac{2 \pi}{h} \cdot} f ) \big\|_{L^2 (\R)} \les h \| f \|_{H^1 (\R)} .
\end{align*}
\end{lemma}

\begin{proof}
Without loss of generality, we only consider the case with $e^{i \frac{2 \pi}{h} \cdot}$. We denote $\xi$ the frequency of $P_{\frac{\pi}{h}} ( e^{i \frac{2 \pi}{h} \cdot} f )$ and $\xi'$ the frequency of $f$. Then, we have $\xi \in [- \frac{\pi}{h}, \frac{\pi}{h}]$ and $\xi = \xi' + \frac{2 \pi}{h}$, which implies that $\xi' \in [- \frac{3 \pi}{h}, - \frac{\pi}{h}]$. Thus, by denoting $P_{[- \frac{3 \pi}{h}, - \frac{\pi}{h}]}$ be the frequency projection onto $[- \frac{3 \pi}{h}, - \frac{\pi}{h}]$ and using the fact that $|\xi'| \geq \frac{\pi}{h}$, we have
\begin{align*}
    \big\| P_{\frac{\pi}{h}} ( e^{\pm i \frac{2 \pi}{h} \cdot} f ) \big\|_{L^2 (\R)} = \big\| P_{\frac{\pi}{h}} \big( e^{\pm i \frac{2 \pi}{h} \cdot} P_{[- \frac{3 \pi}{h}, - \frac{\pi}{h}]} f \big) \big\|_{L^2 (\R)} \leq \big\| P_{[- \frac{3 \pi}{h}, - \frac{\pi}{h}]} f \big\|_{L^2 (\R)} \les h \| f \|_{H^1 (\R)} ,
\end{align*}

\noi
as desired.
\end{proof}

\subsection{Linear propagators and Fourier restriction norms}
\label{SUB:lin}

In this subsection, we exploit the properties of the linear propagator for the equation satisfied by $u^h$ defined in \eqref{defuh}. 

Let us first discuss a periodic Littlewood-Paley theory, for which we follow the procedure described in \cite[Lemma~2.7]{KOVW}. We let $\varphi$ be a smooth even function supported in $(- \frac{3 \pi}{4}, \frac{3 \pi}{4})$ such that
\begin{align*}
    \sum_{k \in \Z} \varphi (\xi + k \pi) = 1
\end{align*}

\noi
for any $\xi \in \R$.
Given $0 < h \leq 1$, we let $\varphi^h (\cdot) = \varphi (h \cdot)$, so that $\varphi^h$ is a smooth even function supported in $(- \frac{3 \pi}{4h}, \frac{3 \pi}{4h})$ such that
\begin{align*}
    \sum_{k \in \Z} \varphi^h (\xi + \tfrac{k \pi}{h}) = 1
\end{align*}

\noi
for any $\xi \in \R$. Given a dyadic number $N \leq 1$, we define
\begin{align*}
    \varphi_N^{h, +} (\xi) := \sum_{k \in \Z} \varphi^h (N^{-1} (\xi + \tfrac{2 k \pi}{h})) \quad \text{and} \quad \varphi_N^{h, -} (\xi) := \varphi^h (N^{-1} (\xi + \tfrac{(2 k - 1) \pi}{h})) ,
\end{align*}

\noi
and we define
\begin{align*}
    \psi_N^{h, \pm} (\xi) := \varphi_N^{h, \pm} (\xi) - \varphi_{\frac{N}{2}}^{h, \pm} (\xi) .
\end{align*}

\noi
We then denote $\mathbf{P}^{h, \pm}_N$ as the Fourier multiplier operator by $\psi_N^{h, \pm}$. Note that we have
\begin{align}
    \Id = \sum_{\s = \pm} \sum_{\substack{N \leq 1 \\ \text{dyadic}}} \mathbf{P}_N^{h, \pm} .
\label{per_dya}
\end{align}

In view of the formulation \eqref{eq_uh} below, the linear propagator is given by $e^{- 6 h^{-2} t (\dh - \dx)}$, which is a Fourier multiplier operator with symbol $e^{- 6 i h^{-3} t (\sin (h \xi) - h \xi)}$. We shall show useful linear estimates for $e^{- 6 h^{-2} t (\dh - \dx)}$. For some of these estimates, we need to restrict our attention to the frequency range $\xi \in [- \frac{\pi}{h}, \frac{\pi}{h}]$ or even a smaller range. For this reason and also in view of the periodic dyadic decomposition \eqref{per_dya}, we will study the following parts of the integral kernel for $|\dh|^s e^{- 6 h^{-2} t (\dh - \dx)} \mathbf{P}_N^{h, \pm}$ given $s \in \R$:
\begin{align}
\begin{split}
    |\dh|^s K_N^{h, +} (t, x) &= \frac{1}{2\pi} \int_{- \frac{3\pi}{4h}}^{\frac{3\pi}{4h}} e^{- 6 i h^{-3} t (\sin (h \xi) - h \xi) + i \xi x} |h^{-1} \sin(h \xi)|^s \psi_N^{h, +} (\xi) d \xi , \\
    |\dh|^s K_N^{h, -} (t, x) &= \frac{1}{2\pi} \int_{- \frac{7\pi}{4h}}^{- \frac{\pi}{4h}} e^{- 6 i h^{-3} t (\sin (h \xi) - h \xi) + i \xi x} |h^{-1} \sin(h \xi)|^s \psi_N^{h, -} (\xi) d \xi \\
    &\quad + \frac{1}{2\pi} \int_{\frac{\pi}{4h}}^{\frac{7\pi}{4h}} e^{- 6 i h^{-3} t (\sin (h \xi) - h \xi) + i \xi x} |h^{-1} \sin(h \xi)|^s \psi_N^{h, -} (\xi) d \xi .
\end{split}
\label{defKN}
\end{align}

Let us first show the following computations.
\begin{lemma}
\label{LEM:KN}
Let $0 < h \leq 1$.
\textup{(i)} For any $t \neq 0$, we have
\begin{align}
    \bigg\| \sum_{\substack{N \leq 1 \\ \textup{dyadic}}} |\dh|^{\frac 12} K_N^{h, +} (t, x) \bigg\|_{L_x^\infty (\R)} \les |t|^{- \frac 12} 
\label{KN+}
\end{align}

\noi
and
\begin{align}
    \bigg\| \sum_{\substack{N \leq 1 \\ \textup{dyadic}}} |\dh|^{\frac 12} K_N^{h, -} (t, x) \bigg\|_{L_x^\infty (\R)} \les |t|^{- \frac 12} ,
\label{KN-}
\end{align}

\noi
where the underlying constant is independent of $h$.

\smallskip \noi
\textup{(ii)} Let $T > 0$. Then, for any dyadic number $h < N \leq 1$, we have
\begin{align}
    \Big\| \sup_{t \in [-2T, 2T]} |K_N^{h, +} (t, x)| \Big\|_{L_x^1 (\R)} \les (1 + T)^{\frac 12} h^{- \frac 32} N^{\frac 32} .
\label{KN+1}
\end{align}

\noi
Also, we have
\begin{align}
    \bigg\| \sup_{t \in [-2T, 2T]} \bigg| \sum_{\substack{N \leq h \\ \textup{dyadic}}} K_N^{h, +} (t, x) \bigg| \bigg\|_{L_x^1 (\R)} \les (1 + T)^2 .
\label{KN+2}
\end{align}
\end{lemma}

\begin{proof}
We first show some useful estimates for $|\dh|^s K_N^{h, +}$ with a general $s \in \R$ to be used in both parts of the proof. Let $s \in \R$. Since $\xi \in [- \frac{3\pi}{4h}, \frac{3\pi}{4h}] \cap \text{supp} \, \psi_N^{h, +}$, we have $|\xi| \sim h^{-1} N$, so that
\begin{align}
    (\sin (h \xi) - h \xi)' = h (\cos (h \xi) - 1) = -2 h \sin (\tfrac{h \xi}{2})^2 \sim - h^3 \xi^2 \sim - h N^2.
\label{der}
\end{align}

\noi
Thus, using \eqref{der} along with the fact that $|h^{-1} \sin (h \xi)| \sim |\xi| \sim h^{-1} N$, $|(\psi_N^{h, +})' (\xi)| \les h N^{-1}$, and $|(\psi_N^{h, +})'' (\xi)| \les h^2 N^{-2}$, we directly compute that if $x^2 \ll h^{-4} N^4 t^2$ or $x^2 \gg h^{-4} N^4 t^2$, then
\begin{align*}
    \bigg| \Big( \frac{1}{x - 6 h^{-3} t (\sin (h \xi) - h \xi)'} \Big( \frac{|h^{-1} \sin (h \xi)|^s \psi_N^{h, +} (\xi)}{x - 6 h^{-3} t (\sin (h \xi) - h \xi)'} \Big)' \Big)' \bigg| \les \frac{h^{- s + 2} N^{s-2}}{\max (x^2, h^{-4} N^4 t^2)} ,
\end{align*}

\noi
so that integration by parts twice yields
\begin{align}
    \big| |\dh|^s K_N^{h, +} (t, x) \big| \les \frac{h^{- s + 1} N^{s-1}}{\max (x^2, h^{-4} N^4 t^2)} .
\label{KNbdd1-1}
\end{align}

\noi
We also note that when $\xi \in [- \frac{3\pi}{4h}, \frac{3\pi}{4h}] \cap \text{supp} \, \psi_N^{h, +}$, we have
\begin{align}
    \Big| \frac{d^2}{d \xi^2} \big( - 6 h^{-3} t (\sin (h \xi) - h \xi) + \xi x \big) \Big| = 6 h^{-1} |t \sin (h \xi)| \sim |t| |\xi| \sim h^{-1} N |t| .
\label{van}
\end{align}

\noi
Thus, from the Van der Corput lemma (see \cite[Corollary~2.6.8]{Gra1}) along with the fact that $|h^{-1} \sin (h \xi)| \sim |\xi| \sim h^{-1} N$ and $|(\psi_N^{h, +})' (\xi)| \les h N^{-1}$, we obtain
\begin{align}
    \big| |\dh|^s K_N^{h, +} (t, x) \big| \les h^{-s + \frac 12} N^{s - \frac 12} |t|^{- \frac 12} .
\label{KNbdd1-2}
\end{align}

\noi
Also, for any dyadic $N_0 \leq 1$, we have the trivial bound
\begin{align}
    \bigg| \sum_{\substack{N \leq N_0 \\ \text{dyadic}}} |\dh|^s K_N^{h, +} (t, x) \bigg| \les h^{-s - 1} N_0^{s + 1} .
\label{KNbdd1-3}
\end{align}

We now consider each part of the lemma.

\smallskip \noi
(i) We first consider \eqref{KN+}. From \eqref{KNbdd1-1}, \eqref{KNbdd1-2}, and \eqref{KNbdd1-3} with $N_0 \sim h |t|^{- \frac 13}$ and $s = \frac 12$, we have for any $t \neq 0$ and $x \in \R$ that
\begin{align}
\begin{split}
    \bigg| \sum_{\substack{N \leq 1 \\ \text{dyadic}}} |\dh|^{\frac 12} K_N^{h, +} (t, x) \bigg| &\leq \bigg| \sum_{\substack{N \ll h |t|^{- \frac 13} \\ \text{dyadic}}} |\dh|^{\frac 12} K_N^{h, +} (t, x) \bigg| + \sum_{\substack{N \sim h |t|^{- \frac 12} |x|^{\frac 12} \\ \text{dyadic}}} \big| |\dh|^{\frac 12} K_N^{h, +} (t, x) \big| \\
    &\quad + \sum_{\substack{h |t|^{- \frac 13} \les N \leq 1, N \nsim h |t|^{- \frac 12} |x|^{\frac 12} \\ \text{dyadic}}} \big| |\dh|^{\frac 12} K_N^{h, +} (t, x) \big| \\
    &\les |t|^{- \frac 12} + |t|^{- \frac 12} + \sum_{\substack{h |t|^{- \frac 13} \les N \leq 1 \\ \text{dyadic}}} h^{\frac 92} N^{- \frac 92} |t|^{-2}  \\
    &\les |t|^{- \frac 12} ,
\end{split}
\label{KNbdd1}
\end{align}

\noi
which gives \eqref{KN+}.

We now consider \eqref{KN-}. When $\xi \in [-\frac{7 \pi}{4h}, - \frac{\pi}{4h}] \cap \text{supp}\, \psi_N^{h, -}$, we have $|\xi + \tfrac{\pi}{h}| \sim h^{-1} N$, so that
\begin{align}
    (\sin (h \xi) + h \xi)' = h (\cos (h \xi) + 1) = 2 h \cos (\tfrac{h \xi}{2})^2 = 2h \sin (\tfrac{h \xi}{2} + \tfrac{\pi}{2})^2 \sim h^3 |\tfrac{\xi}{h} + \tfrac{\pi}{h}|^2 \sim h N^2 .
\label{der2}
\end{align}

\noi
Thus, using \eqref{der2} along with the fact that $|h^{-1} \sin (h \xi)| = |h^{-1} \sin (h \xi + \pi)| \sim |\xi + \frac{\pi}{h}| \sim h^{-1} N$, $|(\psi_N^{h, -})' (\xi)| \les h N^{-1}$, and $|(\psi_N^{h, -})'' (\xi)| \les h^2 N^{-2}$, we directly compute that if $(x + 12 h^{-2}t)^2 \ll h^{-4} N^4 t^2$ or $(x + 12 h^{-2} t)^2 \gg h^{-4} N^4 t^2$, then
\begin{align}
\begin{split}
    \bigg| &\Big( \frac{1}{x + 12 h^{-2} t - 6 h^{-3} t (\sin (h \xi) + h \xi)'} \Big( \frac{|h^{-1} \sin (h \xi)|^{\frac 12} \psi_N^{h, -} (\xi)}{x + 12 h^{-2} t - 6 h^{-3} t (\sin (h \xi) + h \xi)'} \Big)' \Big)' \bigg| \\
    &\les \frac{h^{\frac 32} N^{- \frac 32}}{\max ((x + 12 h^{-2} t)^2, h^{-4} N^4 t^2)} .
\end{split}
\label{ibp2}
\end{align}

\noi
When $\xi \in [\frac{\pi}{4 h}, \frac{7 \pi}{4 h}] \cap \text{supp} \, \psi_N^{h, -}$, we have $|h^{-1} \sin (h \xi)| = |h^{-1} \sin (h \xi - \pi)| \sim |\xi - \frac{\pi}{h}| \sim h^{-1} N$. Thus, using similar steps, the estimate \eqref{ibp2} also holds in this case. Thus, by integration by parts twice, we see that if $(x + 12 h^{-2}t)^2 \ll h^{-4} N^4 t^2$ or $(x + 12 h^{-2} t)^2 \gg h^{-4} N^4 t^2$, then
\begin{align}
    \big| |\dh|^{\frac 12} K_N^{h, -} (t, x) \big| \les \frac{h^{\frac 12} N^{- \frac 12}}{\max ((x + 12 h^{-2} t)^2, h^{-4} N^4 t^2)} .
\label{KNbdd2-1}
\end{align}

\noi
Also, since $|h^{-1} \sin (h \xi)| \sim h^{-1} N$ for $\xi \in ([- \frac{7 \pi}{4 h}, - \frac{\pi}{4 h}] \cup [\frac{\pi}{4 h}, \frac{7 \pi}{4 h}] ) \cap \text{supp} \, \psi_N^{h, -}$, we know that \eqref{van} holds for such $\xi$, so that from the Van der Corput lemma (see \cite[Corollary~2.6.8]{Gra1}) as in \eqref{KNbdd1-2}, we obtain
\begin{align}
    \big| |\dh|^{\frac 12} K_N^{h, -} (t, x) \big| \les  |t|^{- \frac 12} .
\label{KNbdd2-2}
\end{align}

\noi
Moreover, for any dyadic $N_0 \leq 1$, we have the trivial bound
\begin{align}
    \bigg| \sum_{\substack{N \leq N_0 \\ \text{dyadic}}} |\dh|^{\frac 12} K_N^{h, -} (t, x) \bigg| \les h^{- \frac 32} N_0^{\frac 32} .
\label{KNbdd2-3}
\end{align}

\noi
Thus, the desired estimate \eqref{KN-} follows from similar steps as in \eqref{KNbdd1} along with \eqref{KNbdd2-1}, \eqref{KNbdd2-2}, and \eqref{KNbdd2-3}.

\smallskip \noi
(ii) We first consider \eqref{KN+1}. Since $h < N \leq 1$, from \eqref{KNbdd1-1}, \eqref{KNbdd1-2}, and \eqref{KNbdd1-3}, we have
\begin{align}
    \sup_{t \in [-2T, 2T]} |K_N^{h, +} (t, x)| \les
    \begin{cases}
        h^{-1} N & \text{if } |x| \leq 1  \\
        h^{- \frac 12} N^{\frac 12} |x|^{- \frac 12} & \text{if } 1 \leq |x| \les h^{-2} N^2 T \\
        h N^{-1} |x|^{-2} & \text{if } |x| \gg h^{-2} N^2 T .
    \end{cases}
\label{KNbdd4-1}
\end{align}

\noi
The desired estimate \eqref{KN+1} then follows directly from \eqref{KNbdd4-1}.

We now consider \eqref{KN+2}. Let us write
\begin{align*}
    K^{h, +} (t, x) := \sum_{\substack{N \leq h \\ \text{dyadic}}} K_N^{h, +} (t, x) = \frac{1}{2\pi} \int_{- \frac{3 \pi}{4h}}^{\frac{3 \pi}{4h}} e^{- 6 i h^{-3} t (\sin (h \xi) - h \xi) + i \xi x} \psi^{h, +} (\xi) d \xi
\end{align*}

\noi
with
\begin{align*}
    \psi^{h, +} (\xi) = \sum_{\substack{N \leq h \\ \text{dyadic}}} \psi_N^{h, +} (\xi) .
\end{align*}

\noi
Using the fact that $|\xi| \les 1$ on the support of $\psi^{h, +}$, we have
\begin{align*}
    &| h^{-3} t (\sin (h \xi) - h \xi)' | = |h^{-2} t (\cos (h \xi) - 1)| \les |t| |\xi|^2 \les |t| , \\
    &| h^{-3} t (\sin (h \xi) - h \xi)'' | = |h^{-1} t \sin (h \xi)| \les |t| |\xi| \les |t|,
\end{align*}

\noi
so that using the fact that $|(\psi^{h, +})' (\xi)| \les 1$ and $|(\psi^{h, +})'' (\xi)| \les 1$, we directly compute that
\begin{align*}
    \big| \big( e^{- 6 i h^{-3} t (\sin (h \xi) - h \xi)} \psi^{h, +} (\xi) \big)'' \big| \les 1 + t^2 .
\end{align*}

\noi
Thus, by integration by parts twice, we get
\begin{align}
    \sup_{t \in [-2T, 2T]} |K^{h, +} (t, x)| \les
    \begin{cases}
        1 & \text{if } |x| \leq 1  \\
        (1 + T)^2 |x|^{-2} & \text{if } |x| > 1 .
    \end{cases}
\label{KNbdd4-2}
\end{align}

\noi
The desired estimate \eqref{KN+2} then follows direct from \eqref{KNbdd4-2}.
\end{proof}

We now show the following list of useful linear estimates for not only the linear propagator $e^{- 6 h^{-2} t (\dh - \dx)}$ but also the Airy flow $e^{-t \dx^3}$.

\begin{lemma}
\label{LEM:lin_d}
Let $0 < h \leq 1$. Let $f$ be a function defined on $\R$.

\smallskip \noi
\textup{(i)~($H^s$-isometry)} For any $s \in \R$, we have
\begin{align*}
    \big\| e^{- 6 h^{-2} t (\dh - \dx)} P_{\frac{\pi}{h}} f \big\|_{C_t (\R ; H_x^s (\R))} \leq \| f \|_{H^s (\R)} 
\end{align*}

\noi
and
\begin{align*}
    \big\| e^{-t \dx^3} f \big\|_{C_t (\R ; H_x^s ( \R))} = \| f \|_{H^s (\R)} .
\end{align*}

\smallskip \noi
\textup{(ii)~(Strichartz estimates)} Let $4 \leq q \leq \infty$ and $2 \leq r \leq \infty$ be such that $\frac{2}{q} + \frac{1}{r} = \frac 12$. Then, we have
\begin{align*}
    \big\| |\dh|^{\frac 1q} e^{- 6 h^{-2} t (\dh - \dx)} P_{\frac{\pi}{h}} f \big\|_{L^q_t ( \R ; L^{r}_x (\R))} \les \| f \|_{L^2 (\R)} ,
\end{align*}

\noi
where the underlying constant is independent of $h$. Moreover, we have
\begin{align*}
    \big\| |\dx|^{\frac 1q} e^{- t \dx^3} f \big\|_{L^q_t ( \R ; L^{r}_x (\R))} \les \| f \|_{L^2 (\R)} .
\end{align*}

\smallskip \noi
\textup{(iii)~(Maximal function estimates)} Let $s > \frac 34$ and $T > 0$. Then, we have
\begin{align*}
    \big\| e^{- 6 h^{-2} t (\dh - \dx)} P_{\frac{\pi}{4h}} f \big\|_{L_x^2 (\R ; L_t^\infty ([-T, T]))} \les (1 + T) \| \jb{\dh}^s f \|_{L^2 (\R)} ,
\end{align*}

\noi
where the underlying constant is independent of $h$. Moreover, we have
\begin{align*}
    \big\| e^{- t \dx^3} f \big\|_{L_x^2 ( \R ; L_t^\infty ([-T, T]))} \les (1 + T)^{\frac 12} \| f \|_{H^s (\R)}.
\end{align*}

\smallskip \noi
\textup{(iv)~(Local smoothing estimates)} For any $0 \leq \ta \leq 1$, we have
\begin{align*}
    \big\| \dh^\pm e^{- 6 \ta h^{-2} t (\dh - \dx) - (1 - \ta) t \dx^3} P_{\frac{\pi}{h}} f \big\|_{L_x^\infty ( \R ; L_t^2 (\R))} \les \| f \|_{L^2 (\R)} ,
\end{align*}

\noi
where the underlying constant is independent of $h$ and $\ta$. The same estimate holds if $\dh^\pm$ is replaced by $\dx$. Moreover, we have
\begin{align*}
    \| \dx e^{- t \dx^3} f \|_{L_x^\infty ( \R ; L_t^2 (\R))} \les \| f \|_{L^2 (\R)}.
\end{align*}
\end{lemma}

\begin{proof}
(i) The $H^s$-isometry follows directly from Plancherel's identity. The continuity in time follows from the dominated convergence theorem.

\smallskip \noi
(ii) The Strichartz estimates for the Airy flow $e^{- t \dx^3}$ are classical. See, for example, \cite{KPV}.

We now prove the Strichartz estimates for $e^{-6 h^{-2} t (\dh - \dx)}$. Let us first consider a smooth frequency cutoff on $[-\frac{\pi}{h}, \frac{\pi}{h}]$ instead of the sharp frequency cutoff $P_{\frac{\pi}{h}}$. 
Let $\phi$ be a smooth function supported on $\{ \xi \in \R: |\xi| \leq \frac 54 \}$ such that $\phi (\xi) = 1$ for all $|\xi| \leq 1$ and for any $R > 0$, we define $\phi_R (\cdot) = \phi (R^{-1} \cdot)$. We denote by $\mathcal{P}_R$ as the Fourier multiplier operator by $\phi_R$. Note that $\mathcal{P}_{R}$ is bounded uniformly in $R$ in $L^p (\R)$ for any $1 \leq p \leq \infty$, which follows easily from Young's convolution inequality and a scaling argument. Using the periodic dyadic decomposition \eqref{per_dya} and recalling the definitions of $K_N^{h, +}$ and $K_N^{h, -}$ in \eqref{defKN}, we write
\begin{align}
    |\dh|^{\frac 12} e^{- 6 h^{-2} t (\dh - \dx)} \mathcal{P}_{\frac{\pi}{h}} f = \sum_{\s = \pm} \sum_{\substack{N \leq 1 \\ \text{dyadic}}} |\dh|^{\frac 12} K_N^{h, \s} (t) * \mathcal{P}_{\frac{\pi}{h}} f .
\label{KNker}
\end{align}

\noi
By \eqref{KNker}, Young's convolution inequality, Lemma~\ref{LEM:KN}~(i), and the boundedness of $\mathcal{P}_{\frac{\pi}{h}}$ in $L^1 (\R)$, we obtain the following dispersive estimate:
\begin{align}
\begin{split}
    \big\| |\dh|^{\frac 12} e^{- 6 h^{-2} t (\dh - \dx)} \mathcal{P}_{\frac{\pi}{h}} f \big\|_{L_x^\infty (\R)} &= \bigg\| \sum_{\s = \pm} \sum_{\substack{N \leq 1 \\ \text{dyadic}}} |\dh|^{\frac 12} K_N^{h, \s} (t) * \mathcal{P}_{\frac{\pi}{h}} f \bigg\|_{L_x^\infty (\R)} \\
    &\les \bigg\| \sum_{\s = \pm} \sum_{\substack{N \leq 1 \\ \text{dyadic}}} |\dh|^{\frac 12} K_N^{h, \s} (t, x) \bigg\|_{L_x^\infty (\R)} \| \mathcal{P}_{\frac{\pi}{h}} f \|_{L^1 (\R)} \\
    &\les |t|^{- \frac 12} \| f \|_{L^1 (\R)}.
\end{split}
\label{disp}
\end{align}

\noi
Using a complex interpolation of \eqref{disp} and the $L^2$-isometry in part (i), we obtain
\begin{align}
    \big\| |\dh|^{\frac{2}{q}} e^{- 6 h^{-2} t (\dh - \dx)} \mathcal{P}_{\frac{\pi}{h}} f \big\|_{L_x^r (\R)} \les |t|^{- (\frac 12 - \frac 1r)} \| f \|_{L^{r'} (\R)} 
\label{disp2}
\end{align}

\noi
for any $2 \leq r \leq \infty$ and $4 \leq q \leq \infty$ satisfying $\frac 2q + \frac 1r = \frac 12$, where $r'$ is the H\"older conjugate of $r$. Thus, for any space-time function $F$ defined on $\R \times \R$, from Minkowski's integral inequality, \eqref{disp2}, and the Hardy-Littlewood-Sobolev inequality, we obtain
\begin{align*}
    \bigg\| &\int_\R |\dh|^{\frac 2q} e^{- 6 h^{-2} (t - t') (\dh - \dx)} \mathcal{P}_{\frac{\pi}{h}} F (t', x) dt' \bigg\|_{L_t^q ( \R ; L_x^r (\R))} \\
    &\leq \bigg\| \int_\R \big\| |\dh|^{\frac 2q} e^{- 6 h^{-2} (t - t') (\dh - \dx)} \mathcal{P}_{\frac{\pi}{h}} F (t', x) \big\|_{L_x^r (\R)} dt' \bigg\|_{L_t^q (\R)} \\
    &\les \bigg\| \int_\R |t - t'|^{- (\frac 12 - \frac 1r)} \| F (t', x) \|_{L_x^{r'} (\R)} dt' \bigg\|_{L_t^q (\R)} \\
    &\les \| F \|_{L_t^{q'} ( \R ; L_x^{r'} (\R) )} ,
\end{align*}

\noi
where $q'$ is the H\"older conjugate of $q$. Therefore, from a standard $TT^*$-argument, we obtain
\begin{align*}
    \big\| |\dh|^{\frac{2}{q}} e^{- 6 h^{-2} t (\dh - \dx)} \mathcal{P}_{\frac{\pi}{h}} f \big\|_{L_t^q (\R; L_x^r (\R))} \les \| f \|_{L^2 (\R)} .
\end{align*}

\noi
The desired Strichartz estimate then follows from the fact that $P_{\frac{\pi}{h}} = \mathcal{P}_{\frac{\pi}{h}} P_{\frac{\pi}{h}}$ and the $L^2$-boundedness of $P_{\frac{\pi}{h}}$ at the end.

\smallskip \noi
(iii) The maximal function estimate for the Airy flow $e^{-t \dx^3}$ is established in \cite{KPV}. We focus on proving the maximal function estimate for $e^{- 6 h^{-2} t (\dh - \dx)}$. Since we have the frequency restriction $\xi \in [- \frac{\pi}{4h}, \frac{\pi}{4h}]$, we use the periodic dyadic decomposition \eqref{per_dya} to write
\begin{align}
    e^{- 6 h^{-2} t (\dh - \dx)} P_{\frac{\pi}{4h}} f = \sum_{\substack{N \leq 1 \\ \text{dyadic}}} e^{- 6 h^{-2} t (\dh - \dx)} \mathbf{P}_N^{h, +} P_{\frac{\pi}{4h}} f = \sum_{\substack{N \leq 1 \\ \text{dyadic}}} K_N^{h, +} * P_{\frac{\pi}{4h}} f ,
\label{KNker2}
\end{align}

\noi
where $K_N^{h, +}$ is defined in \eqref{defKN}.
Then, for any dyadic number $h < N \leq 1$ and space-time function $F$ defined on $\R \times h\Z$, by \eqref{KNker2}, Minkowski's integral inequality, \eqref{KN+1} in Lemma~\ref{LEM:KN}~(ii), and Young's convolution inequality, we obtain
\begin{align}
\begin{split}
    \bigg\| &\int_{-T}^T  e^{- 6 h^{-2} (t - t') (\dh - \dx)} (\mathbf{P}_N^{h, +})^2 P_{\frac{\pi}{4h}}^2 F (t', x) dt' \bigg\|_{L_x^2 ( \R ; L_t^\infty ([-T, T]))} \\
    &\leq \bigg\| \int_{-T}^T \sup_{t \in [-T, T]} | K_N^{h, +} (t - t', \cdot) | * |\mathbf{P}_N^{h, +} P_{\frac{\pi}{4h}}^2 F (t')| dt'  \bigg\|_{L_x^2 (\R)} \\
    &\les (1 + T)^{\frac 12} h^{- \frac 32} N^{\frac 32} \| F \|_{L_x^2 ( \R ; L_t^1 ([-T, T]))}.
\end{split}
\label{max_est}
\end{align}

\noi
Thus, a standard $TT^*$-argument along with \eqref{max_est} and the fact that $\mathbf{P}_N^{h, +} = \mathbf{P}_N^{h, +} (\mathbf{P}_{\frac N2}^{h, +} + \mathbf{P}_N^{h, +} + \mathbf{P}_{2N}^{h, +})$ (we take $\mathbf{P}_{2N}^{h, +} = 0$ if $N = 1$) on the frequency support $[-\frac{\pi}{4h}, \frac{\pi}{4h}]$ yields
\begin{align}
\begin{split}
\big\| &e^{- 6 h^{-2} t (\dh - \dx)} \mathbf{P}_N^{h, +} P_{\frac{\pi}{4h}} f \big\|_{L_x^2 (\R ; L_t^\infty ([-T, T]))} \\
&\leq \sum_{M = \frac{N}{2}, N, 2N} \big\| e^{- 6 h^{-2} t (\dh - \dx)} \mathbf{P}_N^{h, +} P_{\frac{\pi}{4h}} \mathbf{P}_M^{h, +} f \big\|_{L_x^2 (\R ; L_t^\infty ([-T, T]))} \\ 
&\les (1 + T)^{\frac 14} h^{- \frac 34} N^{\frac 34} \sum_{M = \frac{N}{2}, N, 2N} \| \mathbf{P}_M^{h, +} f \|_{L^2 (\R)} 
\end{split}
\label{maxN1}
\end{align}

\noi
when $h < N \leq 1$.
Similarly, by proceeding as in \eqref{max_est} but with \eqref{KN+2} in Lemma~\ref{LEM:KN}~(ii), we obtain
\begin{align}
    \bigg\| \sum_{\substack{N \leq h \\ \text{dyadic}}}  e^{- 6 h^{-2} t (\dh - \dx)} \mathbf{P}_N^{h, +} P_{\frac{\pi}{4h}} f \bigg\|_{L_x^2 (\R ; L_t^\infty ([-T, T]))} \les (1 + T) \| f \|_{L^2 (\R)} .
\label{maxN2}
\end{align}

\noi
Therefore, using \eqref{KNker2}, \eqref{maxN1}, \eqref{maxN2}, and a Cauchy-Schwarz inequality in dyadic $N$ with $h < N \leq 1$, we get
\begin{align*}
    \big\| &e^{- 6 h^{-2} t (\dh - \dx)} P_{\frac{\pi}{4h}} f \big\|_{L_x^2 (\R ; L_t^\infty ([-T, T]))} \\
    &\leq \bigg\| \sum_{\substack{N \leq h \\ \text{dyadic}}}  e^{- 6 h^{-2} t (\dh - \dx)} \mathbf{P}_N^{h, +} P_{\frac{\pi}{4h}} f \bigg\|_{L_x^2 (\R ; L_t^\infty ([-T, T]))} \\
    &\quad + \sum_{\substack{h < N \leq 1 \\ \text{dyadic}}} \big\| e^{- 6 h^{-2} t (\dh - \dx)} \mathbf{P}_N^{h, +} P_{\frac{\pi}{4h}} f \big\|_{L_x^2 (\R ; L_t^\infty ([-T, T]))} \\
    &\les (1 + T) \| f \|_{L^2 (\R)} + (1 + T)^{\frac 14} \sum_{\substack{h < N \leq 1 \\ \text{dyadic}}} h^{- \frac 34} N^{\frac 34} \sum_{M = \frac{N}{2}, N, 2N} \| \mathbf{P}_M^{h, +} f \|_{L^2 (\R^2)} \\
    &\les (1 + T) \| \jb{\dh}^s f \|_{L^2 (\R)}
\end{align*}

\noi
for any $s > \frac 34$, which is the desired estimate.

\smallskip \noi
(iv) The proof of the local smoothing estimates follows from minor modifications of \cite[Proposition~5.1~(ii)]{HKY} using change of variables and Plancherel's identity, which follows from the argument in \cite{CS}. The main observation is that the phase function 
$$\phi (\xi) = - 6 \ta h^{-3} (\sin (h \xi) - h \xi) + (1 - \ta) \xi^3$$ has derivative
$$\phi' (\xi) = - 6 \ta h^{-2} (\cos (h \xi) - 1) + 3 (1 - \ta) \xi^2,$$

\noi
which is positive for almost every $\xi \in [- \frac{\pi}{h}, \frac{\pi}{h})$. Also, from \eqref{der}, we have 
\begin{align*}
    \phi' (\xi) = 12 \ta h^{-2} \sin (\tfrac{h \xi}{2})^2 + 3 (1 - \ta) \xi^2 \geq \tfrac{12}{\pi^2} \ta \xi^2 + 3 (1 - \ta) \xi^2 \geq \tfrac{12}{\pi^2} \xi^2
\end{align*}

\noi
given $\xi \in [- \frac{\pi}{h}, \frac{\pi}{h})$ and $0 \leq \ta \leq 1$. We omit the rest of the steps and refer the readers to \cite[Proposition~5.1~(ii)]{HKY}.
\end{proof}

\begin{remark} \rm
The Strichartz estimates in Lemma~\ref{LEM:lin_d}~(ii) also hold in the discrete setting. Let $f^h$ be a function defined on $h \Z$ and let $4 \leq q \leq \infty$ and $2 \leq r \leq \infty$ be such that $\frac{2}{q} + \frac{1}{r} = \frac 12$. Then, we have
\begin{align*}
    \big\| |\dh|^{\frac 1q} e^{- 6 h^{-2} t \dh} f^h \big\|_{L_t^q (\R; L_\ld^r (h \Z))} \les \| f^h \|_{L^2 (h \Z)} .
\end{align*}

\noi
This can be obtained from using similar steps as in proof of Lemma~\ref{LEM:lin_d}~(ii). Also, for the Strichartz estimates, it is crucial to use $|\dh|$ instead of $|\dh^\pm|$ or $|\dx|$.

For the maximal function estimate in Lemma~\ref{LEM:lin_d}~(iii), we make the projection onto a smaller frequency set $[- \frac{\pi}{2 h}, \frac{\pi}{2h}]$. We remark that this projection captures the right going wave of the Toda dynamics.
\end{remark}

For later convenience, we will use the $X^{s,b}$-spaces, or the Fourier restriction norm method introduced by Bourgain in \cite{Bour93}.
Given $s, b \in \R$, we define the following $X^{s,b}$-norm for the Airy flow:
\begin{align*}
    \| u \|_{X^{s,b}} := \big\| e^{t \dx^3} u \big\|_{H^b_t H^s_x} = \| \jb{\xi}^s \jb{\tau - \xi^3}^b \ft u (\tau, \xi) \|_{L_\tau^2 L_\xi^2} .
\end{align*}

\noi
Given $0 < h \leq 1$ and $s, b \in \R$, we define the following $X^{s, b, h}$-norm for the linear propagator $e^{- 6 h^{-2} t (\dh - \dx)}$:
\begin{align*}
    \| u \|_{X^{s,b, h}} := \big\| e^{6 h^{-2} t (\dh - \dx)} u \big\|_{H^b_t H^s_x} = \| \jb{\xi}^s \jb{\tau + 6 h^{-3}  (\sin (h \xi) - h \xi)}^b \ft u (\tau, \xi) \|_{L_\tau^2 L_\xi^2} .
\end{align*}

\noi
For $X = X^{s, b}$ or $X^{s, b, h}$, given an interval $I \subset \R$, we define $X_I$ as the norm with the time restriction:
\begin{align*}
    \| u \|_{X_I} := \inf \{ \| v \|_{X}: u = v \text{ on } I \}.
\end{align*}

\noi
For convenience for later presentation, we write $X^{s, b, 0} = X^{s, b}$.

We first recall the following linear estimates. For a proof, see \cite{Bour93, GTV, Tao}.

\begin{lemma}
\label{LEM:Xsb_lin}
Let $0 \leq h \leq 1$ and $s \in \R$.

\smallskip \noi
\textup{(i)} For any $b \in \R$ and $I \subset \R$ a closed interval with $0 < |I| \leq 1$, we have
\begin{align*}
    \big\| e^{- 6 h^{-2} t (\dh - \dx)} f \big\|_{X_I^{s, b, h}} \les \| f \|_{H^s (\R)}.
\end{align*}

\smallskip \noi
\textup{(ii)} For any $\frac 12 < b < 1$, $t_0 \in \R$, and $I \subset \R$ a closed interval with $t_0 \in I$ and $0 < |I| \leq 1$, we have
\begin{align*}
    \bigg\| \int_{t_0}^t e^{- 6 h^{-2} (t - t') (\dh - \dx)} F (t') dt' \bigg\|_{X^{s, b, h}_{I}} \les |I|^{1 - b} \| F \|_{L_t^2 H_x^s (I \times \R)}.
\end{align*}
\end{lemma}

We will also need the continuity of the $X_T^{s, b, h}$-norm with respect to $T$. For a proof of the following lemma, see \cite[Lemma~2.7]{BLLZ}.

\begin{lemma}
\label{LEM:Xsb_cty}
Let $0 \leq h \leq 1$, $s, b \in \R$, $t_0 \in \R$, and $u \in X^{s, b, h}$. Then, the functions
\begin{align*}
    t \mapsto \| u \|_{X^{s, b, h}_{[t_0 - t, t_0]}} \quad \text{and} \quad t \mapsto \| u \|_{X^{s, b, h}_{[t_0, t_0 + t]}}
\end{align*}

\noi
are continuous on $[0, \infty)$.
\end{lemma}


We also have the following list of useful linear estimates in the context of $X^{s, b, h}$ and $X^{s, b}$-spaces.

\begin{lemma}
\label{LEM:Xsb_d}
Let $0 < h \leq 1$ and $b > \frac 12$. Let $u$ be a function defined on $\R \times \R$. 

\smallskip \noi
\textup{(i)} Let $s \in \R$. Then, we have
\begin{align*}
    \| u \|_{C_t (\R ; H_x^s (\R))} \les \| u \|_{X^{s, b, h}}  ,
\end{align*}

\noi
where the underlying constant is independent of $h$. Moreover, we have
\begin{align*}
    \| u \|_{C_t ( \R; H_x^s (\R))} \les \| u \|_{X^{s, b}} .
\end{align*}

\smallskip \noi
\textup{(ii)} Let $4 \leq q \leq \infty$ and $2 \leq r \leq \infty$ be such that $\frac{2}{q} + \frac{1}{r} = \frac 12$. Then, we have
\begin{align*}
    \big\| |\dh|^{\frac 1q} P_{\frac{\pi}{h}} u \big\|_{L^q_t (\R ; L^{r}_x (\R))} \les \| u \|_{X^{0, b, h}} ,
\end{align*}

\noi
where the underlying constant is independent of $h$. Moreover, we have
\begin{align*}
    \big\| |\dx|^{\frac 1q} u \big\|_{L^q_t (\R ; L^{r}_x (\R))} \les \| u \|_{X^{0, b}}.
\end{align*}

\smallskip \noi
\textup{(iii)} Let $s > \frac 34$ and $T > 0$. Then, we have
\begin{align*}
    \| P_{\frac{\pi}{4h}} u \|_{L_x^2 ( \R ; L_t^\infty ([-T, T]))} \les (1 + T) \| u \|_{X_T^{s, b, h}} ,
\end{align*}

\noi
where the underlying constant is independent of $h$. Moreover, we have
\begin{align*}
    \| u \|_{L_x^2 ( \R ; L_t^\infty ([-T, T]))} \les (1 + T)^{\frac 12} \| u \|_{X_T^{s, b}} .
\end{align*}

\smallskip \noi
\textup{(iv)} We have
\begin{align*}
    \| \dh^\pm P_{\frac{\pi}{h}} u \|_{L_x^\infty ( \R ; L_t^2 (\R))} \les \| u \|_{X^{0, b, h}} ,
\end{align*}

\noi
where the underlying constant is independent of $h$. The same estimate holds if $\dh^\pm$ is replaced by $\dx$. Moreover, we have
\begin{align*}
    \| \dx u \|_{L_x^\infty ( \R; L_t^2 (\R))} \les \| u \|_{X^{0, b}}.
\end{align*}
\end{lemma}

\begin{proof}
All the estimates follow from the linear estimates in Lemma~\ref{LEM:lin_d} and the transference principle (see \cite[Lemma~2.9]{Tao}). For the maximal function estimate in part (iii), we need the fact that $|h^{-1} \sin (h \xi)| \leq |\xi|$ for any $\xi \in \R$.
\end{proof}

\section{Conservation laws and uniform $H^1$ bound}
\label{SEC:cons}

In this section, we discuss conservation laws for the Toda lattice and establish uniform-in-time bounds for the solution of the Toda lattice.

\subsection{Construction of mass and energy}
\label{SUB:cons}

As mentioned in the introduction, the Toda lattice is a completely integrable system. 
Let us consider the system \eqref{toda_ab} for $\al$ and $\be$, which is equivalent to the following Lax equation:
\begin{align}
    \frac{d}{dt} L (t) = L(t) P(t) - P(t) L (t),
\label{Lax}
\end{align}

\noi
where the pair of operators $(L (t), P (t))$ is a Lax pair defined by
\begin{align*}
    &L(t) : f (n) \mapsto \frac 12 (1 + \al (t, n)) f (n + 1) + \frac 12 (1 + \al (t, n - 1)) f (n - 1) + \be (t, n) f (n) , \\
    &P(t) : f (n) \mapsto \frac 12 (1 + \al (t, n)) f (n + 1) - \frac 12 (1 + \al (t, n - 1)) f (n - 1) .
\end{align*}

\noi
The operator $L$ is called a Jacobi operator, which can be understood as the following tridiagonal matrix:
\begin{align*}
\begin{pmatrix}
\ddots & \ddots & \ddots & \\
\frac 12 (1 + \al (n - 2)) & \be (n - 1) & \frac 12 (1 + \al (n - 1)) & 0 & 0 \\
0 & \frac 12 (1 + \al (n - 1)) & \be (n) & \frac 12 (1 + \al (n)) & 0 \\
0 & 0 & \frac 12 (1 + \al (n)) & \be (n + 1) & \frac 12 (1 + \al (n + 1)) \\
 & & \ddots & \ddots & \ddots
\end{pmatrix} .
\end{align*}

\noi
The properties of the Jacobi operators are well studied in \cite{Tes}.

In order to rigorously derive the conserved quantities for the Toda lattice, we consider the following periodic Toda lattice consisting of $2N$ particles for a large $N \in \N$: for any $-N \leq n \leq N - 1$ and $k \in \Z$,
\begin{align}
\begin{cases}
    \dt \al_N (t, n) = - (1 + \al_N (t, n)) (\be_N (t, n + 1) - \be_N (t, n))  \\
    \dt \be_N (t, n) = - \big( 1 + \frac 12 \al_N (t, n) + \frac 12 \al_N (t, n - 1) \big) (\al_N (t, n) - \al_N (t, n - 1)) \\
    \al_N (t, n + 2 k N) = \al_N (t, n) \\
    \be_N (t, n + 2 k N) = \be_N (t, n) .
\end{cases}
\label{toda_abN}
\end{align}

\noi
The periodic Toda lattice \eqref{toda_abN} is a finite dimensional completely integrable system with the Lax operator $L_N$ given by the following $2N \times 2N$ Jacobi matrix:
\begin{align*}
\small
\begin{pmatrix}
    \be_N (-N) & \frac 12 (1 + \al_N (-N)) & 0 & \cdots & \frac 12 (1 + \al_N (N - 1)) \\
    \frac 12 (1 + \al_N (-N)) & \be_N (- N + 1) & \frac 12 (1 + \al_N (-N + 1)) & \cdots & 0 \\
    \vdots & \vdots & \ddots & \vdots & \vdots \\
    0 & \cdots & \frac 12 (1 + \al_N (N - 3)) & \be_N (N - 2) & \frac 12 (1 + \al_N (N - 2)) \\
    \frac 12 (1 + \al_N (N - 1)) & \cdots & 0 & \frac 12 (1 + \al_N (N - 2)) & \be_N (N - 1)
\end{pmatrix}.
\end{align*}

\noi
Then, we know that for any integer $k \geq 1$, the trace $\text{Tr} (L_N (t)^k)$ is preserved by the flow of \eqref{toda_abN}. From direct computations, we have the following expressions for $\text{Tr} (L_N (t)^k)$ with $k = 1, 2, 3$:
\begin{align*}
\text{Tr} (L_N (t)) &= \sum_{\substack{n \in \Z \\ -N \leq n \leq N - 1}} \be_N (t, n) , \\
\text{Tr} (L_N (t)^2) &= \sum_{\substack{n \in \Z \\ -N \leq n \leq N - 1}} \Big( \be_N (t, n)^2 + \frac 12 ( 1 + \al_N (t, n))^2 \Big) , \\
\text{Tr} (L_N (t)^3) &= \sum_{\substack{n \in \Z \\ -N \leq n \leq N - 1}} \Big( \be_N (t, n)^3 + \frac 34 (1 + \al_N (t, n))^2 ( \be_N (t, n) + \be_N (t, n + 1) ) \Big) .
\end{align*}

In order to take $N \to \infty$ and obtain conserved quantities of the Toda lattice on $\R$, we need to make some modifications to the above three traces so that they work for data and solutions in $\ell^2 (\Z)$. We first observe that the quantity
\begin{align*}
    A_N [\al_N (t)] := \sum_{\substack{n \in \Z \\ -N \leq n \leq N - 1}} \ln ( 1 + \al_N (t, n) ) 
\end{align*}

\noi
is also conserved under the flow of the truncated Toda lattice \eqref{toda_abN}. Then, we define the mass 
\begin{align}
\begin{split}
    &M_{N} [\al_N (t), \be_N (t)] \\
    &:= \text{Tr} (L_N (t)^2) - N - A_N [\al_N (t)] \\
    &\hphantom{:}= \sum_{\substack{n \in \Z \\ -N \leq n \leq N - 1}} \! \Big( \al_N (t, n)^2 + \be_N (t, n)^2 + \Big( - \frac 12 \al_N (t, n)^2 + \al_N (t, n) - \ln (1 + \al_N  (t, n)) \Big) \Big) 
\end{split}
\label{H1N}
\end{align}

\noi
and the energy
\begin{align}
\begin{split}
    &E_{N} [\al_N (t), \be_N (t)] \\
    &:= - \frac 13 \text{Tr} (L_N (t)^3) + \frac 12 \text{Tr} (L_N (t)) + \frac 12 M_N [\al_N (t), \be_N (t)] \\
    &\hphantom{:}=  \sum_{\substack{n \in \Z \\ -N \leq n \leq N - 1}} \! \Big( \frac 14 (\al_N (t, n) - \be_N (t, n))^2 + \frac 14 (\be_N (t, n + 1) - \al_N (t, n))^2 \\
    &\qquad - \frac 14 \al_N (t, n)^2 \be_N (t, n) - \frac 14 \al_N (t, n)^2 \be_N (t, n + 1) - \frac 13 \be_N (t, n)^3 - \frac 16 \al_N (t, n)^3 \\
    &\qquad + \frac 12 \Big( \frac 13 \al_N (t, n)^3 - \frac 12 \al_N (t, n)^2 + \al_N (t, n) - \ln (1 + \al_N (t, n)) \Big)\Big) ,
\end{split}
\label{H2N}
\end{align}

\noi
which are conserved quantities for the truncated Toda lattice \eqref{toda_abN}. Let us also define their formal limits
\begin{align}
    M [\al (t), \be (t)] := \sum_{n \in \Z} \Big( \al (t, n)^2 + \be (t, n)^2 + \Big( - \frac 12 \al (t, n)^2 + \al (t, n) - \ln ( 1 + \al (t, n) ) \Big) \Big)
\label{defH1-0}
\end{align}

\noi
and
\begin{align}
\begin{split}
    E [\al (t), \be (t)] &:= \sum_{n \in \Z} \Big( \frac 14 ( \al (t, n) - \be (t, n) )^2 + \frac 14 ( \be (t, n + 1) - \al (t, n) )^2 \\
    &\qquad - \frac 14 \al (t, n)^2 \be (t, n) - \frac 14 \al (t, n)^2 \be (t, n + 1) - \frac 13 \be (t, n)^3 - \frac 16 \al (t, n)^3  \\
    &\qquad + \frac 12 \Big( \frac 13 \al (t, n)^3 - \frac 12 \al (t, n)^2 + \al (t, n) - \ln ( 1 + \al (t, n) ) \Big) \Big) .
\end{split}
\label{defH2-0}
\end{align}


Let us now prove global well-posedness of the Toda lattice \eqref{toda_ab} in $\ell^2 (\Z)$ and also the conservation of mass $M$ and energy $E$ under the flow of the Toda lattice \eqref{toda_ab}.

\begin{proposition}
\label{PROP:cons}
Let $\al_0, \be_0 \in \ell^2 (\Z)$. Suppose that $\al_0 (n) > - 1$ for all $n \in \Z$ and $M [\al_0, \be_0] < \infty$. Then, there exists a unique solution $\al, \be \in C(\R; \ell^2 (\Z))$ to the Toda lattice \eqref{toda_ab} with $\al|_{t = 0} = \al_0$ and $\be |_{t = 0} = \be_0$. Moreover, we have the conservations
\begin{align*}
    M [\al (t), \be (t)] = M [\al_0, \be_0] \quad \text{and} \quad E [\al (t), \be (t)] = E [\al_0, \be_0]
\end{align*}

\noi
for any $t \in \R$.
\end{proposition}

\begin{proof}
Using Minkowski's integral inequality, H\"older's inequality, and the continuous embedding $\ell^2 (\Z) \subset \ell^4 (\Z)$, we obtain the existence of a unique local-in-time solution $\al, \be$ from a standard fixed-point argument in the space $C ([-T, T]; \ell^2 (\Z))$ for some $T = T (\| \al_0 \|_{\ell^2 (\Z)}, \| \be_0 \|_{\ell^2 (\Z)}) > 0$ on the following Duhamel formulation of \eqref{toda_ab}:
\begin{align*}
\begin{cases}
    \al (t, n) = \al_0 (n) - \int_0^t (1 + \al (t', n)) (\be (t', n + 1) - \be (t', n)) dt' , \\
    \be (t, n) = \be_0 (n) - \int_0^t \big( 1 + \tfrac 12 \al (t', n) + \tfrac 12 \al (t', n - 1)\big) (\al (t', n) - \al (t', n - 1)) dt'
\end{cases}
\end{align*}

\noi
for any $n \in \Z$ and $t \in [-T, T]$. Moreover, we have
\begin{align}
    \| \al \|_{C_{T} \ell^2_n (\Z)} + \| \be \|_{C_{T} \ell^2_n (\Z)} \leq 2 \| \al_0 \|_{\ell^2 (\Z)} + 2 \| \be_0 \|_{\ell^2 (\Z)} .
\label{ab_bdd}
\end{align}

We now prove the conservation of mass $M$ and energy $E$ for $t \in [-T, T]$. Given $N \in \N$, we define the $2N$-periodized version of $\al_0$ and $\be_0$ as
\begin{align}
\begin{split}
    \al_{0, N} (n + 2 k N) &\deff \al_0 (n) , \\
    \be_{0, N} (n + 2 k N) &\deff \be_0 (n) 
\end{split}
\label{def_ab0}
\end{align}

\noi
for any $-N \leq n \leq N - 1$ and $k \in \Z$. We let $\al_N$ and $\be_N$ be the unknowns for the periodic Toda lattice \eqref{toda_abN} with $\al_N |_{t = 0} = \al_{0, N}$ and $\be_N |_{t = 0} = \be_{0, N}$. Given $N \in \N$ and a $2N$-periodic function $a$ on $\Z$, we define the norm
\begin{align*}
    \| a \|_{\ell_N^2} := \bigg( \sum_{n = -N}^{N - 1} a (n)^2 \bigg)^{\frac 12} .
\end{align*}

\noi
Then, since $\| \al_{0, N} \|_{\ell_N^2 (\Z)} \leq \| \al_0 \|_{\ell^2 (\Z)}$ and $\| \be_{0, N} \|_{\ell_N^2 (\Z)} \leq \| \be_0 \|_{\ell^2 (\Z)}$, we can use the same fixed-point argument as above to obtain a unique solution $\al_N, \be_N \in C ([-T, T]; \ell_N^2 (\Z))$ satisfying the following Duhamel formulation of \eqref{toda_abN}:
\begin{align*}
\begin{cases}
    \al_N (t, n) = \al_{0, N} (n) - \int_0^t (1 + \al_N (t', n)) (\be_N (t', n + 1) - \be_N (t', n)) dt' \\
    \be_N (t, n) = \be_{0, N} (n) - \int_0^t \big( 1 + \tfrac 12 \al_N (t', n) + \tfrac 12 \al_N (t', n - 1)\big) (\al_N (t', n) - \al_N (t', n - 1)) dt'  \\
    \al_N (t, n + 2 k N) = \al_N (t, n) \\
    \be_N (t, n + 2 k N) = \be_N (t, n) 
\end{cases}
\end{align*}

\noi
for any $- N \leq n \leq N - 1$, $k \in \Z$, and $t \in [- T, T]$. Moreover, we have
\begin{align}
\begin{split}
    \| \al_N \|_{C_T \ell^2_N (\Z)} + \| \be_N \|_{C_T \ell^2_N (\Z)} &\leq 2 \| \al_{0, N} \|_{\ell_N^2 (\Z)} + 2 \| \be_{0, N} \|_{\ell_N^2 (\Z)} \\
    &\leq 2 \| \al_0 \|_{\ell^2 (\Z)} + 2 \| \be_0 \|_{\ell^2 (\Z)} .
\end{split}
\label{abN_bdd}
\end{align}

Let us now define 
\begin{align*}
    \wt \al_N (t, n) &:=
    \begin{cases}
        \al_N (t, n) & \text{if } -N \leq n \leq N - 1 \\
        0 & \text{else} 
    \end{cases} , \\
    \wt \be_N (t, n) &:=
    \begin{cases}
        \be_N (t, n) & \text{if } -N \leq n \leq N - 1 \\
        0 & \text{else} 
    \end{cases} ,
\end{align*}

\noi
and we define $\wt \al_{0, N} = \wt \al_N |_{t = 0}$ and $\wt \be_{0, N} = \wt \be_N |_{t = 0}$. We want to show that $\wt \al_N \to \al$ and $\wt \be_N \to \be$ in $C([-T, T]; \ell^2 (\Z))$ as $N \to \infty$ with a possibly smaller value of $T$. By taking the difference of the two Duhamel formulations and using Minkowski's integral inequality, H\"older's inequality, and the continuous embedding $\ell^2 (\Z) \subset \ell^4 (\Z)$, we obtain
\begin{align*}
    &\| \wt \al_N - \al \|_{L_{T}^\infty \ell^2_n (\Z)} + \| \wt \be_N - \be \|_{L_{T}^\infty \ell^2_n (\Z)} \\
    &\les \| \wt \al_{0, N} - \al_0 \|_{\ell^2 (\Z)} + \| \wt \be_{0, N} - \be_0 \|_{\ell^2 (\Z)} \\
    &\quad + T \big( \| \wt \al_N - \al \|_{L_{T}^\infty \ell^2_n (\Z)} + \| \wt \be_N - \be \|_{L_{T}^\infty \ell^2_n (\Z)} + \| \al (-N - 1) \|_{L_{T}^\infty} \\
    &\qquad + \| \al (N - 1) \|_{L_{T}^\infty} + \| \be (-N) \|_{L_{T}^\infty} + \| \be (N) \|_{L_{T}^\infty} \big) \\
    &\qquad \times  \big( 1 + \| \al_N \|_{L_{T}^\infty \ell^2_N (\Z)} + \| \al \|_{L_{T}^\infty \ell^2_n (\Z)} + \| \be_N \|_{L_{T}^\infty \ell^2_N (\Z)} + \| \be \|_{L_{T}^\infty \ell^2_n (\Z)} \big) .
\end{align*}

\noi
Thus, using \eqref{ab_bdd} and \eqref{abN_bdd} and taking $T = T (\| \al_0 \|_{\ell^2 (\Z)}, \| \be_0 \|_{\ell^2 (\Z)}) > 0$ to be smaller if necessary, we get
\begin{align}
\begin{split}
    \| &\wt \al_N - \al \|_{L_{T}^\infty \ell^2_n (\Z)} + \| \wt \be_N - \be \|_{L_{T}^\infty \ell^2_n (\Z)} \\
    &\les \| \wt \al_{0, N} - \al_0 \|_{\ell^2 (\Z)} + \| \wt \be_{0, N} - \be_0 \|_{\ell^2 (\Z)} + \big( \| \al (-N - 1) \|_{L_{T}^\infty} + \| \al (N - 1) \|_{L_{T}^\infty} \\
    &\quad + \| \be (-N) \|_{L_{T}^\infty} + \| \be (N) \|_{L_{T}^\infty} \big) \big( 1 + \| \al_0 \|_{\ell^2 (\Z)} + \| \be_0 \|_{\ell^2 (\Z)} \big) \\
    &\longrightarrow 0
\end{split}
\label{abN_conv}
\end{align}

\noi
as $N \to \infty$.

Recalling the (truncated) mass $M_{N}$ in \eqref{H1N} and the (truncated) energy $E_{N}$ in \eqref{H2N}, we know that 
\begin{align*}
    M_{N} [\al_N (t), \be_N (t)] = M_{N} [\al_{0, N}, \be_{0, N}] \quad \text{and} \quad E_{N} [\al_N (t), \be_N (t)] = E_{N} [\al_{0, N}, \be_{0, N}]
\end{align*}

\noi
for any $t \in [-T, T]$.
It remains to show the convergences $M_{N} [\al_N (t), \be_N (t)] \to M [\al (t), \be (t)]$ and $E_{N} [\al_N (t), \be_N (t)] \to E [\al (t), \be (t)]$ as $N \to \infty$ for any $t \in [-T, T]$. To achieve this, we only need to show the following convergences as $N \to \infty$:
\begin{align}
&\bigg| \sum_{n \in \Z} \be (t, n)^2 - \sum_{\substack{n \in \Z \\ -N \leq n \leq N - 1}} \be_N (t, n)^2 \bigg| + \bigg| \sum_{n \in \Z} \al (t, n)^2 - \sum_{\substack{n \in \Z \\ -N \leq n \leq N - 1}} \al_N (t, n)^2 \bigg| \longrightarrow 0, \label{est1} \\
&\bigg| \sum_{n \in \Z} \big( \al (t, n) - \ln (1 + \al (t, n) ) \big) - \sum_{\substack{n \in \Z \\ -N \leq n \leq N - 1}} \big( \al_N (t, n) - \ln (1 + \al_N (t, n) ) \big) \bigg| \longrightarrow 0, \label{est2} \\
&\bigg| \sum_{n \in \Z} \al (t, n) \be (t, n) - \sum_{\substack{n \in \Z \\ -N \leq n \leq N - 1}} \al_N (t, n) \be_N (t, n) \bigg|  \longrightarrow 0, \label{est3} \\
&\bigg| \sum_{n \in \Z} \al (t, n) \be (t, n + 1) - \sum_{\substack{n \in \Z \\ -N \leq n \leq N - 1}} \al_N (t, n) \be_N (t, n + 1) \bigg|  \longrightarrow 0, \label{est4} \\
&\bigg| \sum_{n \in \Z} \al (t, n)^2 \be (t, n) - \sum_{\substack{n \in \Z \\ -N \leq n \leq N - 1}} \al_N (t, n)^2 \be_N (t, n) \bigg|  \longrightarrow 0, \label{est5} \\
&\bigg| \sum_{n \in \Z} \al (t, n)^2 \be (t, n + 1) - \sum_{\substack{n \in \Z \\ -N \leq n \leq N - 1}} \al_N (t, n)^2 \be_N (t, n + 1) \bigg|  \longrightarrow 0, \label{est6} \\
&\bigg| \sum_{n \in \Z} \be (t, n)^3 - \sum_{\substack{n \in \Z \\ -N \leq n \leq N - 1}} \be_N (t, n)^3 \bigg|  \longrightarrow 0, \label{est7} 
\end{align}

\noi
The power type convergences \eqref{est1}, \eqref{est3}, \eqref{est4}, \eqref{est5}, \eqref{est6}, and \eqref{est7} follow immediately from H\"older's inequalities, the continuous embedding $\ell^2 (\Z) \subset \ell^3 (\Z)$, the convergence in \eqref{abN_conv}, and the uniform bounds \eqref{ab_bdd} and \eqref{abN_bdd}. For \eqref{est2}, by using the Taylor expansion of $\ln (1 + x)$, H\"older's inequality, the continuous embedding $\ell^2 (\Z) \subset \ell^{k} (\Z)$ for any $k \geq 2$, 
and the convergence in \eqref{abN_conv}, we have
\begin{align*}
\bigg| &\sum_{n \in \Z} \big( \al (t, n) - \ln (1 + \al (t, n) ) \big) - \sum_{\substack{n \in \Z \\ -N \leq n \leq N - 1}} \big( \al_N (t, n) - \ln (1 + \al_N (t, n) ) \big) \bigg| \\
&\leq \sum_{k = 2}^\infty \frac{1}{k} \sum_{\substack{n \in \Z \\ - N \leq n \leq N - 1}} | \al (t, n) - \al_N (t, n) |^k + \sum_{k = 2}^\infty \frac{1}{k} \sum_{\substack{n \in \Z \\ |n| \geq N}} |\al (t, n)|^k  \\
&\leq \sum_{k = 2}^\infty \frac{1}{k} \| \al (t) - \wt \al_N (t) \|_{\ell^2 (\Z)}^k + \sum_{k = 2}^\infty \frac{1}{k} \bigg( \sum_{\substack{n \in \Z \\ |n| \geq N}} |\al (t, n)|^2 \bigg)^{\frac{k}{2}} \\
&\longrightarrow 0
\end{align*}

\noi
as $N \to \infty$, so that \eqref{est2} follows. This shows the conservation of mass $M$ and energy $E$ for $t \in [-T, T]$


We are now ready to prove global well-posedness of the Toda lattice \eqref{toda_ab} in $\ell^2 (\Z)$. By using the fact that $\ln (1 + x) \leq x$ for any $x > - 1$ and the conservation of $M$, we have
\begin{align}
    \frac 12 \| \al (t) \|_{\ell^2 (\Z)}^2 + \| \be (t) \|_{\ell^2 (\Z)}^2 \leq M [ \al (t), \be (t) ] = M [\al_0, \be_0] < \infty .
\label{ab_bdd1}
\end{align}

\noi
for any $t \in [-T, T]$. This bound allows us to iterate the local well-posedness argument and extend the solution to the whole time line. The conservation of mass $M$ and energy $E$ also follows similarly as above.
\end{proof}

\begin{remark} \rm
\label{RMK:cons}
In Proposition~\ref{PROP:cons}, we are able to prove global well-posedness of the Toda lattice \eqref{toda_ab} by assuming that $\al_0 (n) > -1$ for all $n \in \Z$ and that $M [\al_0, \be_0]$ is finite. In particular, these conditions hold if we assume that $\al_0, \be_0 \in \ell^2 (\Z)$ and $\| \al_0 \|_{\ell^\infty (\Z)} \leq \frac 12$. Indeed, using the Taylor expansion of $\ln (1 + x)$ with $|x| < 1$, we obtain
\begin{align}
\begin{split}
    M [\al_0, \be_0] &\leq \sum_{n \in \Z} \bigg( \be_0 (n)^2 + \al_0 (n)^2 + \sum_{k = 3}^\infty \frac{1}{k} |\al_0 (n)|^k \bigg) \\
    &\leq \| \al_0 \|_{\ell^2 (\Z)}^2 + \| \be_0 \|_{\ell^2 (\Z)}^2 +  \sum_{n \in \Z} |\alpha_0 (n)|^3  \\
    &\leq \| \al_0 \|_{\ell^2 (\Z)}^2 + \| \be_0 \|_{\ell^2 (\Z)}^2 + \sum_{n \in \Z} |\alpha_0 (n)|^2 \\
    &\leq 2 \| \al_0 \|_{\ell^2 (\Z)}^2 + \| \be_0 \|_{\ell^2 (\Z)}^2 ,
\end{split}
\label{ab_bdd2}
\end{align}

\noi
which is finite given $\alpha_0, \beta_0 \in \ell^2 (\Z)$.
\end{remark}

\subsection{Uniform $H^1$ bound}
\label{SUB:H1}

We now use the conserved quantities in the previous subsection to establish some uniform bounds for $\gamma^h$ defined in \eqref{defgh} satisfying the equation \eqref{todah}. 

Recalling the definition of $\al^h$ and $\be^h$ in \eqref{def_abh}, we define the scaled versions of the mass $M$ in \eqref{defH1-0} and the energy $E$ in \eqref{defH2-0} as
\begin{align}
\begin{split}
    M^h [\al^h (t), \be^h (t)] &:= \sum_{\ld \in h \Z} \Big( h \al^h (t, \ld)^2 + h \be^h (t, \ld)^2 \\
    &\qquad + h^{-3} \Big( - \frac 12 h^4 \al^h (t, \ld)^2 + h^2 \al^h (t, \ld) - \ln ( h^2 \al^h (t, \ld) + 1 ) \Big) \Big),
\end{split}
\label{defH1}
\end{align}
\begin{align}
\begin{split}
    &E^h [\al^h (t), \be^h (t)] \\
    &:= \sum_{\ld \in h \Z} \Big( \frac 14 h^{-1} (\al^h (t, \ld) - \be^h (t, \ld))^2 + \frac 14 h^{-1} ( \be^h (t, \ld + h) - \al^h (t, \ld) )^2 \\
    &\qquad - \frac 14 h \al^h (t, \ld)^2 \be^h (t, \ld) - \frac 14 h \al^h (t, \ld)^2 \be^h (t, \ld + h) - \frac 13 h \be^h (t, \ld)^3-\frac16 h \al^h(t,\ld)^3  \\
    &\qquad + \frac 12 h^{-5} \Big( \frac13 h^6 \al^h(t,\ld)^3 - \frac12 h^4 \al^h (t, \ld)^2 +  h^2 \al^h (t, \ld) -  \ln ( h^2 \al^h (t, \ld) + 1 ) \Big) \Big) .
\end{split}
\label{defH2}
\end{align}

\noi
For simplicity, we also write $M^h (t) = M^h [\al^h (t), \be^h (t)]$ and $E^h (t) = E^h [\al^h (t), \be^h (t)]$.

\begin{remark} \rm
Let us disregard the remainder parts of $M^h$ and $E^h$ and consider the main parts of them:
\begin{align*}
    M_{\text{main}}^h (t) &:= \sum_{\ld \in h \Z} \Big( h \al^h (t, \ld)^2 + h \be^h (t, \ld)^2 \Big), \\
    E_{\text{main}}^h (t) &:= \sum_{\ld \in h \Z} \Big( \frac 14 h^{-1} (\al^h (t, \ld) - \be^h (t, \ld))^2 + \frac 14 h^{-1} ( \be^h (t, \ld + h) - \al^h (t, \ld) )^2 \\
    &\qquad - \frac 14 h \al^h (t, \ld)^2 \be^h (t, \ld) - \frac 14 h \al^h (t, \ld)^2 \be^h (t, \ld + h) - \frac 13 h \be^h (t, \ld)^3-\frac16 h \al^h(t,\ld)^3 \Big) .
\end{align*}

\noi
Note that $M_{\text{main}}^h$ and $E_{\text{main}}^h$ corresponds exactly to the discrete versions of 
\begin{align*}
    \int_{\R} u^2 dx \quad \text{and} \quad \frac 14 \int_\R (\dx u)^2 dx - \frac 12 \int_\R u^3 dx ,
\end{align*}

\noi
respectively, which are the mass and the energy for the KdV equation \eqref{KdV}. This provides another heuristic on the continuum limit of the Toda lattice to the KdV equation.
\end{remark}


From Proposition~\ref{PROP:cons} and Remark~\ref{RMK:cons}, we know that the solution $\gamma^h$ to \eqref{todah} exists in $C (\R; L^2 (h \Z))$, as long as the scaled initial data $\gamma_0^h$ satisfies $h^2 \| \gamma_0^h \|_{L^\infty (h \Z)} \leq \frac 12$. 


We first show the following $L^2 (h \Z)$-bound for $\gamma^h$.

\begin{proposition}
\label{PROP:L2bdd}
Let $0 < h \leq 1$ and $\gamma_0^h \in L^2 (h \Z)$ be such that 
\begin{align}
    h^2 \| \gamma_0^h \|_{L^\infty (h \Z)} \leq \frac 12 .
\label{h_cond2}
\end{align} 

\noi
Let $\gamma^h \in C(\R; L^2 (h \Z))$ be the solution to the scaled Toda lattice \eqref{todah} with initial data $\gamma^h |_{t = 0} = \gamma_0^h$.
Then, for any $t \in \R$, we have
\begin{align}
    \| \gamma^h (t) \|_{L^2 (h \Z)} \leq 2 \| \gamma_0^h \|_{L^2 (h \Z)} .
\label{L2bdd_goal}
\end{align}
\end{proposition}

\begin{proof}
The $L^2 (h \Z)$-bound for $\gamma^h (t)$ follows immediately from the bound \eqref{ab_bdd1} in the proof of Proposition~\ref{PROP:cons} and the bound \eqref{ab_bdd2} in Remark~\ref{RMK:cons}, as long as we have the condition \eqref{h_cond2}. 
\end{proof}

We now establish the following $H^1 (h \Z)$-bound for the solution $\gamma^h$ to the scaled Toda lattice \eqref{todah}. 
\begin{proposition}
\label{PROP:H1bdd}
Let $0 < h \leq 1$ and $\gamma_0^h \in H^1 (h \Z)$ be such that the condition \eqref{h_cond} holds. Let $\gamma^h \in C(\R; L^2 (h\Z))$ be the solution to the scaled Toda lattice \eqref{todah} with initial data $\gamma^h|_{t = 0} = \gamma_0^h$.  Then, for any $t \in \R$, we have
\begin{align}
    \| \gamma^h (t) \|_{\dot{H}^1 (h \Z)} \les \| \gamma_0^h \|_{\dot{H}^1 (h \Z)} + \| \gamma_0^h \|_{L^2 (h \Z)}^{\frac 53} ,
\label{H1bdd_goal}
\end{align}

\noi
where the underlying constant is independent of $h$.
\end{proposition}

\begin{proof}
From Lemma~\ref{LEM:Lpbdd} and \eqref{h_cond}, we know that
\begin{align}
    h^2 \| \gamma_0^h \|_{L^\infty (h \Z)} \leq h^{\frac 32} \| \gamma_0^h \|_{L^2 (h \Z)} \leq \frac 14  ,
\label{sob}
\end{align} 
Thus, from Proposition~\ref{PROP:cons} and Remark~\ref{RMK:cons}, we have the conservation of energy $E^h (t) = E^h (0)$ for any $t \in \R$.

From \eqref{H1norm}, \eqref{defH2}, \eqref{defgh}, the Taylor expansion of $\ln (1 + x)$, Young's inequalities, and Lemma~\ref{LEM:Lpbdd}, we see that for any $t \in \R$, we have
\begin{align}
\begin{split}
    \Big| &E^h (t) - \frac 14 \| \gamma^h (t) \|_{\dot{H}^1 (h \Z)}^2 \Big| \\
    &\les \sum_{\ld \in h \Z} \Big( h | \gamma^h (t, 2 \ld)^2 \gamma^h (t, 2 \ld - h) | + h |  \gamma^h (t, 2 \ld)^2 \gamma^h (t, 2 \ld + h) | \\
    &\qquad + h |\gamma^h (t, 2 \ld - h)|^3 + h |\gamma^h (t, 2 \ld)|^3 + \sum_{k = 4}^\infty h^{2k - 5} |\al^h (t, \ld)|^k \Big) \\
    &\les \sum_{\ld \in h \Z} h | \gamma^h (t, \ld) |^3 + \sum_{k = 4}^\infty h^{2k - 6} \| \gamma^h (t) \|_{L^k (h \Z)}^k \\
    &\les \| \gamma^h (t) \|_{L^3 (h \Z)}^3 + \sum_{k = 4}^\infty h^{\frac 32 k - 5} \| \gamma^h (t) \|_{L^2 (h \Z)}^k .
\end{split}
\label{H2diff1}
\end{align}

\noi
By the Gagliardo-Nirenberg inequality (Lemma~\ref{LEM:GN}), Young's inequality, and Proposition~\ref{PROP:L2bdd}, we have
\begin{align}
\begin{split}
    \| \gamma^h (t) \|_{L^3 (h \Z)}^3 &\les \| \gamma^h (t) \|_{L^2 (h \Z)}^{\frac 52} \| \gamma^h (t) \|_{\dot{H}^1 (h \Z)}^{\frac 12} \\
    &\les \dl \| \gamma^h (t) \|_{\dot{H}^1 (h \Z)}^2 + \dl^{-1} \| \gamma^h (t) \|_{L^2 (h \Z)}^{\frac{10}{3}} \\
    &\les \dl \| \gamma^h (t) \|_{\dot{H}^1 (h \Z)}^2 + \dl^{-1} \| \gamma_0^h \|_{L^2 (h \Z)}^{\frac{10}{3}} 
\end{split}
\label{H2diff2}
\end{align}

\noi
for any $\dl > 0$. Also, by Proposition~\ref{PROP:L2bdd} and the condition \eqref{h_cond}, we have 
\begin{align}
\begin{split}
    \sum_{k = 4}^\infty h^{\frac 32 k - 5} \| \gamma^h (t) \|_{L^2 (h \Z)}^k &\leq \sum_{k = 4}^\infty 2^k h^{\frac 32 k - 5} \| \gamma_0^h \|_{L^2 (h \Z)}^k \\
    &\les h \| \gamma_0^h \|_{L^2 (h \Z)}^4 \\
    &\les \| \gamma_0^h \|_{L^2 (h \Z)}^{\frac{10}{3}} .
\end{split}
\label{H2diff3}
\end{align}

\noi
Thus, from \eqref{H2diff1}, \eqref{H2diff2}, and \eqref{H2diff3}, we obtain that for any $t \in \R$,
\begin{align}
    \Big| E^h (t) - \frac 14 \| \gamma^h (t) \|_{\dot{H}^1 (h \Z)}^2 \Big| \leq \frac{1}{8} \| \gamma^h (t) \|_{\dot{H}^1 (h \Z)}^2 + C \| \gamma_0^h \|_{L^2 (h \Z)}^{\frac{10}{3}} 
\label{H2diff}
\end{align}

\noi
for some constant $C > 0$. Therefore, from \eqref{H2diff} and the conservation of $E^h$, for any $t \in \R$, we have
\begin{align*}
    \frac 14 \| \gamma^h (t) \|_{\dot{H}^1 (h \Z)}^2 &\leq  \Big| E^h (t) - \frac 14 \| \gamma^h (t) \|_{\dot{H}^1 (h \Z)}^2 \Big| + |E^h (t)| \\
    &\leq \frac{1}{8} \| \gamma^h (t) \|_{\dot{H}^1 (h \Z)}^2 + C \| \gamma_0^h \|_{L^2 (h \Z)}^{\frac{10}{3}} + |E^h (0)| \\
    &\leq \frac{1}{8} \| \gamma^h (t) \|_{\dot{H}^1 (h \Z)}^2 + \frac 38 \| \gamma_0^h \|_{\dot{H}^1 (h \Z)}^2 + 2 C \| \gamma_0^h \|_{L^2 (h \Z)}^{\frac{10}{3}}.
\end{align*}

\noi
This gives the desired estimate.
\end{proof}

Note that from \eqref{sob}, the condition \eqref{h_cond} on $h$ is stronger than the condition \eqref{h_cond2} in Proposition~\ref{PROP:L2bdd}. Thus, in the rest of this paper, we work on the condition \eqref{h_cond}.

Gathering the a priori bounds in this section, we are able to prove global well-posedness of the scaled Toda lattice \eqref{todah} in $H^1 (h \Z)$.

\begin{proof}[Proof of Theorem~\ref{THM:conv}~\textup{(i)}]
We recall the assumption $\gamma_0^h \in H^1 (h \Z)$ and the condition \eqref{h_cond}. Thus, from Proposition~\ref{PROP:L2bdd} and Proposition~\ref{PROP:H1bdd}, we have the a priori bound
\begin{align*}
    \| \gamma^h (t) \|_{H^1 (h \Z)} \les \| \gamma^h (t) \|_{L^2 (h \Z)} + \| \gamma^h (t) \|_{\dot{H}^1 (h \Z)} \les \| \gamma_0^h \|_{H^1 (h \Z)} + \| \gamma_0^h \|_{L^2 (h \Z)}^{\frac 53}
\end{align*}

\noi
for any $t \in \R$ such that $\gamma^h (t)$ is already constructed. This allows us to iterate the Picard iteration scheme and construct the global-in-time solution to the scaled Toda lattice \eqref{todah}.
\end{proof}

\section{Continuum limit and long-wave limits}
\label{SEC:conv}

Our goal in this section is to prove the continuum limit and long-wave limits of the scaled Toda lattice to the KdV equation, as stated in Theorem~\ref{THM:conv}~(ii), Theorem~\ref{THM:lw}, and Theorem~\ref{THM:long-wave}.

\subsection{A priori bounds for solutions}
\label{SUB:bdd}

Let us consider the scaled Toda lattice in Flaschka's form \eqref{todah}. We first write out the Duhamel formulation of the scaled Toda lattice \eqref{todah}:
\begin{align}
    \gamma^h (t) = e^{- 6 h^{-2} t \dh} \gamma^h_0  - \int_0^t e^{- 6 h^{-2} (t - t') \dh} \mathcal{N}^h [\gamma^h] (t') dt' ,
\label{duh_g}
\end{align}

\noi
where the nonlinearity $\mathcal{N}^h [\gamma^h]$ is given by \eqref{defNh} and can be written as
\begin{align}
\mathcal{N}^h [\gamma^h] (t) = 3 \wt{\gamma}^h (t) \cdot h^{-1} ( \gamma^h (t, \cdot + h) - \gamma^h (t, \cdot - h) ) =  6 \wt{\gamma}^h (t) \dh \gamma^h (t)
\label{Nh}
\end{align}

\noi
with $\wt{\gamma}^h$ being defined as
\begin{align}
\wt{\gamma}^h (t, \ld) =
\begin{cases}
    \gamma^h (t, \ld)  &  \text{if $\dfrac{\ld}{h}$ is even} \vspace{5pt} \\
    \dfrac 12 ( \gamma^h (t, \ld + h) + \gamma^h (t, \ld - h) )  &  \text{if $\dfrac{\ld}{h}$ is odd} .
\end{cases}
\label{defgh2}
\end{align} 

\noi
From Theorem~\ref{THM:conv}~(i), we know that \eqref{duh_g} has a unique solution in $C(\R; H^1 (h \Z))$. 
From \eqref{defuh} and \eqref{duh_g}, we have
\begin{align}
    u^h (t) = e^{- 6 h^{-2} t (\dh - \partial_x)} \mathcal{E} \gamma_0^h - 6 \int_0^t e^{- 6 h^{-2} (t - t') (\dh - \dx)} e^{6 h^{-2} t' \dx} \mathcal{E} (\wt \gamma^h \dh \gamma^h) (t') dt' ,
\label{eq_uh}
\end{align}

\noi
where we recall that $\mathcal{E}$ is the extension operator in \eqref{extend}. We also write out the Duhamel formulation of the KdV equation \eqref{KdV}:
\begin{align}
    u (t) = e^{- t \partial_x^3} u_0 - 6 \int_0^t e^{- (t - t') \partial_x^3} (u \partial_x u) (t') dt'.
\label{duh_kdv}
\end{align}

In order to prove global-in-time convergence of the solutions, we need to establish some a priori bounds of $u^h$ and $u$ in certain function spaces.

Let us first show the following lemma on the estimate of the quadratic term on the right-hand-side of \eqref{eq_uh}.
\begin{lemma}
\label{LEM:quad}
Let $0 < h \leq 1$ and $\gamma_0^h \in H^1 (h \Z)$ be such that the condition \eqref{h_cond} holds. Let $\gamma^h$ be the solution to \eqref{duh_g} with initial data $\gamma_0^h$, $\wt \gamma^h$ be given by \eqref{defgh2}, and $u^h$ be given by \eqref{defuh}. Let $I \subset \R$ be a closed interval. Then, we have
\begin{align}
    \| \mathcal{E} (\wt \gamma^h \dh \gamma^h) (t) \|_{H^1 (\R)} \les \| \mathcal{E} \wt \gamma^h (t) \mathcal{E} (\dh \gamma^h) (t) \|_{H^1 (\R)}
\label{quad01}
\end{align}

\noi
for any $t \in I$ and
\begin{align}
\begin{split}
    \| \mathcal{E} \wt \gamma^h \mathcal{E} (\dh \gamma^h) \|_{L_t^2 (I ; H_x^1 (\R))} 
    &\les |I|^{\frac 12} \| u^h \|_{L_t^\infty (I ; H_x^1 (\R))}^2 + |I|^{\frac 14} \| u^h \|_{L_t^\infty (I ; H_x^1 (\R))}  \| \dh u^h \|_{L_t^4 (I; L_x^\infty (\R))} \\
    &\quad + \| P_{\frac{\pi}{4h}} u^h \|_{L_x^2 (\R ; L_t^\infty (I))} \| \dh \dx u^h \|_{L_x^\infty (\R; L_t^2 (I))} \\
    &\quad + |I|^{\frac 14} \Big( \| \gamma_0^h \|_{H^1 (h \Z)} + \| \gamma_0^h \|_{L^2 (h \Z)}^{\frac 53} \Big) \| \dh u^h \|_{L_t^4 (I; L_x^\infty (\R))} \\
    &\quad + |I|^{\frac 12} \Big( \| \gamma_0^h \|_{H^1 (h \Z)} + \| \gamma_0^h \|_{L^2 (h \Z)}^{\frac 53} \Big) \| u^h \|_{L_t^\infty (I; H_x^1 (\R))} .
\end{split}
\label{quad02}
\end{align}
\end{lemma}

\begin{proof}
We first consider \eqref{quad01}. Let $t \in I$. From Lemma~\ref{LEM:Eprod}, we have
\begin{align}
\begin{split}
    \| &\mathcal{E} (\wt \gamma^h \dh \gamma^h) (t) \|_{H^1 (\R)} \\
    &\leq \| \mathcal{E} \wt \gamma^h (t) \mathcal{E} (\dh \gamma^h) (t) \|_{H^1 (\R)} + \big\| P_{\frac{\pi}{h}} \big( \cos (\tfrac{2\pi}{h} \cdot) \mathcal{E} \wt \gamma^h (t) \mathcal{E} (\dh \gamma^h) (t) \big) \big\|_{H^1 (\R)} .
\end{split}
\label{quad1-1}
\end{align}

\noi
For the second term on the right-hand-side of \eqref{quad1-1}, we use Lemma~\ref{LEM:Ecos} to compute that
\begin{align}
\begin{split}
    \big\| &P_{\frac{\pi}{h}} \big( \cos (\tfrac{2\pi}{h} \cdot) \mathcal{E} \wt \gamma^h (t) \mathcal{E} (\dh \gamma^h) (t) \big) \big\|_{H^1 (\R)} \\
    &\les  \big\| P_{\frac{\pi}{h}} \big( \cos (\tfrac{2 \pi}{h} \cdot)\mathcal{E} \wt \gamma^h (t) \mathcal{E} (\dh \gamma^h) (t) \big) \big\|_{L^2 (\R)} \\
    &\quad + \big\| P_{\frac{\pi}{h}} \big( \cos (\tfrac{2 \pi}{h} \cdot) \dx ( \mathcal{E} \wt \gamma^h (t) \mathcal{E} (\dh \gamma^h) (t) ) \big) \big\|_{L^2 (\R)} \\
    &\quad + \big\| P_{\frac{\pi}{h}} \big( \dx \cos (\tfrac{2 \pi}{h} \cdot) \mathcal{E} \wt \gamma^h (t) \mathcal{E} (\dh \gamma^h) (t) \big) \big\|_{L^2 (\R)} \\
    &\les \| \mathcal{E} \wt \gamma^h (t) \mathcal{E} (\dh \gamma^h) (t) \|_{H^1 (\R)} + h^{-1} \big\| P_{\frac{\pi}{h}} \big( \sin (\tfrac{2 \pi}{h} \cdot) \mathcal{E} \wt \gamma^h (t) \mathcal{E} (\dh \gamma^h) (t) \big) \big\|_{L^2 (\R)} \\
    &\les \| \mathcal{E} \wt \gamma^h (t) \mathcal{E} (\dh \gamma^h) (t) \|_{H^1 (\R)} .
\end{split}
\label{quad1-2}
\end{align}

\noi
The desired estimate \eqref{quad01} then follows from \eqref{quad1-1} and \eqref{quad1-2}.

We now consider \eqref{quad02}. Using a spatial translation and \eqref{defuh}, we get
\begin{align}
\begin{split}
    \| & \mathcal{E} \wt \gamma^h \mathcal{E} (\dh \gamma^h) \|_{L_t^2 (I ; H_x^1 (\R))} \\
    &= \big\| \mathcal{E} \wt \gamma^h (\cdot + 6 h^{-2} t) \dh \mathcal{E} \gamma^h (\cdot + 6 h^{-2} t) \big\|_{L_t^2 (I ; H_x^1 (\R))} \\
    &\les \| \dx \mathcal{E} \wt \gamma^h (\cdot + 6 h^{-2} t) \dh u^h \|_{L_t^2 (I; L_x^2 (\R))} + \| \mathcal{E} \wt \gamma^h (\cdot + 6 h^{-2} t) \dx \dh u^h \|_{L_t^2 (I ; L_x^2 (\R))} \\
    &\quad + \| \mathcal{E} \wt \gamma^h (\cdot + 6 h^{-2} t) \dh u^h \|_{L_t^2 (I ; L_x^2 (\R))} \\
    &=: \textup{I}_1 + \textup{I}_2 + \textup{I}_3 .
\end{split}
\label{quad1}
\end{align}

\noi
For $\textup{I}_1$, by H\"older's inequalities and \eqref{defuh}, we have
\begin{align}
\begin{split}
    \textup{I}_1 &\leq |I|^{\frac 14} \| \dx \mathcal{E} \wt \gamma^h (\cdot + 6 h^{-2} t) \|_{L_t^\infty (I ; L_x^2 (\R))} \| \dh u^h \|_{L_t^4 (I ; L_x^\infty (\R))} \\
    &\leq |I|^{\frac 14} \big( \| u^h \|_{L_t^\infty (I ; H_x^1 (\R))} + \| \mathcal{E} \wt \gamma^h - \mathcal{E} \gamma^h \|_{L_t^\infty (I ; H_x^1 (\R))} \big) \| \dh u^h \|_{L_t^4 (I ; L_x^\infty (\R))} .
\end{split}
\label{quad2-1}
\end{align}

\noi
For any $t \in I$, by the fact that $\mathcal{E} \wt \gamma^h$ and $\mathcal{E} \gamma^h$ have Fourier support on $[- \frac{\pi}{h}, \frac{\pi}{h}]$, Lemma~\ref{LEM:Hsbdd}, \eqref{H1norm}, and Proposition~\ref{PROP:H1bdd}, we have
\begin{align}
\begin{split}
    \| \mathcal{E} \wt \gamma^h (t) - \mathcal{E} \gamma^h (t) \|_{H_x^1 (\R)} &\les h^{-1} \| \mathcal{E} \wt \gamma^h (t) - \mathcal{E} \gamma^h (t) \|_{L_x^2 (\R)} \\
    &\leq h^{-1} \| \wt \gamma^h (t) - \gamma^h (t) \|_{L^2 (h \Z)} \\
    &\les \bigg( h^{-1} \sum_{\ld \in h \Z} |\gamma^h (t, \ld + h) - \gamma^h (t, \ld)|^2 \bigg)^{\frac 12} \\
    &= \| \gamma^h (t) \|_{\dot{H}^1 (h \Z)} \\
    &\les \| \gamma_0^h \|_{\dot{H}^1 (h \Z)} + \| \gamma_0^h \|_{L^2 (h \Z)}^{\frac 53}.
\end{split}
\label{quad2-2}
\end{align}

\noi
Combining \eqref{quad2-1} and \eqref{quad2-2}, we get
\begin{align}
\begin{split}
    \textup{I}_1 &\les |I|^{\frac 14} \| u^h \|_{L_t^\infty (I ; H_x^1 (\R))}  \| \dh u^h \|_{L_t^4 (I; L_x^\infty (\R))} \\
    &\quad + |I|^{\frac 14} \Big( \| \gamma_0^h \|_{H^1 (h \Z)} + \| \gamma_0^h \|_{L^2 (h \Z)}^{\frac 53} \Big) \| \dh u^h \|_{L_t^4 (I; L_x^\infty (\R))} .
\end{split}
\label{quad2}
\end{align}

\noi
For $\textup{I}_2$, by \eqref{defuh}, H\"older's inequalities, and Bernstein's inequality (see \cite[Page~333]{Tao}) along with the fact that $u^h$ has Fourier support $[-\frac{\pi}{h}, \frac{\pi}{h}]$, we have
\begin{align}
\begin{split}
    \textup{I}_2 &\leq \| P_{\frac{\pi}{4h}} u^h \dx \dh u^h \|_{L_x^2 (\R ; L_t^2 (I))} + \| P_{\frac{\pi}{4h}}^\perp u^h \dx \dh u^h \|_{L_t^2 (I ; L_x^2 (\R))} \\
    &\quad + \| (\mathcal{E} \wt \gamma^h (\cdot + 6 h^{-2} t) - \mathcal{E} \gamma^h (\cdot + 6 h^{-2} t)) \dx \dh u^h \|_{L_t^2 (I ; L_x^2 (\R))} \\
    &\leq \| P_{\frac{\pi}{4h}} u^h \|_{L_x^2 (\R ; L_t^\infty (I))} \| \dh \dx u^h \|_{L_x^\infty (\R ; L_t^2 (I))} \\
    &\quad + |I|^{\frac 14} \| P_{\frac{\pi}{4h}}^\perp u^h \|_{L_t^\infty (I; L_x^2 (\R))} \| \dx \dh u^h \|_{L_t^4 (I ; L_x^\infty (\R))} \\
    &\quad + |I|^{\frac 14} \| \mathcal{E} \wt \gamma^h - \mathcal{E} \gamma^h \|_{L_t^\infty (I ; L_x^2 (\R))} \| \dx \dh u^h \|_{L_t^4 (I ; L_x^\infty (\R))} \\
    &\les \| P_{\frac{\pi}{4h}} u^h \|_{L_x^2 (\R ; L_t^\infty (I))} \| \dh \dx u^h \|_{L_x^\infty (\R ; L_t^2 (I))} \\
    &\quad + |I|^{\frac 14} \| u^h \|_{L_t^\infty (I; H_x^1 (\R))} \| \dh u^h \|_{L_t^4 (I; L_x^\infty (\R))} \\
    &\quad + |I|^{\frac 14} h^{-1} \| \mathcal{E} \wt \gamma^h - \mathcal{E} \gamma^h \|_{L_t^\infty (I; L_x^2 (\R))} \| \dh u^h \|_{L_t^4 (I; L_x^\infty (\R))} .
\end{split}
\label{quad3-1}
\end{align}

\noi
We now use similar steps in \eqref{quad2-2} to obtain that for any $t \in I$,
\begin{align}
\begin{split}
    \| &\mathcal{E} \wt \gamma^h (t) - \mathcal{E} \gamma^h (t) \|_{L_x^2 (\R)} \les h \| \gamma^h (t) \|_{\dot{H}^1 (h \Z)} \les h \big( \| \gamma_0^h \|_{H^1 (h \Z)} + \| \gamma_0^h \|_{L^2 (h \Z)}^{\frac 53} \big).
\end{split}
\label{quad3-2}
\end{align}

\noi
Combining \eqref{quad3-1} and \eqref{quad3-2}, we get
\begin{align}
\begin{split}
    \textup{I}_2 &\les \| P_{\frac{\pi}{4h}} u^h \|_{L_x^2 (\R ; L_t^\infty (I))} \| \dh \dx u^h \|_{L_x^\infty (\R ; L_t^2 (I))} \\
    &\quad + |I|^{\frac 14} \| u^h \|_{L_t^\infty (I; H_x^1 (\R))} \| \dh u^h \|_{L_t^4 (I; L_x^\infty (\R))} \\
    &\quad + |I|^{\frac 14} \Big( \| \gamma_0^h \|_{H^1 (h \Z)} + \| \gamma_0^h \|_{L^2 (h \Z)}^{\frac 53} \Big) \| \dh u^h \|_{L_t^4 (I; L_x^\infty (\R))} .
\end{split}
\label{quad3}
\end{align}

\noi
For $\textup{I}_3$, by H\"older's inequalities, Sobolev's embedding, the fact that $|\frac{1}{h} \sin (h \xi)| \leq |\xi|$, and \eqref{quad2-2}, we obtain
\begin{align}
\begin{split}
    \textup{I}_3 &\leq |I|^{\frac 12} \| \mathcal{E} \wt \gamma^h (\cdot + 6 h^{-2} t) \|_{L_t^\infty (I ; L^\infty_x (\R))} \| \dh u^h \|_{L_t^\infty (I ; L_x^2 (\R))} \\
    &\les |I|^{\frac 12} \big( \| u^h \|_{L_t^\infty (I ; H_x^1 (\R))} + \| \mathcal{E} \wt \gamma^h - \mathcal{E} \gamma^h \|_{L_t^\infty (I ; H_x^1 (\R))} \big) \| u^h \|_{L_t^\infty (I ; H_x^1 (\R))} \\
    &\les |I|^{\frac 12} \| u^h \|_{L_t^\infty (I ; H_x^1 (\R))}^2 + |I|^{\frac 12} \Big( \| \gamma_0^h \|_{H^1 (h \Z)} + \| \gamma_0^h \|_{L^2 (h \Z)}^{\frac 53} \Big) \| u^h \|_{L_t^\infty (I ; H_x^1 (\R))} .
\end{split}
\label{quad4}
\end{align}

\noi
Combining \eqref{quad1}, \eqref{quad2}, \eqref{quad3}, and \eqref{quad4}, we obtain the desired estimate \eqref{quad02}.
\end{proof}

We now prove the following uniform-in-$h$ a priori bound for $u^h$ in several relevant norms.

\begin{proposition}
\label{PROP:uhbdd}
Let $0 < h \leq 1$ and $\gamma_0^h \in H^1 (h \Z)$ be such that the condition \eqref{h_cond} holds. Let $u^h$ be defined by \eqref{defuh} with $\gamma^h$ being the solution to \eqref{duh_g} with initial data $\gamma_0^h$. Let $I \subset \R$ be a closed interval with $|I| \leq 1$. Then, for any $\frac 12 < b < 1$, we have
\begin{align*}
    \| &u^h \|_{C_t (I ; H_x^1 (\R))} + \| \dh u^h \|_{L_t^4 (I ; L_x^\infty (\R))} + \| P_{\frac{\pi}{4h}} u^h \|_{L_x^2 (\R ; L_t^\infty (I))} \\
    &+ \| \dh \dx u^h \|_{L_x^\infty (\R ; L_t^2 (I))} + \| \dx^2 u^h \|_{L_x^\infty (\R ; L_t^2 (I))} \\
    &\quad \les \| \gamma_0^h \|_{H^1 (h \Z)} + \| \gamma_0^h \|_{L^2 (h \Z)}^{\frac 53} + \Big( \| \gamma_0^h \|_{H^1 (h \Z)} + \| \gamma_0^h \|_{L^2 (h \Z)}^{\frac 53} \Big)^{\frac{2 - b}{1 - b}} ,
\end{align*}

\noi
where the underlying constant is independent of $h$.
\end{proposition}

\begin{proof}
For simplicity, we assume that $I$ is a unit interval in the positive time line. Let $I = [t_0, t_0 + 1]$ for some $t_0 \geq 0$.
We first show the bound for a short time interval $[t_0, t_0 + T_0]$ with $0 < T_0 \ll 1$ to be determined later. Let $I_0 = [t_0, t_0 + T'] \subset [t_0, t_0 + T_0]$ be a closed interval. On $I_0$, we have the following formulation for $u^h$: 
\begin{align*}
    u^h (t) = e^{- 6 h^{-2} (t - t_0) (\dh - \partial_x)} u^h (t_0) - 6 \int_{t_0}^t e^{- 6 h^{-2} (t - t') (\dh - \partial_x)} e^{6 h^{-2} t' \dx} \mathcal{E} (\wt \gamma^h \dh \gamma^h) (t') dt' .
\end{align*}

\noi
From Lemma~\ref{LEM:Xsb_lin}, we have
\begin{align}
    \| u^h \|_{X_{I_0}^{1, b, h}} \les \| u^h (t_0) \|_{H^1 (\R)} + T_0^{1 - b} \| \mathcal{E} (\wt \gamma^h \dh \gamma^h) \|_{L_t^2 (I_0 ; H_x^1 (\R))} .
\label{lwp1}
\end{align}

\noi
From Lemma~\ref{LEM:quad} and Lemma~\ref{LEM:Xsb_d} along with the fact that $|\frac 1h \sin (h \xi)| \leq |\xi|$, we have
\begin{align}
\begin{split}
    \| \mathcal{E} (\wt \gamma^h \dh \gamma^h) \|_{L_t^2 (I_0 ; H_x^1 (\R))} &\les T_0^{\frac 12} \| u^h \|_{L_t^\infty (I_0 ; H_x^1 (\R))}^2 + T_0^{\frac 14} \| u^h \|_{L_t^\infty (I_0 ; H_x^1 (\R))}  \| \dh u^h \|_{L_t^4 (I_0; L_x^\infty (\R))} \\
    &\quad + \| P_{\frac{\pi}{4h}} u^h \|_{L_x^2 (\R ; L_t^\infty (I_0))} \| \dh \dx u^h \|_{L_x^\infty (\R; L_t^2 (I_0))} \\
    &\quad + T_0^{\frac 14} \Big( \| \gamma_0^h \|_{H^1 (h \Z)} + \| \gamma_0^h \|_{L^2 (h \Z)}^{\frac 53} \Big) \| \dh u^h \|_{L_t^4 (I_0; L_x^\infty (\R))} \\
    &\quad + T_0^{\frac 12} \Big( \| \gamma_0^h \|_{H^1 (h \Z)} + \| \gamma_0^h \|_{L^2 (h \Z)}^{\frac 53} \Big) \| u^h \|_{L_t^\infty (I_0; H_x^1 (\R))} \\
    &\les \Big( \| \gamma_0^h \|_{H^1 (h \Z)} + \| \gamma_0^h \|_{L^2 (h \Z)}^{\frac 53} + \| u^h \|_{X_{I_0}^{1, b, h}} \Big) \| u^h \|_{X_{I_0}^{1, b, h}} .
\end{split}
\label{quadI}
\end{align}

\noi
Thus, from \eqref{lwp1} and \eqref{quadI}, we get
\begin{align}
    \| u^h \|_{X_{I_0}^{1, b, h}} \les \| u^h (t_0) \|_{H^1 (\R)} + T_0^{1 - b} \Big( \| \gamma_0^h \|_{H^1 (h \Z)} + \| \gamma_0^h \|_{L^2 (h \Z)}^{\frac 53} + \| u^h \|_{X_{I_0}^{1, b, h}} \Big) \| u^h \|_{X_{I_0}^{1, b, h}} .
\label{boot}
\end{align}

From \eqref{defuh}, Lemma~\ref{LEM:Hsbdd}, Proposition~\ref{PROP:L2bdd}, and Proposition~\ref{PROP:H1bdd}, we have
\begin{align}
    \| u^h (t_0) \|_{H^1 (\R)} = \| \mathcal{E} \gamma^h (t_0) \|_{H^1 (\R)} \sim \| \gamma^h (t_0) \|_{H^1 (h \Z)} \les \| \gamma_0^h \|_{H^1 (h \Z)} + \| \gamma_0^h \|_{L^2 (h \Z)}^{\frac 53} ,
\label{boot2}
\end{align}

\noi
where the underlying constant is independent of $t_0$.
Thus, from \eqref{boot} and \eqref{boot2},
by letting
\begin{align}
    T_0 \sim \Big( 1 + \| \gamma_0^h \|_{H^1 (h \Z)} + \| \gamma_0^h \|_{L^2 (h \Z)}^{\frac 53} \Big)^{- \frac{1}{1 - b}}
\label{T0size}
\end{align}
and using a standard bootstrap argument along with Lemma~\ref{LEM:Xsb_cty}, we obtain
\begin{align}
    \| u^h \|_{X_{[t_0, t_0 + T_0]}^{1, b, h}} \les  \| \gamma_0^h \|_{H^1 (h \Z)} + \| \gamma_0^h \|_{L^2 (h \Z)}^{\frac 53} ,
\label{boot3}
\end{align}

\noi
where the underlying constant is independent of the time interval $[t_0, t_0 + T_0]$. The same bound holds true for any interval of size $T_0$. We may further shrink the value of $T_0$ so that $\frac{1}{T_0}$ is a positive integer. 
By using Lemma~\ref{LEM:Xsb_d}~(i) along with \eqref{T0size} and \eqref{boot3}, we obtain
\begin{align*}
    \| u^h \|_{C_t (I; H_x^1 (\R))}
    &\leq \sum_{j = 0}^{\frac{1}{T_0} - 1} \| u^h \|_{C_t  ([t_0 + j T_0, t_0 + (j + 1) T_0] ; H_x^1 ( \R ))}  \\
    &\les \sum_{j = 0}^{\frac{1}{T_0} - 1} \| u^h \|_{X_{[t_0 + j T_0, t_0 + (j + 1)T_0]}^{1, b, h}} \\
    &\les \| \gamma_0^h \|_{H^1 (h \Z)} + \| \gamma_0^h \|_{L^2 (h \Z)}^{\frac 53} + \Big( \| \gamma_0^h \|_{H^1 (h \Z)} + \| \gamma_0^h \|_{L^2 (h \Z)}^{\frac 53} \Big)^{\frac{2 - b}{1 - b}} ,
\end{align*}

\noi
which is a valid bound. The other norm bounds follow similarly from the other parts of Lemma~\ref{LEM:Xsb_d}.
\end{proof}

\begin{remark} \rm
Note that this is the only place where we use the Fourier restriction norm $X^{s, b, h}$. Compared to the $L^2_x L_t^\infty$-norm in the statement, the $X^{s, b, h}$-norm is more suitable for performing a bootstrap argument thanks to Lemma~\ref{LEM:Xsb_cty}. It may be possible to prove the a priori bound by iterating the contraction argument, but one may need some efforts to deal with the $\wt \gamma^h$ term appearing in the nonlinearity of the equation \eqref{eq_uh}. 
\end{remark}

We also have the following a priori bound for the solution $u$ to the KdV equation.

\begin{proposition}
\label{PROP:ubdd}
Let $u_0 \in H^1 (\R)$ and let $u \in C (\R; H^1 (\R))$ be the unique global-in-time solution to the KdV equation \eqref{duh_kdv}. Let $I \subset \R$ be a closed interval with $|I| \leq 1$. Then, for any $\frac 12 < b < 1$, we have
\begin{align*}
    \| &u \|_{C_t (I; H_x^1 (\R))} + \| \dx u \|_{L_t^4 (I ; L_x^\infty (\R))} + \| u \|_{L_x^2 (\R ; L_t^\infty (I))} \\
    &+ \| \dx^2 u \|_{L_x^\infty (\R ; L_t^2 (I))} \les  \| u_0 \|_{H^1 (\R)} + \| u_0 \|_{H^1 (\R)}^{\frac{2 - b}{1 - b}} .
\end{align*}
\end{proposition}

\begin{proof}
The proof is similar and easier to that of Proposition~\ref{PROP:uhbdd}. Indeed, the linear estimates in Lemma~\ref{LEM:Xsb_d} hold also in the case $h = 0$. Moreover, we have the conservation of mass and energy
\begin{align*}
    M (t) = \int_{\R} u (t, x)^2 dx \quad \text{and} \quad E(t) = \int_\R \Big( \frac 12 \dx u (t, x)^2 - u (t, x)^3 \Big) dx 
\end{align*}

\noi
under the flow of the KdV equation. The mass $M(t)$ provides an a priori bound for the $L^2(\R)$-norm of the solution and the energy $E (t)$, along with the Gagliardo-Nirenberg inequality, provides an a priori bound for the $\dot{H}^1 (\R)$-norm of the solution. We omit the rest of the steps.
\end{proof}

\subsection{Proof of the dynamical convergence}
\label{SUB:dyn}

In this subsection, we prove the convergence of the dynamics as stated in Theorem~\ref{THM:conv}~(ii).

Given an interval $I \subset \R$, we define the norm
\begin{align*}
    \| u \|_{S_I} := \| u \|_{C_t (I; L_x^2 (\R))} + \| \dx u \|_{L_x^\infty (\R; L_t^2 (I))}.
\end{align*}


\noi
Let us first prove a useful lemma exploiting the convergence rate of the discrete linear propagator to the continuous linear propagator. For the proof, we follow \cite[Proposition~5.10]{HKY}. 
\begin{lemma}
\label{LEM:conv0}
Let $0 < h \leq 1$. Then, for any interval $I = [t_1, t_2] \subset \R$ and function $f$ defined on $\R$, we have
\begin{align*}
    \big\| (e^{- 6 h^{-2} t (\dh - \partial_x)} - e^{- t \partial_x^3}) P_{\frac{\pi}{h}} f \big\|_{S_I} \les  \max (|t_1|, |t_2|)^{\frac 15}  h^{\frac 25} \| f \|_{H^1 (\R)}.
\end{align*}
\end{lemma}

\begin{proof}
We first consider the $C_t L_x^2$-norm. By Plancherel's identity, we have
\begin{align}
\begin{split}
    \big\| &(e^{- 6 h^{-2} t (\dh - \partial_x)} - e^{- t \partial_x^3}) P_{\frac{\pi}{h}} f \big\|_{C_t (I; L_x^2 (\R))} \\
    &= \big\| (e^{- 6 i h^{-3} t (\sin (h \xi) - h \xi)} - e^{i t \xi^3}) \ft{P_{\frac{\pi}{h}} f } (\xi) \big\|_{C_t (I; L_\xi^2 (\R))} \\
    &= \Big\| \big( e^{- 6 i h^{-3} t (\sin (h \xi) - h \xi + \frac{(h \xi)^3}{6})} - 1 \big) \ft{P_{\frac{\pi}{h}} f } (\xi) \Big\|_{C_t (I; L_\xi^2 (\R))} .
\end{split}
\label{conv0-1}
\end{align}

\noi
Note that for any $\xi \in \R$, we have
\begin{align}
    h^{-3} \Big| \sin (h \xi) - h \xi + \frac{(h \xi)^3}{6} \Big| \les h^2 |\xi|^5.
\label{taylor}
\end{align}

\noi
Thus, by the mean value theorem with \eqref{taylor}, we get
\begin{align}
    \big| e^{- 6 i h^{-3} t (\sin (h \xi) - h \xi + \frac{(h \xi)^3}{6})} - 1 \big| \les \min ( |t| h^2 |\xi|^5, 1 ) \les |t|^{\frac 15} h^{\frac 25} |\xi|.
\label{mvt}
\end{align}

\noi
Thus, by \eqref{conv0-1} and \eqref{mvt}, we obtain
\begin{align*}
    \big\| (e^{- 6 h^{-2} t (\dh - \partial_x)} - e^{- t \partial_x^3}) P_{\frac{\pi}{h}} f  \big\|_{C_t (I ; L_x^2 ( \R ))} \les \max (|t_1|, |t_2|)^{\frac 15} h^{\frac 25} \| f \|_{H^1 (\R)} ,
\end{align*}

\noi
which gives the desired estimate for the $C_t L_x^2$-norm.

We now consider the $L_x^\infty L_t^2$-norm. We denote by $P_{\textup{hi}}$ the frequency projection onto high frequency $\{ |\xi| > \max (|t_1|, |t_2|)^{- \frac 15} h^{- \frac 25} \}$ and $P_{\textup{lo}}$ the frequency projection onto low frequency $\{|\xi| \leq \max (|t_1|, |t_2|)^{- \frac 15} h^{- \frac 25}\}$. For the high frequency case, we use Lemma~\ref{LEM:lin_d}~(iv) to obtain
\begin{align}
\begin{split}
    \big\| &\dx (e^{- 6 h^{-2} t (\dh - \partial_x)} - e^{- t \partial_x^3}) P_{\frac{\pi}{h}} P_{\textup{hi}} f  \big\|_{L_x^\infty (\R ; L_t^2 (I))} \\
    &\leq \big\| \dx e^{- 6 h^{-2} t (\dh - \partial_x)}  P_{\frac{\pi}{h}} P_{\textup{hi}} f \big\|_{L_x^\infty (\R ; L_t^2 (I))} + \big\| \dx e^{- t \partial_x^3} P_{\frac{\pi}{h}} P_{\textup{hi}} f \big\|_{L_x^\infty (\R ; L_t^2 (I))} \\
    &\les \| P_{\textup{hi}} f \|_{L^2 (\R)} \\
    &\leq \max (|t_1|, |t_2|)^{\frac 15} h^{\frac 25} \| f \|_{H^1 (\R)}.
\end{split}
\label{conv0-3}
\end{align}

\noi
For the low frequency case, we use the fundamental theorem of calculus to obtain
\begin{align}
\begin{split}
    &e^{- 6 i h^{-3} t (\sin (h \xi) - h \xi)} - e^{i t \xi^3} \\
    &= - 6 i h^{-3} t \Big( \sin (h \xi) - h \xi + \frac{(h \xi)^3}{6} \Big) \int_0^1 e^{- 6 \ta i h^{-3} t (\sin (h \xi) - h \xi) + (1 - \ta) i t \xi^3} d \ta.
\end{split}
\label{ftc}
\end{align}

\noi
for each $t \in I$. Thus, by \eqref{ftc}, Minkowski's integral inequality, Lemma~\ref{LEM:lin_d}~(iv), \eqref{taylor}, and the condition that $|\xi| \leq \max (|t_1|, |t_2|)^{- \frac 15} h^{- \frac 25}$, we have
\begin{align}
\begin{split}
    \big\| &\dx (e^{- 6 h^{-2} t (\dh - \partial_x)} - e^{- t \partial_x^3}) P_{\frac{\pi}{h}} P_{\textup{lo}} f \big\|_{L_x^\infty (\R ; L_t^2 (I))} \\
    &\leq \max (|t_1|, |t_2|) \\
    &\quad \times \int_0^1 \big\| \dx e^{- 6 \ta h^{-2} t (\dh - \dx) - (1 - \ta) t \dx^3} ( 6 h^{-2} (\dh - \dx) - \dx^3 ) P_{\frac{\pi}{h}} P_{\textup{lo}} f \big\|_{L_x^\infty (\R ; L_t^2 (I))} d \ta \\
    &\les \max (|t_1|, |t_2|) \Big\| h^{-3} \big| \sin (h \xi) - h \xi + \tfrac{(h \xi)^3}{6} \big| \ft{P_{\textup{lo}} f} (\xi) \Big\|_{L_\xi^2 (\R)} \\
    &\les \max (|t_1|, |t_2|) \big\| h^2 |\xi|^5 \ft{P_{\textup{lo}} f} (\xi) \big\|_{L_\xi^2 (\R)} \\
    &\leq \max (|t_1|, |t_2|)^{\frac 15} h^{\frac 25} \| f \|_{H^1 (\R)}.
\end{split}
\label{conv0-4}
\end{align}

\noi
Combining \eqref{conv0-3} and \eqref{conv0-4}, we obtain the desired estimate for the $L_x^\infty L_t^2$-norm.
\end{proof}

We now use Lemma~\ref{LEM:conv0} to prove the following convergence results.

\begin{proposition}
\label{PROP:conv1}
Let $0 < h \leq 1$ and $\gamma_0^h \in H^1 (h \Z)$ be such that the condition \eqref{h_cond} holds. Let $\gamma^h$ be the solution to \eqref{duh_g} with initial data $\gamma_0^h$ and let $\wt \gamma^h$ be given by \eqref{defgh2}. Then, for any interval $I = [t_1, t_2] \subset \R$ with $|I| \leq 1$, we have
\begin{align}
    \bigg\| \int_{t_1}^{t} \big( e^{- 6 h^{-2} (t - t') (\dh - \partial_x)} - e^{- (t - t') \partial_x^3} \big) e^{6 h^{-2} t' \dx} \mathcal{E} (\wt \gamma^h \dh \gamma^h) (t') dt' \bigg\|_{S_I} \les  h^{\frac 25} C (\| \gamma_0^h \|_{H^1 (h \Z)})
\label{conv1-2}
\end{align}

\noi
for some constant $C (\| \gamma_0^h \|_{H^1 (h \Z)}) > 0$.
\end{proposition}

\begin{proof}
By Minkowski's integral inequality, Lemma~\ref{LEM:conv0}, and the Cauchy-Schwarz inequality in $t'$, we have
\begin{align}
\begin{split}
    \bigg\| &\int_{t_1}^t \big( e^{- 6 h^{-2} (t - t') (\dh - \partial_x)} - e^{- (t - t') \partial_x^3} \big) e^{6 h^{-2} t' \dx} \mathcal{E} (\wt \gamma^h \dh \gamma^h) (t') dt' \bigg\|_{S_I} \\
    &\leq \int_I \big\| \big( e^{- 6 h^{-2} (t - t') (\dh - \partial_x)} - e^{- (t - t') \partial_x^3} \big) e^{6 h^{-2} t' \dx} \mathcal{E} (\wt \gamma^h \dh \gamma^h) (t') \big\|_{S_{I}} dt' \\
    &\les h^{\frac 25} \int_I \| \mathcal{E} (\wt \gamma^h \dh \gamma^h) (t') \|_{H_x^1 (\R)} dt' \\
    &\leq h^{\frac 25} \| \mathcal{E} (\wt \gamma^h \dh \gamma^h) \|_{L_t^2 (I; H_x^1 (\R))}.
\end{split}
\label{conv1-3}
\end{align}

\noi
To deal with the quadratic term, we use Lemma~\ref{LEM:quad} and Proposition~\ref{PROP:uhbdd} to obtain
\begin{align}
\begin{split}
    \| \mathcal{E} (\wt \gamma^h \dh \gamma^h) \|_{L_t^2 (I ; H_x^1 (\R))} 
    &\les \| u^h \|_{L_t^\infty (I ; H_x^1 (\R))}^2 + \| u^h \|_{L_t^\infty (I ; H_x^1 (\R))}  \| \dh u^h \|_{L_t^4 (I; L_x^\infty (\R))} \\
    &\quad + \| P_{\frac{\pi}{4h}} u^h \|_{L_x^2 (\R ; L_t^\infty (I))} \| \dh \dx u^h \|_{L_x^\infty (\R; L_t^2 (I))} \\
    &\quad + \Big( \| \gamma_0^h \|_{H^1 (h \Z)} + \| \gamma_0^h \|_{L^2 (h \Z)}^{\frac 53} \Big) \| \dh u^h \|_{L_t^4 (I; L_x^\infty (\R))} \\
    &\quad + \Big( \| \gamma_0^h \|_{H^1 (h \Z)} + \| \gamma_0^h \|_{L^2 (h \Z)}^{\frac 53} \Big) \| u^h \|_{L_t^\infty (I; H_x^1 (\R))} \\
    &\leq C ( \| \gamma_0^h \|_{H^1 (h \Z)} )
\end{split}
\label{conv1-4}
\end{align}

\noi
for some constant $C ( \| \gamma_0^h \|_{H^1 (h \Z)} ) > 0$. Combining \eqref{conv1-3} and \eqref{conv1-4}, we obtain the desired estimate \eqref{conv1-2}.
\end{proof}

We now exploit the convergence of the quadratic nonlinearities.

\begin{proposition}
\label{PROP:conv4}
Let $0 < h \leq 1$ and $\gamma_0^h \in H^1 (h \Z)$ be such that the condition \eqref{h_cond} holds. Let $\gamma^h$ be the solution to \eqref{duh_g} with initial data $\gamma_0^h$, $\wt \gamma^h$ be given by \eqref{defgh2}, and $u^h$ be given by \eqref{defuh}. Let $u_0 \in H^1 (\R)$ and let $u \in C (\R; H^1 (\R))$ be the unique global-in-time solution to the KdV equation \eqref{duh_kdv}. Then, for any $I = [t_1, t_2] \subset \R$ with $|I| \leq 1$, we have
\begin{align*}
    \bigg\| &\int_{t_1}^t e^{- (t - t') \dx^3} \big( e^{6 h^{-2} t' \dx} \mathcal{E} (\wt \gamma^h \dh \gamma^h) (t') -  u (t') \partial_x u (t') \big) dt' \bigg\|_{S_I} \\
    &\les C (\| \gamma_0^h \|_{H^1 (h \Z)}, \| u_0 \|_{H^1 (\R)}) \big( h + |I|^{\frac 12} \| u^h - u \|_{S_I} \big)
\end{align*}

\noi
for some constant $C (\| \gamma_0^h \|_{H^1 (h \Z)}, \| u_0 \|_{H^1 (\R)}) > 0$.
\end{proposition}

\begin{proof}
For any space-time function $F$, by Minkowski's integral inequality, Lemma~\ref{LEM:lin_d}~(i) and (iv), and H\"older's inequality in $t'$, we have
\begin{align*}
    \bigg\| \int_{t_1}^t e^{- (t - t') \dx^3} F (t') dt' \bigg\|_{C_t (I; L_x^2 (\R))} &\les \int_I \big\| e^{- (t - t') \dx^3} F (t') \big\|_{C_t (I; L_x^2 (\R))} dt' \\
    &\leq |I|^{\frac 12} \| F(t) \|_{L_t^2 (I; L_x^2 (\R))}
\end{align*}

\noi
and
\begin{align*}
    \bigg\| \int_{t_1}^t \dx e^{- (t - t') \dx^3} F (t') dt' \bigg\|_{L_x^\infty (\R; L_t^2 (I))} 
    &\les \int_I \big\| \dx e^{- (t - t') \dx^3} F (t') \big\|_{L_x^\infty (\R; L_t^2 (I))} dt' \\
    &\leq |I|^{\frac 12} \| F(t) \|_{L_t^2 (I; L_x^2 (\R))}.
\end{align*}

\noi
Thus, the above bounds along with Lemma~\ref{LEM:Eprod}, \eqref{defuh}, and the fact that $|I| \leq 1$ give
\begin{align}
\begin{split}
    &\bigg\| \int_{t_1}^t e^{- (t - t') \dx^3} \big( e^{6 h^{-2} t' \dx} \mathcal{E} (\wt \gamma^h \dh \gamma^h) (t') -  u (t') \partial_x u (t') \big) dt' \bigg\|_{S_I} \\
    &\les |I|^{\frac 12} \big\| e^{6 h^{-2} t \dx} \mathcal{E} (\wt \gamma^h \dh \gamma^h) - u \partial_x u \big\|_{L_t^2 (I; L_x^2 (\R))} \\
    &\leq 2 \big\| P_{\frac{\pi}{h}} \big( \cos (\tfrac{2 \pi}{h} \cdot) \mathcal{E} \wt \gamma^h \mathcal{E} (\dh \gamma^h) \big) \big\|_{L_t^2 (I; L_x^2 (\R))} + \big\| P_{\frac{\pi}{h}}^\perp \big( \mathcal{E} \wt \gamma^h \mathcal{E} (\dh \gamma^h) \big) \big\|_{L_t^2 (I; L_x^2 (\R))} \\
    &\quad + \big\| ( \mathcal{E} \wt \gamma^h - \mathcal{E} \gamma^h ) \mathcal{E} (\dh \gamma^h) \big\|_{L_t^2 (I; L_x^2 (\R))} + |I|^{\frac 12} \| u^h \dh u^h - u \dh u^h \|_{L_t^2 (I; L_x^2 (\R))} \\
    &\quad + \| u \dh u^h - u \dx u^h \|_{L_t^2 (I; L_x^2 (\R))} + |I|^{\frac 12} \| u \dx u^h - u \dx u \|_{L_t^2 (I; L_x^2 (\R))} \\
    &=: \textup{I}_1 + \textup{I}_2 + \textup{I}_3 + \textup{I}_4 + \textup{I}_5 + \textup{I}_6 .
\end{split}
\label{conv4-0}
\end{align}

\noi
For $\textup{I}_1$, we use Lemma~\ref{LEM:Ecos}, Lemma~\ref{LEM:quad}, and Proposition~\ref{PROP:uhbdd} to obtain
\begin{align}
\begin{split}
    \textup{I}_1 &\les h \| \mathcal{E} \wt \gamma^h \mathcal{E} (\dh \gamma^h) \|_{L_t^2 (I; H_x^1 (\R))} \\
    &\les h \| u^h \|_{L_t^\infty (I ; H_x^1 (\R))}^2 + h \| u^h \|_{L_t^\infty (I ; H_x^1 (\R))}  \| \dh u^h \|_{L_t^4 (I; L_x^\infty (\R))} \\
    &\quad + h \| P_{\frac{\pi}{4 h}} u^h \|_{L_x^2 (\R ; L_t^\infty (I))} \| \dh \dx u^h \|_{L_x^\infty (\R; L_t^2 (I))} \\
    &\quad + h C (\| \gamma_0^h \|_{H^1 (h \Z)}) \big( \| \dh u^h \|_{L_t^4 (I; L_x^\infty (\R))} + \| u^h \|_{L_t^\infty (I; H_x^1 (\R))} \big) \\
    &\les h C (\| \gamma_0^h \|_{H^1 (h \Z)})
\end{split}
\label{conv4-1}
\end{align}

\noi
for some constant $C (\| \gamma_0^h \|_{H^1 (h \Z)}) > 0$ which may vary from line to line below.
For $\textup{I}_2$, the frequency $\xi$ satisfies $|\xi| > \frac{\pi}{h}$, so that with similar steps to \eqref{conv4-1} we get
\begin{align}
    \textup{I}_2 \les h \| \mathcal{E} \wt \gamma^h \mathcal{E} (\dh \gamma^h) \|_{L_t^2 (I; H_x^1 (\R))} \les h C (\| \gamma_0^h \|_{H^1 (h \Z)}) .
\label{conv4-2}
\end{align}

\noi
For $\textup{I}_3$, we use H\"older's inequality, \eqref{quad3-2}, \eqref{defuh}, and Proposition~\ref{PROP:uhbdd} to obtain
\begin{align}
\begin{split}
    \textup{I}_3 &\leq \| \mathcal{E} \wt \gamma^h - \mathcal{E} \gamma^h \|_{L_t^\infty (I; L_x^2(\R))} \| \mathcal{E} (\dh \gamma^h) \|_{L_t^4 (I; L_x^\infty (\R))} \\
    &\les h C (\| \gamma_0^h \|_{H^1 (h \Z)}) \| \dh u^h \|_{L_t^4 (I; L_x^\infty (\R))} \\
    &\les h C (\| \gamma_0^h \|_{H^1 (h \Z)}) .
\end{split}
\label{conv4-3}
\end{align}

\noi
For $\textup{I}_4$, we use H\"older's inequality and Proposition~\ref{PROP:uhbdd} to obtain
\begin{align}
\begin{split}
    \textup{I}_4 &\les |I|^{\frac 34} \| u^h - u \|_{L_t^\infty (I; L_x^2 (\R))} \| \dh u^h \|_{L_t^4 (I; L_x^\infty (\R))}  \\
    &\les C (\| \gamma_0^h \|_{H^1 (h \Z)}) |I|^{\frac 34} \| u^h - u \|_{S_I} .
\end{split}
\label{conv4-4}
\end{align}

\noi
For $\textup{I}_5$, we first see that by \eqref{dhf} and the fundamental theorem of calculus, we get that for any $x \in \R$,
\begin{align}
\begin{split}
    \dh u^h (x) - \dx u^h (x) &= \frac{1}{2h} \big( u^h (x + h) - u^h (x - h) \big) - \dx u^h (x) \\
    &= \frac 12 \int_{-1}^1 \big( \dx u^h (x + sh) - \dx u^h (x) \big) ds \\
    &= \frac{h}{2} \int_{-1}^1 s \int_0^1 \dx^2 u^h (x + s' sh) ds' ds ,
\end{split}
\label{dhdx}
\end{align}

\noi
so that by H\"older's inequality, \eqref{dhdx}, Minkowski's integral inequality, Proposition~\ref{PROP:uhbdd}, and Proposition~\ref{PROP:ubdd}, we obtain
\begin{align}
\begin{split}
    \textup{I}_5 &\les \| u \|_{L_x^2 (\R; L_t^\infty (I))} \| \dh u^h - \dx u^h \|_{L_x^\infty (\R; L_t^2 (I))} \\
    &\leq h \| u \|_{L_x^2 (\R; L_t^\infty (I))} \| \dx^2 u^h \|_{L_x^\infty (\R; L_t^2 (I))} \\
    &\les h C (\| \gamma_0^h \|_{H^1 (h \Z)}, \| u_0 \|_{H^1 (\R)})
\end{split}
\label{conv4-5}
\end{align}

\noi
for some constant $C (\| \gamma_0^h \|_{H^1 (h \Z)}, \| u_0 \|_{H^1 (\R)}) > 0$. For $\textup{I}_6$, we use H\"older's inequalities, Proposition~\ref{PROP:uhbdd}, and Proposition~\ref{PROP:ubdd} to obtain
\begin{align}
\begin{split}
    \textup{I}_6 &\les |I|^{\frac 12} \| u \|_{L_x^2 (\R; L_t^\infty (I))} \| \dx (u^h - u) \|_{L_x^\infty (\R; L_t^2 (I))} \\
    &\les C (\| \gamma_0^h \|_{H^1 (h \Z)}, \| u_0 \|_{H^1 (\R)}) |I|^{\frac 12} \| u^h - u \|_{S_I} .
\end{split}
\label{conv4-6}
\end{align}

We obtain the desired estimate by combining \eqref{conv4-0}, \eqref{conv4-1}, \eqref{conv4-2}, \eqref{conv4-3}, \eqref{conv4-4}, \eqref{conv4-5}, and \eqref{conv4-6}.
\end{proof}

We are now ready to prove the main statement on the dynamical convergence.
\begin{proof}[Proof of Theorem~\ref{THM:conv}~\textup{(ii)}]
Fix $T > 0$. For simplicity, we only show the convergence of dynamics on $[0, T]$, since the convergence on $[-T, 0]$ is similar. 

We first consider the convergence on the time interval $[0, T_0]$ for some small $0 < T_0 \leq \min ( 1, T )$ to be chosen later.
From \eqref{eq_uh} and \eqref{duh_kdv}, we have for any $t \in [0, T_0]$ that
\begin{align*}
    u^h (t) - u (t) &= e^{- 6 h^{-2} t (\dh - \dx)} \mathcal{E} \gamma_0^h - e^{- t \dx^3} u_0 \\
    &\quad - 6 \int_0^t \big( e^{- 6 h^{-2} (t - t') (\dh - \dx)} - e^{- (t - t') \dx^3} \big) e^{6 h^{-2} t' \dx} \mathcal{E} (\wt \gamma^h \dh \gamma^h) (t') dt' \\
    &\quad - 6 \int_0^t e^{- (t - t') \dx^3} \big( e^{6 h^{-2} t' \dx} \mathcal{E} (\wt \gamma^h \dh \gamma^h) (t') - u (t') \dx u (t') \big) dt'.
\end{align*}

\noi
By Lemma~\ref{LEM:conv0}, Lemma~\ref{LEM:lin_d}~(i) and (iv), and Lemma~\ref{LEM:Hsbdd}, we have
\begin{align}
\begin{split}
    \big\| &e^{- 6 h^{-2} t (\dh - \partial_x)} \mathcal{E} \gamma^h_0 - e^{- t \partial_x^3} u_0 \big\|_{S_{[0, T_0]}} \\
    &\leq \big\| ( e^{- 6 h^{-2} t (\dh - \partial_x)}  - e^{- t \partial_x^3} ) \mathcal{E} \gamma_0^h \big\|_{S_{[0, T_0]}} + \big\| e^{- t \dx^3} (\gamma_0^h - u_0) \big\|_{S_{[0, T_0]}} \\
    &\leq C_0 T_0^{\frac 15} h^{\frac 25} \| \mathcal{E} \gamma_0^h \|_{H^1 (\R)} + C_0 \| \mathcal{E} \gamma_0^h - u_0 \|_{L^2 (\R)} \\
    &\leq \wt{C}_0 T^{\frac 15} h^{\frac 25} \| \gamma_0^h \|_{H^1 (h \Z)} + C_0 \| \mathcal{E} \gamma_0^h - u_0 \|_{L^2 (\R)},
\end{split}
\label{diff_init}
\end{align}

\noi
for some constants $C_0, \wt C_0 > 0$. Thus, by \eqref{diff_init}, Proposition~\ref{PROP:conv1}, and Proposition~\ref{PROP:conv4}, we have
\begin{align*}
    \| u^h - u \|_{S_{[0, T_0]}} &\leq T^{\frac 15} h^{\frac 25} C( \| \gamma_0^h \|_{H^1 (h \Z)}, \| u_0 \|_{H^1 (\R)} ) +  C_0 \| \mathcal{E} \gamma_0^h - u_0 \|_{L^2 (\R)} \\
    &\quad + T_0^{\frac 12} C( \| \gamma_0^h \|_{H^1 (h \Z)}, \| u_0 \|_{H^1 (\R)} ) \| u^h - u \|_{S_{[0, T_0]}} 
\end{align*}

\noi
for some constant $C( \| \gamma_0^h \|_{H^1 (h \Z)}, \| u_0 \|_{H^1 (\R)} ) > 0$. By choosing $T_0$ small enough in such a way that
\begin{align}
    T_0^{\frac 12} C( \| \gamma_0^h \|_{H^1 (h \Z)}, \| u_0 \|_{H^1 (\R)} ) \leq \frac 12, 
\label{T0}
\end{align}

\noi
we obtain 
\begin{align}
    \| u^h - u \|_{S_{[0, T_0]}} \leq 2 T^{\frac 15} h^{\frac 25} C( \| \gamma_0^h \|_{H^1 (h \Z)}, \| u_0 \|_{H^1 (\R)} ) +  2 C_0 \| \mathcal{E} \gamma_0^h - u_0 \|_{L^2 (\R)}.
\label{conv_main}
\end{align}

We now repeat the above procedure on the time interval $[T_0, 2T_0]$. For any $t \in [T_0, 2 T_0]$, we have the following difference of the Duhamel integrals:
\begin{align*}
    u^h (t) - u (t) &= e^{- 6 h^{-2} (t - T_0) (\dh - \dx)} u^h (T_0) - e^{- (t - T_0) \dx^3} u (T_0) \\
    &\quad - 6 \int_{T_0}^t \big( e^{- 6 h^{-2} (t - t') (\dh - \dx)} - e^{- (t - t') \dx^3} \big) e^{6 h^{-2} t' \dx} \mathcal{E} (\wt \gamma^h \dh \gamma^h) (t') dt' \\
    &\quad - 6 \int_{T_0}^t e^{- (t - t') \dx^3} \big( e^{6 h^{-2} t' \dx} \mathcal{E} ( \wt \gamma^h \dh \gamma^h ) (t') - u (t') \dx u (t') \big) dt' .
\end{align*}

\noi
For the difference of the linear flows, we use similar steps as \eqref{diff_init} along with Proposition~\ref{PROP:uhbdd} to obtain
\begin{align}
\begin{split}
    \big\| &e^{- 6 h^{-2} t (\dh - \partial_x)} u^h (T_0) - e^{- t \partial_x^3} u (T_0) \big\|_{S_{[T_0, 2T_0]}} \\
    &\leq \big\| ( e^{- 6 h^{-2} t (\dh - \partial_x)}  - e^{- t \partial_x^3} ) u^h (T_0) \big\|_{S_{[T_0, 2T_0]}} + \big\| e^{- t \dx^3} (u^h (T_0) - u (T_0)) \big\|_{S_{[T_0, 2 T_0]}} \\
    &\leq C_0 T^{\frac 15} h^{\frac 25} \| u^h (T_0) \|_{H^1 (\R)} + C_0 \| u^h (T_0) - u (T_0) \|_{L^2 (\R)} \\
    &\leq T^{\frac 15} h^{\frac 25} C ( \| \gamma_0^h \|_{H^1 (h \Z)}, \| u_0 \|_{H^1 (\R)} ) + C_0 \| u^h - u \|_{S_{[0, T_0]}} .
\end{split}
\label{diff_init2}
\end{align}

\noi
Thus, by \eqref{diff_init2}, Proposition~\ref{PROP:conv1}, and Proposition~\ref{PROP:conv4}, we have
\begin{align*}
    \| u^h - u \|_{S_{[T_0, 2 T_0]}} &\leq T^{\frac 15} h^{\frac 25} C( \| \gamma_0^h \|_{H^1 (h \Z)}, \| u_0 \|_{H^1 (\R)} ) +  C_0 \| u^h - u \|_{S_{[0, T_0]}} \\
    &\quad + T_0^{\frac 12} C( \| \gamma_0^h \|_{H^1 (h \Z)}, \| u_0 \|_{H^1 (\R)} ) \| u^h - u \|_{S_{[T_0, 2 T_0]}} ,
\end{align*}

\noi
so that with the choice of $T_0$ in \eqref{T0}, we obtain
\begin{align*}
    \| u^h - u \|_{S_{[T_0, 2 T_0]}} \leq 2 T^{\frac 15} h^{\frac 25} C( \| \gamma_0^h \|_{H^1 (h \Z)}, \| u_0 \|_{H^1 (\R)} ) +  2 C_0 \| u^h - u \|_{S_{[0, T_0]}} .
\end{align*}

We then repeat the above procedure until we reach $T$, and a simple induction with \eqref{conv_main} gives
\begin{align}
\begin{split}
    \| u^h - u \|_{S_{[0, T]}} &\leq (2 C_0)^{4 C ( \| \gamma_0^h \|_{H^1 (h \Z)}, \| u_0 \|_{H^1 (\R)} )^2 T} T^{\frac 15} h^{\frac 25} C( \| \gamma_0^h \|_{H^1 (h \Z)}, \| u_0 \|_{H^1 (\R)} ) \\
    &\quad + (2 C_0)^{4 C ( \| \gamma_0^h \|_{H^1 (h \Z)}, \| u_0 \|_{H^1 (\R)} )^2 T} \| \mathcal{E} \gamma_0^h - u_0 \|_{L^2 (\R)} .
\end{split}
\label{conv}
\end{align}

\noi
The desired convergence in Theorem~\ref{THM:conv}~(ii) then follows from
the estimate \eqref{conv} and also the assumption \eqref{u0cond}.
\end{proof}

\subsection{Long-wave limit in Flaschka's form}
\label{SUB:lw1}

In this subsection, we prove Theorem~\ref{THM:lw}, the long-wave limit of the scaled Toda lattice in Flaschka's form \eqref{todah}.

\begin{proof}[Proof of Theorem~\ref{THM:lw}]
From Theorem~\ref{THM:conv} and \eqref{defuh}, we have
\begin{align}
    \| \mathcal{E} \gamma^h (t, x) - u (t, x - 6 h^{-2} t) \|_{C_T L_x^2 (\R)} \leq C e^{C' T} h^{\frac 25}.
\label{conv_gh-}
\end{align}

\noi
We define a function $w^h$ on $\R \times h \Z$ by
\begin{align*}
    w^h (t, \ld) := u (t, \ld - 6 h^{-2} t).
\end{align*}

\noi
Then, by using the Poisson summation formula (see \cite[(2.13)]{KOVW} and \cite[Theorem~3.2.8]{Gra1}), we have
\begin{align*}
    \F^h w^h (t, \xi) = h \sum_{\ld \in h \Z} u (t, \ld - 6 h^{-2} t) e^{- i \xi \ld}  = \sum_{k \in \Z} e^{- 6 i h^{-2} t (\xi + \frac{2 \pi k}{h})} \ft{u} \Big(t, \xi + \frac{2 \pi k}{h} \Big) .
\end{align*}

\noi
Thus, by Plancherel's identity \eqref{plan}, the Cauchy-Schwarz inequality in $k$, and \eqref{conv_gh-}, we obtain
\begin{align*}
    &\| \gamma^h (t, \ld) - w^h (t, \ld) \|_{C_T L_\ld^2 (h \Z)} \\
    &\leq \big\| \F^h \gamma^h (t, \xi) - e^{- 6 i h^{-2} t \xi} \ft{u} (t, \xi) \big\|_{C_T L_\xi^2 ([- \frac{\pi}{h}, \frac{\pi}{h}))} \\
    &\quad + \Big\| \sum_{k \in \Z \setminus \{0\}} e^{- 6 i h^{-2} t (\xi + \frac{2 \pi k}{h})} \ft{u} \Big(t, \xi + \frac{2 \pi k}{h} \Big) \Big\|_{C_T L_\xi^2 ([- \frac{\pi}{h}, \frac{\pi}{h}))} \\
    &\les \| \mathcal{E} \gamma^h (t, x) - u (t, x - 6 h^{-2} t) \|_{C_T L_x^2 (\R)} + \bigg( \sum_{k \in \Z \setminus \{0\}} \int_{[- \frac{\pi}{h}, \frac{\pi}{h})} |k|^{1 + \eps} \Big| \ft{u} \Big(t, \xi + \frac{2 \pi k}{h} \Big) \Big|^2 d\xi \bigg)^{\frac 12} \\
    &\les e^{C' T} h^{\frac 25} + \| \jb{\xi}^{\frac 12 + \frac{\eps}{2}} \ft u (t, \xi) \|_{C_T L_\xi^2 ([- \frac{\pi}{h}, \frac{\pi}{h})^c)} \\
    &\les e^{C' T} h^{\frac 25} + h^{\frac 12 - \frac{\eps}{2}} \| u \|_{C_T H_x^1 (\R)}
\end{align*}

\noi
for any $\eps > 0$, where we used the fact that $|k| \les \jb{\frac{h}{2 \pi} \xi} \les \jb{\xi}$ given $\xi \in [- \frac{\pi}{h}, \frac{\pi}{h}) + \frac{2 \pi k}{h}$. By choosing $\eps \leq \frac 15$, we obtain the desired estimate.
\end{proof}

\subsection{Long-wave limit in the original form}
\label{SUB:lw2}

In this section, we consider the long-wave limit of the scaled Toda lattice in the form \eqref{eqr2} and prove Theorem~\ref{THM:long-wave}.

We first consider global well-posedness of the Toda lattice in the original form \eqref{eqr2}. As mentioned in the introduction, we will reduce this problem to Theorem~\ref{THM:conv}~(i), global well-posedness of the Toda lattice in Flaschka's form \eqref{toda_abh}.
For this purpose, we need the following lemma.
\begin{lemma}
\label{LEM:exp}
Let $0 < h \leq 1$. Let $f^h \in H^1 (h \Z)$ and let $g_1^h$, $g_2^h$ be functions on $h \Z$ defined by
\begin{align*}
    g_1^h (\ld) =
    \begin{cases}
    h^{-2} \exp ( h^2 f^h (\ld) ) - h^{-2} & \textup{if } \frac{\ld}{h} \textup{ is even} \vspace{5pt} \\
    f^h (\ld) & \textup{if } \frac{\ld}{h} \textup{ is odd}
    \end{cases}
\end{align*}

\noi
and
\begin{align*}
    g_2^h (\ld) = h^{-2} \exp ( h^2 f^h (\ld) ) - h^{-2} .
\end{align*}

\noi
Then, we have $g_1^h, g_2^h \in H^1 (h \Z)$ and
\begin{align*}
    \| g_1^h - f^h \|_{L^2 (h \Z)} &\les h^{\frac 32} \exp ( \| f^h \|_{L^{2} (h \Z)} ) , \\
    \| g_1^h - f^h \|_{H^1 (h \Z)} &\les h^{\frac 12} \exp ( \| f^h \|_{L^{2} (h \Z)} ) , \\
    \| g_2^h - f^h \|_{H^1 (h \Z)} &\les h^{\frac 12} \exp ( \| f^h \|_{L^{2} (h \Z)} ) ,
\end{align*}

\noi
where the underlying constants are independent of $h$.
\end{lemma}

\begin{proof}
We only show the $H^1$-estimate for $g_1^h$, as the $L^2$-estimate for $g_1^h$ and the estimate for $g_2^h$ is similar. By the Taylor expansion, when $\frac{\ld}{h}$ is even, we have
\begin{align*}
    g_1^h (\ld) - f^h (\ld) = \sum_{k = 2}^\infty \frac{h^{2k - 2} (f^h (\ld))^k}{k!} .
\end{align*}

\noi
Thus, by \eqref{H1norm} and Lemma~\ref{LEM:Lpbdd}, we have
\begin{align*}
    \| g_1^h - f^h \|_{H^1 (h \Z)} &\les \| g_1^h - f^h \|_{L^2 (h \Z)} + \| g_1^h - f^h \|_{\dot{H}^1 (h \Z)} \\
    &\les (1 + h^{-1}) \sum_{k = 2}^\infty \frac{h^{2k - 2}}{k!} \| (f^h)^k \|_{L^{2} (h \Z)} \\
    &= (1 + h^{-1}) \sum_{k = 2}^\infty \frac{h^{2k - 2}}{k!} \| f^h \|_{L^{2 k} (h \Z)}^k \\
    &\leq (1 + h^{-1}) \sum_{k = 2}^\infty \frac{h^{\frac 32 k - \frac 32}}{k!} \| f^h \|_{L^{2} (h \Z)}^k \\
    &\leq (h^{\frac 32} + h^{\frac 12}) \exp ( \| f^h \|_{L^{2} (h \Z)} ).
\end{align*}

\noi
This gives the desired estimate.
\end{proof}

We will also need the $L^2$-conservation of $r^h$ satisfying the equation \eqref{eqr2}. This is achieved via the following scaled version of the Hamiltonian \eqref{defH}:
\begin{align}
    H^h (t) := h \sum_{\ld \in h \Z} \bigg( \frac{2}{9} \big( h^2 (\dh^+)^{-1} \dt r^h (t, \ld) \big)^2 + h^{-4} \Big( \exp \big( 2 h^2 r^h (t, \ld) \big) - 2 h^2 r^h (t, \ld) - 1 \Big) \bigg) .
\label{defHh}
\end{align}

\begin{proposition}
\label{PROP:L2bddr}
Let $0 < h \leq 1$ and let $r_0^h$, $r_1^h$, and $u_0$ be as given in Theorem~\ref{THM:long-wave}. Let $r^h \in C ([-T, T]; H^1 (h \Z))$ be the solution to the scaled Toda lattice \eqref{eqr2} for some $T > 0$ with initial data $(r^h, \dt r^h)|_{t = 0} = (r_0^h, \pm 3 h^{-2} \dh^+ r_1^h)$. Then, there exists $h_0 = h_0 (\| u_0 \|_{L^2}) > 0$ sufficiently small such that for any $0 < h \leq h_0$ and $t \in [-T, T]$, we have
\begin{align*}
    \| r^h (t) \|_{L^2 (h \Z)} \leq 3 \| u_0 \|_{L^2 (\R)} .
\end{align*}
\end{proposition}

\begin{proof}
From \eqref{r0cond}, Lemma~\ref{LEM:Hsbdd}, and the condition in \eqref{u0cond2}, we can let $h = h ( \| u_0 \|_{L^2 (\R)} ) > 0$ be sufficiently small so that
\begin{align}
    \| r_0^h \|_{L^2 (h \Z)} \leq \| u_0^h \|_{L^2 (\R)} \leq 2 \| u_0 \|_{L^2 (\R)}
\label{r0h_bdd}
\end{align}

\noi
and
\begin{align}
    \| r_1^h \|_{L^2 (h \Z)} \leq \| u_0^h \|_{L^2 (\R)} \leq 2 \| u_0 \|_{L^2 (\R)} .
\label{r1h_bdd}
\end{align}

\noi
By continuity in time and \eqref{r0h_bdd}, there exists $0 < T_0 \leq T$ such that
\begin{align}
    \| r^h (t) \|_{L^2 (h \Z)} \leq 2 \| r_0^h \|_{L^2 (h \Z)} \leq 4 \| u_0 \|_{L^2 (\R)} 
\label{rh_bdd1}
\end{align}

\noi
for any $t \in [-T_0, T_0]$. Thus, from the Taylor expansion of $e^x$, Lemma~\ref{LEM:Lpbdd}, and \eqref{rh_bdd1}, we obtain
\begin{align}
\begin{split}
    \bigg| &h \sum_{\ld \in h \Z} h^{-4} \Big( \exp \big( 2 h^2 r^h (t, \ld) \big) - 2 h^2 r^h (t, \ld) - 1 \Big) - 2 h \sum_{\ld \in h \Z} r^h (t, \ld)^2 \bigg| \\
    &\leq h \sum_{\ld \in h \Z} \sum_{k = 3}^\infty \frac{2^{k} h^{2k - 4}}{k!} |r^h (t, \ld)|^k \\
    &\leq 2 \sum_{k = 3}^\infty h^{2k - 4} \| r^h (t) \|_{L^k (h \Z)}^k \\
    &\leq 2 \sum_{k = 3}^\infty h^{\frac 32 k - 3} \| r^h (t) \|_{L^2 (h \Z)}^k \\
    &\leq 2 \sum_{k = 3}^\infty 4^k h^{\frac 32 k - 3} \| u_0 \|_{L^2 (\R)}^k
\end{split}
\label{rh_bdd2}
\end{align}

\noi
for any $t \in [-T_0, T_0]$. By using the conservation of the Hamiltonian $H^h$ in \eqref{defHh}, applying the estimates \eqref{rh_bdd2}, \eqref{r1h_bdd}, and \eqref{r0h_bdd}, and letting $h = h (\| u_0 \|_{L^2 (\R)}) > 0$ be sufficiently small, we then deduce that
\begin{align*}
    \| r^h (t) \|_{L^2 (h \Z)}^2 &\leq \frac 12 H^h (t) + \sum_{k = 3}^\infty 4^k h^{\frac 32 k - 3} \| u_0 \|_{L^2 (\R)}^k \\
    &= \frac 12 H^h (0) + \sum_{k = 3}^\infty 4^k h^{\frac 32 k - 3} \| u_0 \|_{L^2 (\R)}^k \\
    &\leq \| r_1^h \|_{L^2 (h \Z)}^2 + \| r_0^h \|_{L^2 (h \Z)}^2 + 2 \sum_{k = 3}^\infty 4^k h^{\frac 32 k - 3} \| u_0 \|_{L^2 (\R)}^k \\
    &\leq 9 \| u_0 \|_{L^2 (\R)}^2 
\end{align*}

\noi
for any $t \in [-T_0, T_0]$, which improves the bound \eqref{rh_bdd1}. We can then repeat the above bootstrap argument to obtain the desired bound.
\end{proof}

We are now ready to show global well-posedness of the Toda lattice in the original form \eqref{eqr2}.

\begin{proof}[Proof of Theorem~\ref{THM:long-wave}~\textup{(i)}]
Let us define
\begin{align}
    \gamma_0^h (\ld) := 
    \begin{cases}
        h^{-2} \exp \big( h^2 r_0^h ( \frac{\ld}{2} ) \big) - h^{-2} & \text{if } \frac{\ld}{h} \text{ is even} \vspace{5pt} \\
        r_1^h ( \frac{\ld + h}{2} ) & \text{if } \frac{\ld}{h} \text{ is odd} 
    \end{cases}
\label{defg0h}
\end{align}

\noi
for any $\ld \in h \Z$. Then, from \eqref{r0cond}, we have the following relationship between $\gamma_0^h$ and $u_0^h$:
\begin{align}
    \gamma_0^h (\ld) = 
    \begin{cases}
        h^{-2} \exp \big( h^2 u_0^h ( \ld ) \big) - h^{-2} & \text{if } \frac{\ld}{h} \text{ is even} \vspace{5pt} \\
        u_0^h ( \ld ) & \text{if } \frac{\ld}{h} \text{ is odd} .
    \end{cases}
\label{g0h}
\end{align}

\noi
From Lemma~\ref{LEM:exp} and the assumption \eqref{u0cond2}, we see that $\gamma_0^h$ satisfies the assumption \eqref{u0cond} in Theorem~\ref{THM:conv}. From \eqref{u0cond} and Lemma~\ref{LEM:Hsbdd}, the condition \eqref{h_cond} is then satisfied with all $h = h(\| u_0 \|_{L^2 (\R)}) > 0$ sufficiently small. Thus, by Theorem~\ref{THM:conv}~(i), given $h = h (\| u_0 \|_{L^2 (\R)})$ sufficiently small, we have a unique solution $\gamma^h \in C(\R; H^1 (h \Z))$ to the scaled Toda lattice \eqref{todah} with initial data $\gamma^h |_{t = 0} = \gamma_0^h$. 

In view of \eqref{defgh} and \eqref{abr}, we only need to take
\begin{align}
    r^h (t, \ld) = h^{-2} \ln \big( h^2 \al^h (t, \ld) + 1 \big) = h^{-2} \ln \big( h^2 \gamma^h (t, 2 \ld) + 1 \big)
\label{rh_form}
\end{align}

\noi
for any $t \in \R$ and $\ld \in h \Z$. From Proposition~\ref{PROP:L2bdd} and the assumption \eqref{u0cond}, we have with $h = h(\| u_0 \|_{L^2 (\R)}) > 0$ sufficiently small that
\begin{align*}
    h^2 |\al^h (t, \ld)| \leq h^{\frac 32} \| \gamma^h (t) \|_{L^2 (h \Z)} \leq 2 h^{\frac 32} \| \gamma_0^h \|_{L^2 (h \Z)} \leq \frac 12
\end{align*}

\noi
for any $t \in \R$ and $\ld \in h \Z$, so that the logarithm in \eqref{rh_form} makes sense. Moreover, by Lemma~\ref{LEM:exp}, Proposition~\ref{PROP:L2bddr}, \eqref{H1norm}, Proposition~\ref{PROP:L2bdd}, and Proposition~\ref{PROP:H1bdd}, we have
\begin{align*}
    \| r^h (t) \|_{H^1 (h \Z)} &\leq \| r^h (t) - \al^h (t) \|_{H^1 (h \Z)} + \| \al^h (t) \|_{H^1 (h \Z)} \\
    &\les h^{\frac 12} \exp ( \| r^h (t) \|_{L^2 (h \Z)} ) + \| \al^h (t) \|_{L^2 (h \Z)} + \| \al^h (t) \|_{\dot{H}^1 (h \Z)} \\
    &\les h^{\frac 12} \exp ( 3 \| u_0 \|_{L^2 (\R)} ) + \| \gamma^h (t) \|_{L^2 (h \Z)} + \| \gamma^h (t) \|_{\dot{H}^1 (h \Z)} \\
    &\les \exp ( 3 \| u_0 \|_{L^2 (\R)} ) + \| \gamma_0^h \|_{H^1 (h \Z)} ,
\end{align*}

\noi
which gives the desired uniform bound for $r^h$ thanks to the assumption \eqref{u0cond}. This finishes the proof.
\end{proof}

It remains to show the long-wave limit for the Toda lattice in the original form \eqref{eqr2}.

\begin{proof}[Proof of Theorem~\ref{THM:long-wave}~\textup{(ii)}]
In view of \eqref{defgh} and \eqref{abr}, we define
\begin{align}
    \gamma^h (t, \ld) =
    \begin{cases}
        h^{-2} \exp \big( h^2 r^h (t, \frac{\ld}{2}) \big) - h^{-2} & \text{if } \frac{\ld}{h} \text{ is even} \vspace{5pt} \\
        - \frac 13 h^2 (\dh^+)^{-1} \dt r^h (t, \frac{\ld + h}{2}) & \text{if } \frac{\ld}{h} \text{ is odd}
    \end{cases}
\label{gh+}
\end{align}

\noi
for any $t \in \R$ and $\ld \in h \Z$. Then, $\gamma^h$ satisfies the equation \eqref{todah} with initial data $\gamma_0^h$ given by \eqref{g0h} satisfying the assumption \eqref{u0cond} and the condition \eqref{h_cond} with $h = h(\| u_0 \|_{L^2 (\R)}) > 0$ sufficiently small. Then, by Theorem~\ref{THM:lw}, for any $T > 0$, we have
\begin{align}
    \big\| \gamma^h (t, 2 \ld) - u (t, 2 \ld - 6 h^{-2} t) \big\|_{C_T L_\ld^2 (h \Z)} \leq C e^{C' T} h^{\frac 25}
\label{gh+est}
\end{align}

\noi
for some constants $C, C' > 0$ depending on the quantities in the assumption \eqref{u0cond} and also $\| u_0 \|_{H^1 (\R)}$. From \eqref{gh+}, Lemma~\ref{LEM:exp}, \eqref{gh+est}, and Proposition~\ref{PROP:L2bddr}, we obtain
\begin{align*}
    \big\| &r^h (t, \ld) - u (t, 2 \ld - 6 h^{-2} t) \big\|_{C_T L_\ld^2 (h \Z)} \\
    &\leq \Big\| h^{-2} \exp \big( h^2 r^h (t, \ld) \big) - h^{-2} - r^h (t, \ld) \Big\|_{C_T L_\ld^2 (h \Z)} \\
    &\quad + \big\| \gamma^h (t, 2\ld) - u (t, 2 \ld - 6 h^{-2} t) \big\|_{C_T L_\ld^2 (h \Z)} \\
    &\leq h^{\frac 12} \exp ( \| r^h \|_{C_T L_\ld^2 (h \Z)} ) + C e^{C' T} h^{\frac 25} \\
    &\leq 2 C e^{C' T} h^{\frac 25},
\end{align*}

\noi
which gives the desired estimate.
\end{proof}

\begin{ackno} \rm
The authors would like to thank Monica Vi\c{s}an for a helpful discussion. The authors acknowledge the use of AI tools for improvement of the presentation and for correction of typographical errors. All mathematical arguments were verified and written entirely by the authors. 
H.K. and R.L. were supported by the DFG through the Hausdorff Center for Mathematics under Germany's Excellence Strategy - GZ 2047/1, Project-ID 390685813, SFB 1060, Project-ID 211504053  and SFB 1720, Project-ID 539309657.
\end{ackno}

\end{document}